\documentclass[11pt]{article}
\usepackage{graphicx}
\usepackage{latexsym}
\usepackage{amsmath,amssymb,amsfonts}
\usepackage{tipa}
\usepackage{wrapfig}

\usepackage{color, colortbl}
\newtheorem{Theorem}{Theorem}

\newtheorem{Definition}{Definition}

\newtheorem{Lemma}{Lemma}

\newtheorem{Corollary}[Lemma]{Corollary}

 \newcommand{\lab}[1]{\label{#1}}                

\newenvironment{proof}{\noindent\textbf{Proof: }\ignorespaces}
  {\hspace*{\fill}$\Box$\medskip}

\definecolor{Gray}{gray}{0.9}
\newcommand{\MAX}{{\mbox{\rm MAX}}}
\newcommand{\MAXS}{\MAX(S)}
\newcommand{\MAXSN}{\MAX(S_n)}
\newcommand{\bfd}{{\mathbf D}}

\newcommand\Ball[1] {{\mathbf B}_{#1}}

\newcommand\Ballpq[2]{\Ball {#1} + \delta \Ball {#2}}
\newcommand\Ballpqd[3]{\Ball {#1} + {#3}\Ball {#2}}
\def \Ballp{{\Ball p}}

\def \BarBallp{{\bar{\mathbf B}}_{p}} 
\def \Area{{\rm Area}}

\newcommand{\EXP}[1]{\mathbf{E}\left[#1\right]}
\newcommand{\PR}[1]{\mathrm{Pr}\left[#1\right]}
\newcommand{\Strip}{\mathrm{Strip}}
\newcommand{\Striph}{\overline{\mathrm{Strip}}}
\newcommand{\Stripp}{\mathrm{Strip}'}
\newcommand{\Stripph}{\overline{\mathrm{Strip}'}}
\newcommand{\Wedge}{\mathrm{W}}

\newcommand{\MN}{{M_n}}
\newcommand{\EMN}{\EXP \MN}
\begin{document}
\title{Smoothed Analysis of the Expected Number of Maximal Points in Two Dimensions}

\author{Josep D{\'i}az%
         \thanks{\protect\raggedright
                Department of Computer Science,
                Universitat Polit{\`e}cnica de Catalunya. 
                Email: \texttt{diaz@cs.upc.edu}\,
                Supported by TIN2017-86727-C2-1-R.}
                \and
                Mordecai Golin%
                 \thanks{\protect\raggedright
                  Department of Computer Science ,
                 Hong Kong University of Science and Technology.
                Email: \texttt{golin@cse.ust.hk}}}

\date{}

\maketitle

\begin{abstract}
The {\em Maximal} points in a set $S$ are those that aren't {\em dominated} by any other point in $S.$  Such points  arise in multiple application settings in which they are  called by a variety of different  names, e.g.,  maxima,  Pareto optimums, skylines. Because of their ubiquity, there is a large literature on the {\em expected} number of maxima in a set $S$ of $n$ points chosen IID from some distribution.  Most such results assume that the underlying distribution is uniform over some spatial region and  strongly use this uniformity in their analysis.

This work was initially  motivated by the question  of how this expected number changes if the input distribution is perturbed by random noise.  More specifically, let $\Ball p$ denote the uniform distribution from the $2$-d  unit $L_p$ ball,  $\delta \Ball q $ denote the $2$-d  $L_q$-ball, of radius $\delta$ and  $\Ballpq p q$ be the convolution of the two distributions, i.e., a point $v\in \Ball p$ is reported with an error chosen from $\delta \Ball q.$ The question is how   the expected number of maxima change as a function of $\delta$.  Although the original motivation is for small $\delta$ the problem is well defined for any $\delta$ and our analysis treats the general case.

More specifically,  we study, as a function of $n,\delta$,  the expected number of maximal points   when the $n$ points in $S$ are chosen IID from  distributions of the type $\Ballpq p q$ where $p,q \in \{1,2,\infty\}$ for $\delta > 0$ and  also of the type $\Ballpq \infty  q$ where  $q \in [1,\infty)$ for $\delta > 0$.
\end{abstract}

\newpage
\section{Introduction}
\label{sec: Intro}

Let $S$ be a set of  $2$-dimensional points. The {\em largest}  points in $S$ are its  {\em maximal points}  of $S$ and are a well-studied object. More formally\footnote{We restrict our definition to  $S \subset \Re^2$  because that is what this paper addresses;  the concept of maxima generalize naturally to
$S \subset \Re^d$  for $d >2$ and have been well-studied there as well.  We discuss this in more detail in the Conclusions and Extensions section.}

\begin{Definition}
\label{def: Dom1}
For $u \in \Re$ let $u.x$ ($u.y$) denote the $x$ ($y$) coordinate of $u.$  
For $u,v \in \Re^2$,  $u$ is {\em dominated } by $v$ if $u \not=v$, $u.x \le  v.x$ and $u.y  \le v.y$.
If $S \subset \Re^2$ then
$$\MAXS = \{u \in S \,:\, \mbox{$u$ is not dominated by any point in $S \setminus \{u\}$}\}.$$
$\MAXS$ are the {\em maximal points}  of $S.$  
\end{Definition}

The problems of finding and estimating  the number of maximal points  of a set  in $\Re^2$, appear very often in many fields under  different denominations, {\em maximal vectors}, {\em skylines}, {\em Pareto frontier/points} and others, see, e.g., 
\cite{KungLP75,BorzsonyiKS01,Bentley78,Devroye93,geilen2007algebra}, and  for a more exhaustive history of the problems and further references, Sections 1 and 2 in \cite{ChenHT12}.

\begin{figure}[t]
\label{MAX}  
\centerline{
\includegraphics[width = 3.3in]{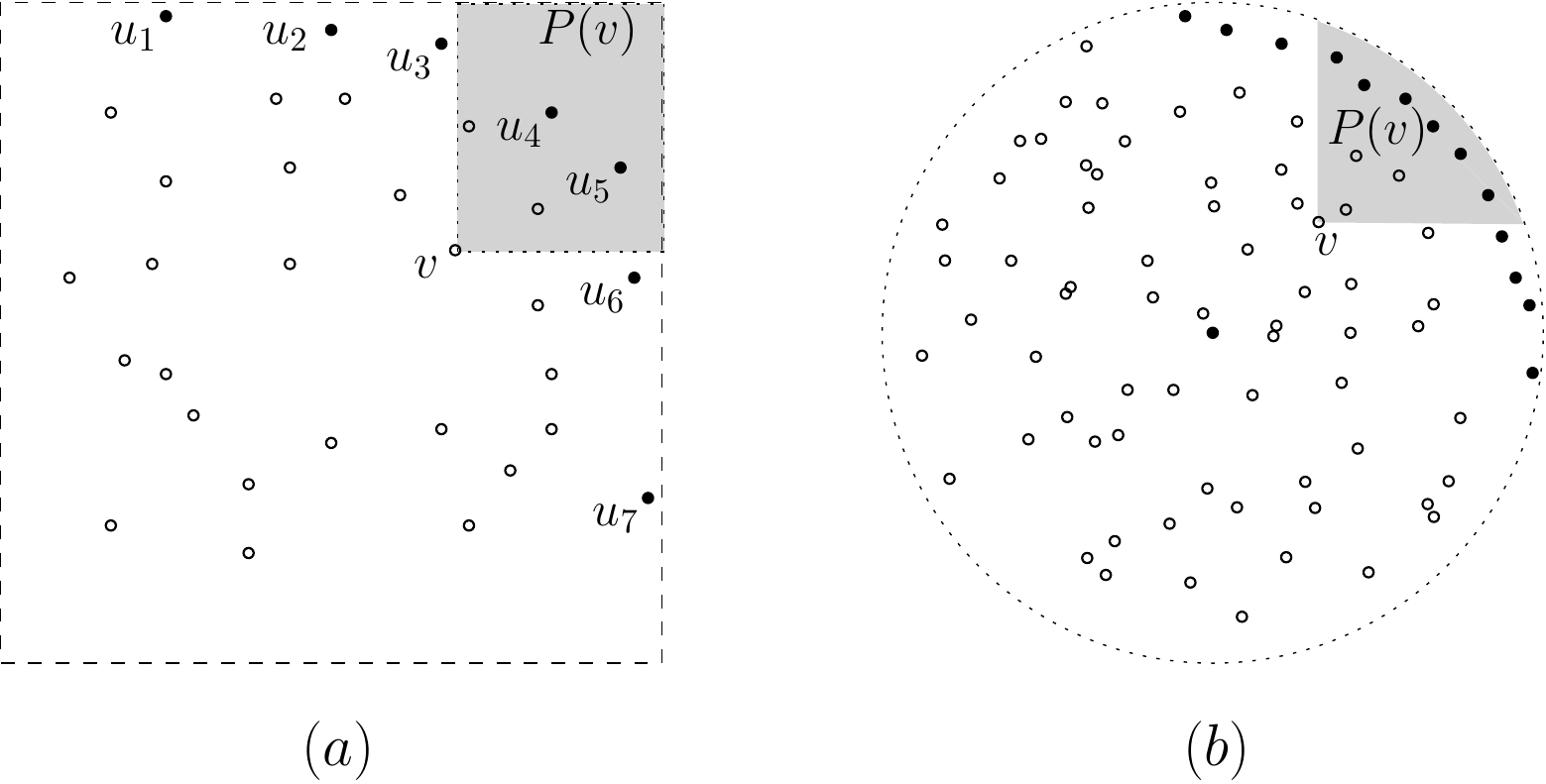}}
\caption{  The diagram shows $\MAXSN$ for two point sets $S_n$.  In both (a) and (b) the circles -- both empty and filled -- denote the points in $S_n$  and the filled circles are $\MAXSN$.  If the points are considered as being drawn from region $D$   $P(v)$ denotes the region in $D$ that dominates $v$.  In (a) $D$ is the dotted square and in (b)  $D$ is the dotted circle.}
\end{figure}

Recall that the $L_p$ metric for points $u,v$ in the $d$-dimensional space is defined by 
$$
||u-v||_p = 
\left\{
\begin{array}{lll}
\left(\sum_{i=1}^d |u_i - v_i|^p\right)^{1/p}	& \mbox{for real $p \ge 0$},\\[0.1in]
\max_i\left(|u_i - v_i|\right)					& \mbox {if $p = \infty$}.
\end{array}	
\right.
$$

Let $S_n$ denote a set of $n$ points chosen  Independently Identically Distributed (IID)  from some 2-D  distribution  $\bfd$ and  $\MN = |\MAXSN|$   be  the  random variable counting the number of maximal points  in $S_n$.
Because maxima are so ubiquitous,  understanding the expected number of maxima has been important in many areas and many  properties of $\MN$ have been studied.

More specifically, if $\bfd$ is the uniform distribution  drawn from an $L_p$ ball with $p \ge 1$ then, it is well known \cite{Devroye93,BaiDHT05,Dwyer1990,Buchta1989}  that
\begin{itemize}
\item  If $p = \infty$, then $\EXP{M_n}= H_n\sim \ln n.$ \\
The same result holds if the points are drawn from some distribution $\bfd = ({\bf X}, {\bf Y})$ where ${\bf X}$ and ${\bf Y}$ are ANY two 1-dimensional distributions that are independent of each other.
\item  If $p >1$, then
$$\lim_{n \to \infty}\frac{\EXP{M_n}}{\sqrt{n}} = C_p$$
where $C_p$ is a constant dependent only upon $p$. 
\item Similar results to the above,  i.e., that $\EMN=O\sqrt n)$, derived using similar  techniques,   are known if $\bfd$ is a {\em uniform} distribution from ANY convex region \cite{devroye1986lecture}.
\end{itemize}
It is also known \cite{Ivanin1975,Ivanin1975a} that if the $n$ points are chosen IID from a $2$-D Gaussian distribution then $\EXP{M_n}\sim \ln n.$ There are also generalizations of these results  (both the $\Ballp$ ones and the Gaussian one) to higher dimensions. See \cite{Dwyer1990} for a  a table containing most known results.

Surprisingly, given the importance of the problem, not much else is known.  The motivation for this work is to extend the family of distributions for which $\EMN$ can be derived.

Consider  a point $u$  that  is originally generated from some uniform distribution over a unit  $L_p$ ball but, has some $\delta$ error in the $L_q$ metric when measured or reported.   The actual reported point can be equivalently considered  as being chosen from a new distribution which we denote by  $\Ballpq p q$ (the next section provides formal definitions).
 Note that the support of this distribution is the Minkowksi sum of the two balls.

\medskip

\begin{wrapfigure}{r}{0pt}
\includegraphics[width=3.7cm]{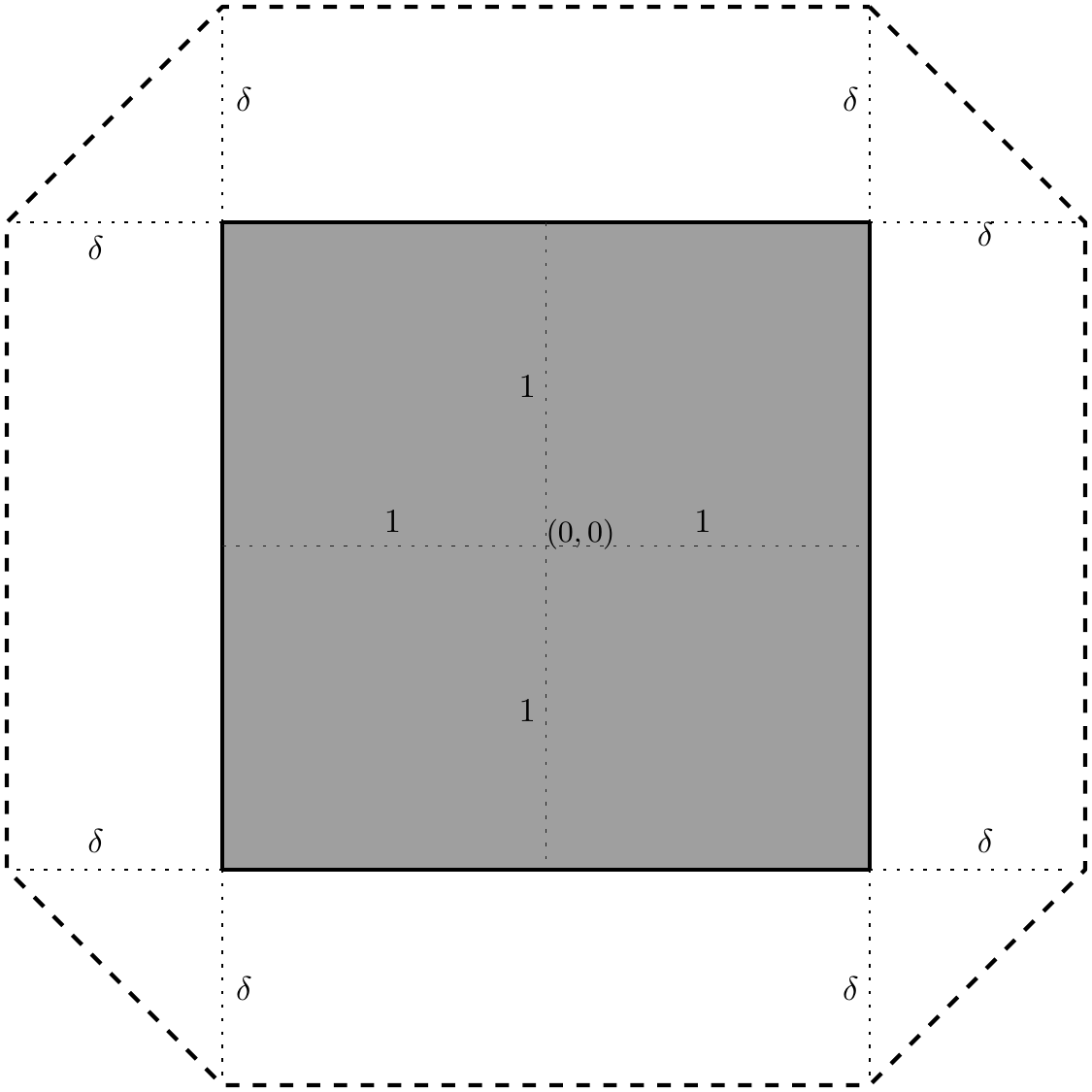}
\caption{$\Ball \infty + \delta \Ball 1$.}
\label{B1 and B inf}
\end{wrapfigure}
As an  example,   Figure \ref{B1 and B inf}   shows the support of $\Ball \infty + \delta \Ball 1$.  In the diagram, the shaded inner square  is the unit $\infty$-ball. A point $u_1$ chosen from that square is  then perturbed by the addition of another point $u_2$ , drawn uniformly from the $B_1$ ball with  radius $\delta$.
The support of this convoluted distribution  is the interior of the dotted region in the figure.

Note that the distribution is 
 NOT uniform in this support.  Towards the centre  the density is uniform but  it decreases approaching the  boundary of the support where it becomes zero.  Note too that the rate of decrease differs in different parts of the support. It is this non-uniformity that will cause complications in calculating $\EXP{M_n}.$

Although the problem described above was for small $\delta$ it is well defined for all $\delta  > 0$ which is what we analyze in this paper.
More specifically,  the  motivation for  the present work is twofold:
\begin{itemize}
\item  Explain  how  $\EXP{M_n}$ changes when the distribution is perturbed and 
\item  Increase   the families of distributions for which  $\EXP{M_n}$  is understood. 
\end{itemize}

The idea of analyzing how quantities  change under perturbations is  {\em smoothed analysis} \cite{spielman2004smoothed,spielman2009smoothed}. In the classic setting, smoothed  analysis of the number of maxima would mean analyzing how, given a {\em fixed} set $S_n,$   $\EXP{M_n}$ would change under small perturbations (as a function of the original set $S_n$). This was the approach in \cite{Damerow2004,Damerow2006}  (see also similar work for convex hulls in \cite{Devillers2016}). This paper differs in that it is the {\em Distribution} that is being smoothed (or convoluted) and not the point set. This paper also differs from recent work \cite{suri2013most,agrawal2017most} on the {\em most-likely} skyline and convex hull problems in that  those papers assume that each point has  a given probability distribution and they are attempting  to find the subset of points that has the highest probability of being the skyline (or convex hull).
\section{Definitions and Results}\label{sec: DefsRes}

\begin{Definition}
$p \in [1,\infty)$  or $p \ge 1$ will denote that $p$ is  a real number $\ge 1.$ $p \in [1,\infty]$ will denote that $p\in [1,\infty)$ OR  $p = \infty.$
\end{Definition}

\begin{Definition}
\label{def: Dist1}
Let $\bfd$
 be a distribution over $\Re^2$.
\begin{itemize}
\item 
If $\delta  \ge 0$, the distribution $\delta \bfd $ is generated by choosing a point $u$ using $\bfd$ and then returning the point $(\delta \cdot u.x, \delta \cdot  u.y).$

\item Let $\bfd_1,  \bfd_2$ be two distributions over $\Re^2$. $\bfd_1 + \bfd_2$ is the {\em convolution}of $\bfd_1, \bfd_1.$  It is generated by choosing a point $u_1$ from
$\bfd_1$ and a point  $u_2$ from  $\bfd_2$ and returning $(u_1.x+u_2.x, u_1.y+u_2.y).$

\item     A  set of  $S_n = \{u_1,\ldots ,u_n\}$ is {\em Chosen from  $\bfd$ } if the $u_i$ are IID with each $u$ being generated using distribution $\bfd$.
\end{itemize}
\end{Definition}

\begin{Definition}
\label{def:Mink}
Let $D \subseteq \Re^2$ be a  set and $ \delta \ge 0.$ \\ Then
$\delta D = \{(\delta \cdot u.x,\,  \delta \cdot  u.y)\,:\, u \in D\}$.

The {\em Minkowski sum} of sets $D_1$ and $D_2$ is
$$D_1 + D_2 = \{ (u_1.x+u_2.x, u_1.y+u_2.y) \,:\, u_1 \in D_1,\, u_2 \in D_2 \}.$$
If $u \in  \Re^2$, let $u + D$ will  denote  the set $\{u\}+D.$
\end{Definition}

\begin{Definition}
\label{def: Dist1}
Let $u \in \Re^2$, $r > 0$ and $p \in [1,\infty],$  
\begin{itemize}
\item The $L_p$ ball of radius $r$ around $u$ is
$$B_p(u,r) =\left \{(x,y) \,:\,  |x-u.x|^p + |y-u.y|^p \le r^p\right\}.$$
\item The $L_\infty$  ball of radius $r$ around $u$ is
$$B_\infty(u,r) = \left\{(x,y) \,:\,  \max (|x-u.x|,\, |y-u.y|) \le r \right\}.$$
\end{itemize}
 Let  $B_p =B_p((0,0),1)$ and $B_\infty =B_\infty((0,0),1)$ denote the respective unit balls and  
$a_p = \Area(B_p)$, $a_\infty = \Area(B_\infty)=2$ denote their respective areas.

\begin{itemize}
\item
For all $p\in [1,\infty],$ $\Ball p$ will denote the uniform distribution that selects a point uniformly from $B_p$. This  distribution has support $B_p$ with uniform density $1/a_p$ within $B_p.$

\item 
$\Ballpq p q$  will be the convolution of distributions $\Ball p$ and $\delta \Ball q$.  This distribution's support is the Minkowski sum
$B_p + \delta B_q$.
Note that the density of $\Ballpq p q$ is NOT uniform in $B_p + \delta B_q$.
\end{itemize}
\end{Definition}

The main result of this paper is 
\begin{Theorem}
\label{thm: main}
Fix $p,q$ so that either $p,q \in \{1,2,\infty\}$ or $p= \infty$ and $q \ge 1.$ Let $S_n$ be $n$ points chosen from the distribution $\Ballpq p q$ and 
$\MN = |\MAXSN|$. Let $\delta \ge 0$ be a function of $n$.  
Then $\EMN$ behaves as below.
$$
\begin{array}{||l|c||l||l|l||l||l||}
\hline \hline
 & (a) & (b) & (c) & (d) & (e) & (f) \\
\hline \hline
\rowcolor{Gray}
 &{\mathbf D=}   & \multicolumn{4}{c||}{ 0 \le \delta }            & \delta=1 \\ \hline \hline
(i) & \Ballpq \infty \infty  & \multicolumn{4}{c||}{\Theta\left( \ln n\right)}               & \Theta\left( \ln n\right)  \\ \hline
\multicolumn{7}{|l||}{}                        \\ \hline\hline 
\rowcolor{Gray}
&    & \delta \le \frac 1 {\sqrt n}    &      \frac 1 {\sqrt n}  \le \delta \le 1   & 1\le \delta \le   {\sqrt n}       &  \sqrt n \le \delta   &  \\ \hline \hline
(ii) & \Ballpq 1 1  &  \Theta \left( \sqrt n \right)    &      \Theta \left(  \frac {n^{1/3}}  { \delta^{1/3}}  \right)      &    \Theta \left(   \delta^{1/3}  n^{1/3} \right)       &   \Theta \left( \sqrt n \right)   &  \Theta\left(n^{1/3}\right) \\ \hline
(iii) &\Ballpq 2 2  &    \Theta \left( \sqrt n \right)   &      \Theta \left(    \frac {n^{2/7}} { \delta^{3/7}} \right)         &        \Theta \left( { \delta^{3/7}}   {n^{2/7}} \right)      &  \Theta \left( \sqrt n \right)    &   \Theta\left(n^{2/7}\right)  \\ \hline
\multicolumn{7}{||l||}{}                        \\ \hline
\rowcolor{Gray}
 & & \delta \le \frac 1 {\sqrt n}      & \multicolumn{2}{c|}{\frac 1 {\sqrt n}  \le \delta \le  \sqrt n} & \sqrt n \le \delta    &  \\ \hline
(iv) &\Ballpq \infty q&    \Theta\left( \ln n\right)  & \multicolumn{2}{c|}{\Theta \left( \ln n + \sqrt  \delta n^{1/4} \right)  }  &   \Theta \left( \sqrt n \right)    &  \Theta \left(n^{1/4} \right)  \\ \hline
\multicolumn{7}{|l|}{}                        \\ \hline \hline
\rowcolor{Gray}
&   &  \delta \le \frac 1 {\sqrt n}    &  \frac 1 {\sqrt n}  \le \delta \le n^{1/26}          &  n^{1/26}    \le \delta \le \sqrt n        & \sqrt n \le \delta    &  \\ \hline
(v) &\Ballpq 1 2 &   \Theta \left( \sqrt n \right)    &      \Theta \left(    \frac {n^{2/7}} { \delta^{3/7}} \right)       &     \Theta \left( \sqrt  \delta n^{1/4} \right)       &   \Theta \left( \sqrt n \right)    &   \Theta\left(n^{2/7}\right) \\ \hline
\end{array}
$$
\end{Theorem}

{\em \small Observations: In $\Ballpq p q$
\begin{itemize}
\item When $p = q = \infty$, $\MN$ has exactly the same distribution as if $S_n$  were chosen from $\Ball \infty$ so this is an uninteresting case.
\item When $\delta$ is small enough $(\le 1/\sqrt n)$,  $\EMN$ behaves  almost as  if $S_n$  were chosen from $\Ball p$ and when
$\delta$ is large enough $(\ge \sqrt n)$   it behaves  almost as if $S_n$  were chosen from $\Ball q.$ 
\item Later Lemma \ref{lem: scaling} will show that $\MN$ has the same distribution for $S_n$ chosen from both
$\Ballpq p q $ and $\Ball q + \frac 1 \delta \Ball p.$   Thus row (iv) gives the behavior for $\Ballpq q \infty$ for any $q \ge 1$ and row (v) the behavior for $\Ballpq 2 1.$
\item  When  $p =q \in \{1,2\}$ the behavior starts at 
$\Theta(\sqrt n)$, 
smoothly decreases until reaching $\delta =1$ and then increases again until reaching $\Theta(\sqrt n)$. The behavior in the middle is different for $1$ and $2.$ 
In both cases there is symmetry between $\delta$ and $1/\delta$ (from Lemma \ref{lem: scaling}).
\item When $p=1, q=2$ there is no symmetry.  Behavior starts at $\Theta(\sqrt n)$, decreases to 
$\Theta\left(n^{7/26}\right)$ at $\delta = n^{1/26}$ and then increases again at a different rate to $\Theta(\sqrt n)$.
\item  When $p= \infty$, the behavior is asymptotically equivalent  for all $q \in[1,\infty),$ not just $q=1,2.$  The only difference is in the value  of the constant hidden by the $\Theta.$ The behavior starts at $\Theta(\ln n)$, stays there for a short while and then smoothly increases to $\Theta(\sqrt n).$
\end{itemize}
}

\begin{figure}[p]
\centerline{
\includegraphics[width = 4.8in]{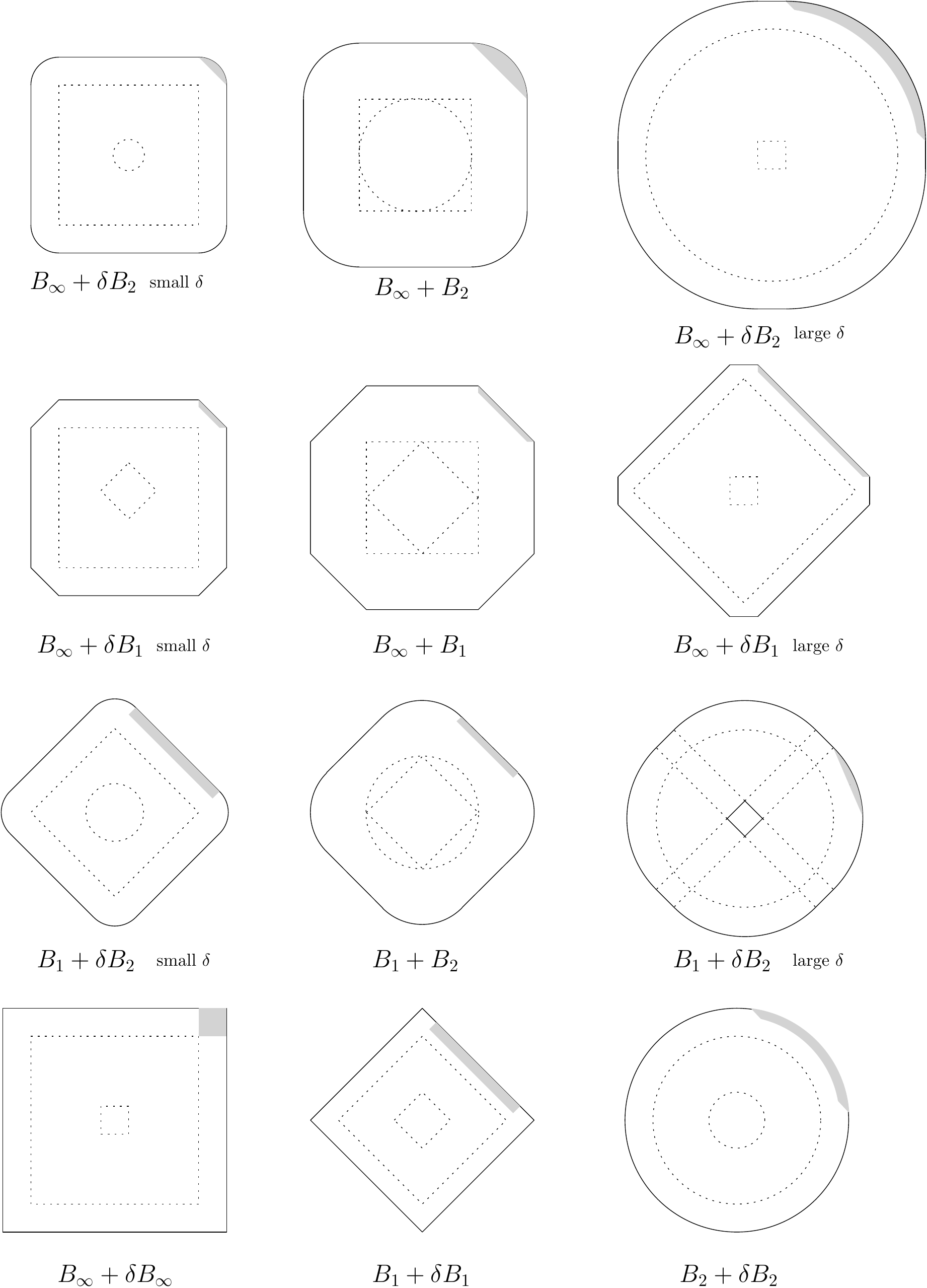}}
\caption{  Illustrations of the supports of some of the different distributions in the form  $\Ballpq p q$ examined in Theorem \ref{thm: main}. The dotted lines  denote the $B_p$ and $\delta B_q$ balls centered at $0.$  Note that in all cases the  density is uniform near the center of the support but then decreases to $0$ as the boundary is approached.  The gray areas denote, approximately, where the maxima of $S_n$ are concentrated.}
\end{figure}

\section{Basic Lemmas}
\label{sec: Basic Lemmas}
The following collection of Lemmas comprise the basic toolkit used to derive Theorem \ref {thm: main}.  They are only stated here, with complete proofs being provided in  Section  \ref {sec: Basic Lemma Proofs}.
\begin{Definition}
\label{def: den}
Let $\bfd$
 be a distribution over $\Re^2$,  $x \in \Re^2$ and $A \subset \Re^2$ a measurable region. 
 \begin{itemize}
 \item $f_\bfd(x)$ will denote the density function of $\bfd.$
 \item $\mu_\bfd(A) = \int_{A} f_\bfd (x) dx$ will denote  the {\em measure} of $A.$
 \end{itemize}
 If $\bfd$ is understood we often simply write $f(x)$ and $\mu(A).$
\end{Definition}

\begin{Definition}
\label{def: PDdef}
Let  $D \subseteq \Re^2$, $v \in D$ and $A \subseteq D.$
\begin{itemize}
\item 
$P(v) = \{u \in D \,:\, \mbox{$u$ dominates $v$}\} \cup \{v\}$.
\item $P(A) = \bigcup_{v \in A} P(v).$
\item $A$ is dominant in $D$  or a dominant region in $D,$ if $P(A) = A.$
\end{itemize}
Note that, by definition, $\forall v \in D,$  $P(v)$  is  a dominant region in $D.$
\end{Definition}

\begin{Lemma}\label{lem: basic mu} Let $v$  and $S_n$ be chosen from $\bfd$ and $A \subseteq D.$ Then
\begin{eqnarray*}
(a) \quad \Pr(v \in A) &=& \mu(A).\\
(b) \quad \EXP {|A \cap S_n|} &=& n \mu(A).\\
(c) \quad  \Pr(A \cap S_n = \emptyset) &=&  \left(1 - \mu(A)\right)^n.
\end{eqnarray*}
\end{Lemma}

The following observation will be used to prove most of our lower bounds.
\begin{Lemma}[Lower Bound] \label{lem: lb}
 Let $S_n$ be chosen from $\bfd$. Further let  
 $A_1,A_2,\ldots,A_m$ be a collection of   pairwise disjoint dominant regions in $D$ with $\mu(A_i) = \Omega(1/n)$ for all $i$.
 Then 
 $$\EMN  \ge  \EXP{ \left|  |\MAX\left(  S_n \cap \bigcup_{i=1}^m A_i\right)  \right|} = \Omega(m).$$
\end{Lemma}

\begin{Definition}
\label{def:Pprime}
Let $D = B_p + \delta B_q.$  For $v \in D$ define
$$P'(v) = B_q(v,\delta) \cap B_p =(v+ \delta B_q) \cap B_p,$$
the preimage of point $v$ in $B_p.$
\end{Definition}

\begin{Lemma}
\label{lem: measure integral}
Fix $p,q\in [1,\infty]$.  Let $\bfd = \Ballpq p q$ and let $v$ be a point chosen from $\bfd.$ 
Let $A \subseteq \Re^2$.  Then
\begin{eqnarray}
f(v) &=& \frac 1 {a_p a_q} \frac  {\Area(\{u \in B_p \,:\, v-u \in \delta B_q\})} {\delta^2}  \nonumber\\
&=&  \frac 1 {a_p a_q} \frac  {\Area(P'v)} {\delta^2}  \label{eq: fdef}\\
\mu(A)  &=&  \frac 1 {a_p a_q}  \int_{u \in B_p} \frac  {\Area((u + \delta B_q) \cap A)} {\delta^2} du. \label{eq: mudef}
\end{eqnarray}
\end{Lemma}

\begin{Lemma}
\lab{lem:easy mu}
Fix $p,q\in [1,\infty]$.  Let $\bfd = \Ballpq p q$ and
$\kappa>0$  be any  constant.  
$$
\begin{array}{lllcl}
(a)\  & \mbox{If } & v \in  D & \ \ \Rightarrow \ \ & f(v) = O(1).\\
(b)\  & \mbox{If } & v \in B_p  \mbox{\rm \ \ and }  \delta \le \kappa & \ \ \Rightarrow \ \  & f(v) = \Theta(1).\\
(c)\  & \mbox{If } & A \subseteq D  & \ \ \Rightarrow \ \  &\mu(A) = O(\Area(A)).\\
(d)\  & \mbox{If } & A \subseteq B_p  \mbox{\rm \ \ and }  \delta \le \kappa  & \ \ \Rightarrow \ \  &\mu(A) = \Theta(\Area(A)).\\
\end{array}
$$
The constants implicit in the $O()$  in (a) and (c) are only dependent upon  $p,q$ while the constants implicit in
the $\Theta( )$ in (b) and (d) are only dependent upon  $p,q,\kappa.$
\end{Lemma}

\begin{Lemma}[Mirror]
\label{lem: Mirror}
Let $\bfd$ be any distribution with a continuous density function $f(u)$ and $S_n$ a set of points chosen from $\bfd$.  Let $A, B$ be two disjoint regions in the support $D$ that are parameterized by $t \in [0,T]$ and satisfy:
\begin{enumerate}
\item 
$A(T) =A;$  $B(T) = B$.
\item (Monotonicity in $t$) $\forall t_1  < t_2$, $A(t_1)  \subseteq A(t_2)$ and $B(t_1)  \subseteq B(t_2)$.
\item (Asymptotic dominance in measure)  $\forall t,\,  \mu(B(t) = O(\mu(A(t)).$
\end{enumerate}
Define the random variables
$$t'=
\left\{
\begin{array}{ll}
\min\{t \,:\,  A(t)\cap S_n\not=\emptyset \}, & \mbox{if $A \cap S_n\not=  \emptyset$},\\
T& \mbox{if $A \cap S_n=  \emptyset$}.
\end{array}
\right.
\quad\mbox{and}\quad
X=|S_n\cup B(t')|.
$$
Then
\begin{equation*}
\label{eq: Xtp}
\EXP{X(t')} = O(1).
\end{equation*}
\end{Lemma}
\begin{Lemma}[Sweep]
\label{lem: Sweep}
Let $\bfd$ be any distribution with a continuous density function $f(u)$ and $S_n$ a set of points chosen from $\bfd$.  Let $A, B$ be two disjoint regions in the support $D$ that are parameterized by $t \in [0,T]$, satisfy conditions 1-3 of Lemma \ref{lem: Mirror} and, in addition, satisfy
$$\forall t \in [0,T],\quad  \mbox {if } v \in B \setminus B(t)\, \mbox{ and } \,  u \in A(t) \quad 
\Rightarrow \quad u \mbox{ dominates } v.$$
Then 
\begin{equation*}
\label{eq: BO1}
\EXP {|S_n \cap \MAX(P)|} =  O(1).
\end{equation*}
\end{Lemma}

\begin{Corollary}
\label{cor: Quadrants}
Fix $p,q\in [1,\infty]$ and choose $S_n$ from  $\bfd= \Ballpq p q.$
Let $Q_1$ be the upper-right quadrant of the plane and $O_{1}$ the first octant , i.e., 
$$Q_1 = \{u \in \Re^2 \,:\,  0 \le u.x,\,  0 \le u.y \}
\quad\mbox{and}\quad
O_{1} = \{u \in \Re^2 \,:\,   0 \le u.y \le u.x\}
$$
Then 
\begin{eqnarray}
\EMN = \EXP {|\MAX(S_n)|} &=&   \EXP {|Q_1 \cap \MAX(S_n)|}  + O(1) \label{eq:sym1}\\
			&=& \Theta\Bigl( \EXP {|O_1 \cap \MAX(S_n)|}\Bigr).\label{eq:sym2}
			\end{eqnarray}
\end{Corollary}
\begin{proof}
Set 
$$A = D \cap \{u \in \Re^2\,:\, u.y \ge 0\}, \quad    B = D \cap \{u \in \Re^2\,:\, u.y < 0\}.$$
For $ t \in [0,2]$ set 
$$A(t)  = \{u \in A\,:\, u.x \ge 1-t\}, \quad    B(t) = \{u \in B\,:\, u.x \ge  1-t\}.$$
Conditions (1) and (2) of Lemma \ref{lem: Mirror}
trivially  hold.  Condition (3) holds because, by symmetry around the $x$-axis
$\mu(B(t)) = \mu(A(t)).$ Finally the additional condition of Lemma \ref{lem: Sweep} holds because every point in $B \setminus B(t)$ is below and to the left of every point in $A(t)$.  Thus the expected number of maximal points  in $S_n$ {\em below} the $x$-axis is $O(1)$.  Note that this is {\em independent} of $n$.  Similarly, the expected number of maximal points to the left of the $y$-axis is $O(1)$.  This proves Eq.~\ref{eq:sym1}

To prove Eq.~\ref{eq:sym2} define the second octant to be
$$O_{2} = \{u \in \Re^2 \,:\,   0 \le u.x\le u.y\}.$$
By the symmetry between the $x$ and $y$ coordinates in the distribution,
$$\EXP{| O_1 \cap \MAX(S_n)|} = \EXP{| O_2 \cap \MAX(S_n)|}.$$
Futhermore, since $O_1$ and $O_2$ partition $Q_1$,
\begin{eqnarray*}
\EXP{|Q_1 \cap  \MAX(S_n)|} &=&\EXP{| O_1 \cap \MAX(S_n)|} + \EXP{| O_2 \cap \MAX(S_n)|} \\
&=& 2 \EXP{|  O_1 \cap \MAX(S_n)|}.
\end{eqnarray*}
Thus
$$\EMN = \EXP{ Q_1 \cap | \MAX(S_n)|}  + O(1)
= \Theta
\left(
 \EXP{| O_1 \cap \MAX(S_n)|}
 \right).$$
\end{proof}

The fact that for $\delta >0$, $u$ dominates $v$ if and only if $\delta u$ dominates $\delta v$  implies
\begin{Lemma}[Scaling]
\label{lem: scaling}
Fix $p,q \in [1,\infty]$, $\bfd = \Ballpq p q$ and $\bfd' = \Ballpqd q p {\frac 1 \delta}.$ Let $S_n$ be $n$ points chosen from $\bfd$ and $S'_n$ $n$ points chosen from $\bfd'$.  Then $\MAX(S_n)$ and $\MAX(S'_n)$ have exactly the same distribution. In particular
$$\EXP{|\MAX(S_n)|} =\EXP{|\MAX(S'_n)|}.$$
\end{Lemma}

\begin{Lemma}[Limiting Behavior]
\label{lem: limiting}
Let $p\in [1,\infty]$, $q\in [1,\infty)$,   $\delta = O(1 / \sqrt n)$ and $S_n$ chosen from $\bfd = \Ballpq p q$. Then
$$
\EMN =
\left\{
\begin{array}{ll}
\Theta(\ln n) & \mbox{ if  $p =\infty$},\\
\Theta(\sqrt n) & \mbox{ if $p \not=\infty$}.
\end{array}
\right.
$$
\end{Lemma}

Note that if  $u$ chosen from $\Ball p$, $u.x$ and $u.y$  are  independent random variables.  Thus, for any $\delta >0$ if $v$ is chosen from $\bfd=\Ballpq \infty \infty$ $v.x$ and $v.y$  are  independent random variables. As noted in the introduction, this means that if $S_n$ is chosen from $\bfd$, $\EMN$ is exactly the same as if 
$S_n$ was just chosen from $\Ball \infty,$  i.e., $\EMN = \Theta(\ln n).$

Now note that Lemma \ref {lem: limiting}  combined with Lemma  \ref{lem: scaling} immediately imply the limiting behavior in columns (b) and (e) of the table in Theorem \ref{thm: main}. Note too that for rows (ii) and (iii),  column (d) follows directly from applying Lemma \ref{lem: scaling} to column (c).

Thus, proving Theorem \ref{thm: main} reduces to proving cells (ii) c,  (iii) c,   (iv) c,d  and  (v) c,d.  In the next sections we sketch how to derive these results with full proofs relegated to the appendix.
\section{The General Approach}
\label{sec:  General Approach}
\subsection{ A Simple Example:  $\Ball 1.$}
\begin{figure}[t]
\centerline{
\includegraphics[width = 4.8in]{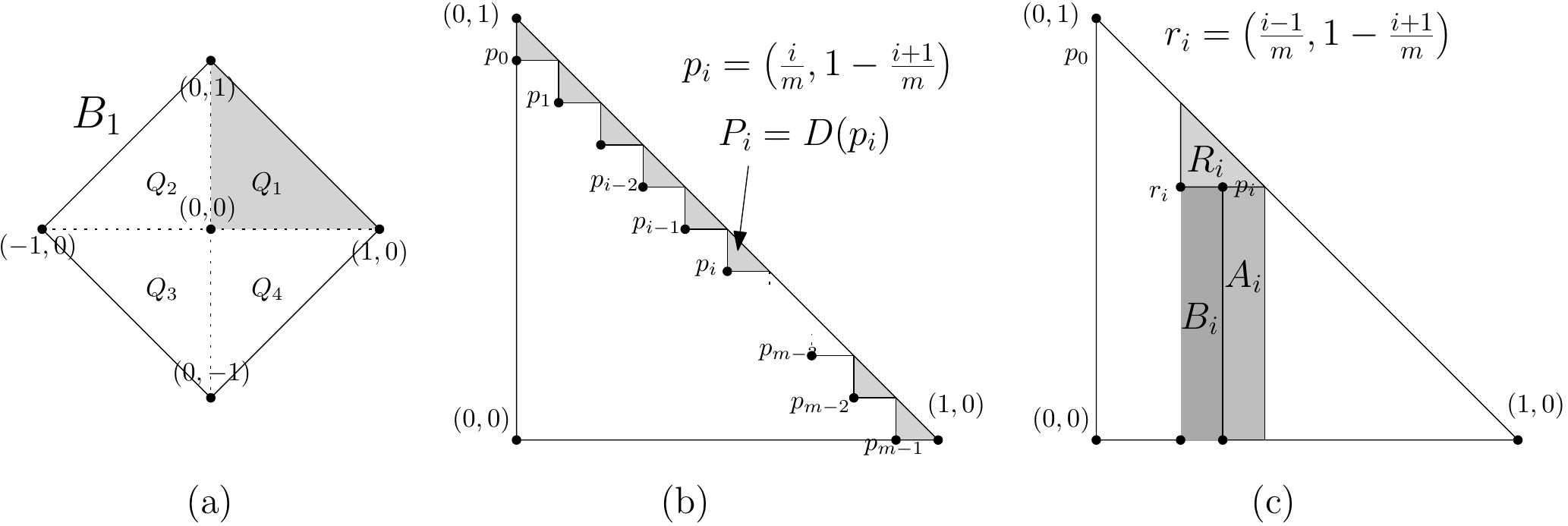}}
\caption{ Illustration of proof  $\EXP{\MN}= \Theta(\sqrt n)$ when $S_n$ is chosen from $\Ball 1$ All but $)(1)$ maxima will be in first quadrant $Q_1$;  (b) and (c) only show $Q_1$. (b) illustrates the lower bound and (c) the upper bound.}
\label{fig:B1full}
\end{figure}

 Before sketching our results it is instructive to see how the Lemmas in the previous section can be used to re-derive that fact that, if $\bfd = \Ball 1$ then $\EXP{\MN}= \Theta(\sqrt n)$.  This is illustrated in Figure \ref{fig:B1full}.   

Even though the behavior of  $\bfd = \Ball 1$ is already well understood we provide this to sketch the generic steps that are needed to derive $\EXP{\MN}$.  These are exactly the same steps that are needed when $\bfd = \Ballpq p q$ and permits identifying where the complications can arise in those more general cases.

Set $m = \lfloor \sqrt n\rfloor$ and let $p_i,r_i$ be the points defined in the figure with $P_i=P(p_i)$ and $R_i= P(r_i).$ Also set
$$B_i = \left\{(x,y)\,:\  \frac {i-1} m \le x \le  \frac i m, \ 0 \le y \le 1 - \frac {i+1} m\right\},
\quad
A_i= \left( \frac 1 m,0\right) + B_i.$$
Finally, for $0 \le t \le (1+i)/m$ set $B_i(t) = B_i \cap \{(x,y)\,:\, y \le (1+i)/m -t\}$ and  
$A_i(t)= \left( \frac 1 m,0\right) + B_i(t).$ The steps in the derivation are.

\begin{enumerate}
\item {\em Restricting to first Quadrant:}  Corollary \ref{cor: Quadrants} implies  that it is only necessary to analyze $ \EXP {|Q_1 \cap \MAX(S_n)|}$.
\item {\em Calculating Density and Measure:} Because $\bfd$ has a uniform density,  $\mu(A) = \Theta(\Area(A))$ for all regions $A \subseteq D.$
\item {\em Lower Bound:}  The $P_i$ are a collection of $t$ pairwise disjoint dominant regions with 
$\mu(P_i) = \Theta(\Area(P_i)) = \Theta(m^{-2}) =\Theta(1/n).$  Thus, from Lemma \ref{lem: lb},  
$\EMN  = \Omega(m) = \Omega(\sqrt n).$
\item {\em Upper bound:}  Note that $Q_1 \cap D = (\cup_i R_i)  \bigcup( \cup_i B_i)$ so
$$|Q_1 \cap \MAX(S_n)| \le \sum_i |R_i \cap \MAX(S_n)| + \sum_i |B_i \cap \MAX(S_n)|.$$
 Since $\mu(R_i) = \Theta(\Area(R_i)) = \Theta(1/n)$, Lemma \ref {lem: basic mu} implies that for all $i,$ 
 $$ \EXP{|R_i \cap \MAX(S_n)|}  \le  \EXP{|R_i \cap S_n)|}  =  n\Theta( \Area(R_i)) = \Theta(1).$$
 The crucial observation is that   for all $i$, the Sweep Lemma (Lemma \ref {lem: Sweep}) holds with  $B= B_i$ and $A=A_i$.  Thus $\EXP{|B_i \cap \MAX(S_n)|} = O(1).$ Combining the above completes the upper bound, showing that
 $$|Q_1 \cap \MAX(S_n)| \le \sum_i  \Theta(1) + \sum_i  \Theta(1) = \Theta(m) = \Theta(\sqrt n).$$
\end{enumerate}
\subsection{The General Approach}
 The proof of Theorem \ref {thm: main} will require case-by-case analyses of $\bfd = \Ballpq p q$ for different pairs $p,q$.  The analysis  for  each pair will follow exactly the same 5 steps as the analysis of $\Ball 1$ above.  We note where the complications arise.
 
 Step 1 of restricting the analysis to quadrant $Q_1$ will be the same for every case.
 
 Step 2, of deriving the measure, will often be quite cumbersome.  While Lemma \ref{lem: measure integral} provides an integral  formula this, in many cases, is unusable.  The density varies quite widely near the border of the support $B_p + \delta B_q$ which is where most of the maxima are located.  A substantial amount of work is involved in finding usable functional representations  for the densities/measures  in different parts of the support.
 
 Step 3, of deriving the  lower bound, is usually a simple application of  Lemma \ref{lem: lb}, given the results of step 2.
 
 Step 4 is the hardest step.  It is usually derived using the sweep lemma with the difficulties arising from how to specify the regions to be swept. This strongly depends upon {\em how} the measure is represented .
\section{Proofs of Basic Lemmas}\label{sec: Basic Lemma Proofs}

\begin{proof} of Lemma \ref {lem: lb}.

First note that, from Lemma \ref{lem: basic mu},
$\Pr(S_n \cap A_i = \emptyset) = \left(1 - \mu(A_i)\right)^n.$ Thus  $\mu(A_i) = \Omega(1/n)$ implies
$$ \Pr(S_n \cap A_i \not= \emptyset) =1 -\Pr(S_n \cap A_i = \emptyset) = \Omega(1).$$
If region $A$ is dominant then points in $A$ can only be dominated by other points in $A$
so  $A \cap \MAX(S_n) = \MAX(S_n \cap A).$  Since each $A_i$ is dominant, this implies
$$\EXP{| \MAX(S_n) \cap A_i|} \ge \Pr(S_n \cap A_i \not= \emptyset)  = \Omega(1).$$
Finally, since the $A_i$ are pairwise disjoint,
\begin{eqnarray*}
\EXP{|MAX(S_n)|} & \ge&   \EXP{\left|MAX(S_n) \cap \left(\bigcup_i A_i\right)\right|} \\
				&=& \EXP{\left|\bigcup_i \left(MAX(S_n) \cap A_i\right)\right|} \\
 		                  & =&  \sum_{i=1}^m \EXP{\left|MAX(S_n) \cap  A_i\right|} 
				\ge \sum_{i=1}^m \Omega(1) = \Omega(m).
\end{eqnarray*}
\end{proof}

\begin{proof} of Lemma \ref{lem: measure integral}:

Note that for $u \in B+p$, $f_{\Ball p}(u) = \frac 1 {a_p}$ and for $u' \in \delta B_q$, $f_{\delta \Ball q}(u') = \frac 1 {a_q \delta^2}$. 

To see Eq.~\ref{eq: mudef} note that
\begin{eqnarray*}
\mu(A) &=&  \int_{u \in \Ball p } \left(  \int_{\substack{w \in \delta B_q \\  u + w  \in A}} f_{\delta \Ball q}(w)  dw\right) f_{\Ball p}(u) du \\[0.1in]
	     &=& \frac 1 {a_p a_q \delta^2 } \int_{u \in B_p } \left(  \int_{\substack{w \in \delta B_q \\  u + w  \in A}}  dw\right) du \\[0.1in]
	     &=& \frac 1 {a_p a_q} \int_{u \in B_p }  \frac{\Area\left((u + \delta B_q ) \cap A\right)}{\delta^2} \,du.
\end{eqnarray*}
For Eq.~\ref{eq: fdef} first note that 
\begin{eqnarray}
\mu(A) 
	&=& \frac 1 {a_p a_q \delta^2 } \int_{u \in B_p } \left(  \int_{\substack{w \in \delta B_q \\ u + w  \in A}}  dw\right) du \nonumber \\[0.1in]
		&=&   \frac 1 {a_p a_q \delta^2 } \int_{u \in B_p }  \left(\int_{\substack{v \in u + \delta B_q \\v    \in A}}  dv \right) du \lab{eq:MDP1}\\[0.1in]
	&=&  \frac 1 {a_p a_q \delta^2 } \int_{v \in A }  \left(\int_{\substack{u \in v - \delta B_q \\u    \in \delta B_p}}  du \right) dv  \nonumber \\ 
	&=&  \frac 1 {a_p a_q} \int_{v \in A }  \frac {\Area\left\{ u \in B_p \,:\, v-u \in \delta B_q\right\}} {\delta^2}  dv \nonumber
\end{eqnarray}
where (\ref{eq:MDP1}) comes from the change of variables  $v = u +w$. 
Differentiating around $v$ yields  Eq.~\ref{eq: fdef}.
\end{proof}

\begin{proof} of Lemma \ref {lem:easy mu}:

The proof  for (a) follows easily from the fact that, for all $u,$
$$\Area (B_p  \cap (u + \delta B_q))) \le  \Area(u + \delta B_q) = \Area(\delta B_q)  =  a_q \delta^2,$$
so from Eq.~\ref{eq: fdef}, 
$f(v) = O (1)$.
Furthermore, if   $u \in B_p$ then 
$$\Area (B_p \cap (u + \delta B_q))) \ge c  \Area(u + \delta B_q) = c a_q \delta^2$$
where $c$ is only dependent upon $p,q,\kappa$.  Thus, again  from Eq.~\ref{eq: fdef},
$f(v) = \Theta (1)$, proving (b).  The proofs for (c) and (d)  follow from plugging these inequalities into Eq.~\ref{eq: mudef}.
\end{proof}

\begin{proof} of Mirror Lemma (Lemma \ref {lem: Mirror}):

Without loss of generality smoothly  rescale $t$ so  that $\mu(A(t))=t$,  and thus $T=\mu (A)$.

The informal intuition  of the Lemma  is that since the ``first'' point in $A$ appears when  the sweep line is $t'$, $\mu(A(t')) ~\sim \frac 1 n.$  Since $B$ is asymptotically dominated in measure by $A,$  $\mu(B(t')) = O(1/n)$ and thus $\EXP{X(t')} = n \mu(B(t')) = O(1).$

Note that by the continuity of the measure we know that $\Pr(|S_n \cap A(t')| =1) =1$.  That is, we may assume that
$|D \setminus A(t')| = n-1.$  

Now assume that
 $t'$ is known.   {\em Conditioned on known $t'$,} the remaining $n-1$ points in $S_n$ are  chosen from   $D \setminus A(t')$ with the associated conditional distribution. More specifically, if  $u$ is any one of those $n-1$ points.  
$$\Pr\left(u \in B(t')\,\Bigm|\, t'\right) 
= \frac  {\mu(B(t'))} {\mu(D\setminus A(t'))} = \frac{\mu(B(t'))}{1-\mu (A(t'))}.$$ 
Thus, conditioning on $t'$, and applying  Lemma \ref{lem: basic mu} (b)
$$\EXP{X(t') \,\Bigm|\, t'} = (n-1)\frac{\mu(B(t'))}{1-\mu(A(t'))}$$
and therefore
$$\EXP{X(t')}=\EXP{ \EXP{X(t') \,\Bigm|\, t'}   } = \EXP{(n-1)\frac{\mu(B(t'))}{1-\mu(A(t'))}}.$$
From the definition of $t'$ and  Lemma \ref{lem: basic mu} (c),  $\mu(A(t')) > 1/2$ with exponentially low probability.
Therefore, recalling that $\mu(A(t)) = t,$ 
\begin{align*}
\EXP{X(t')} &=(n-1)    \EXP{O(\mu(B(t')))} \\
&=(n-1)\EXP{O(\mu (A(t')))}=(n-1)O\left(\EXP{t'}\right).
\end{align*}
Another application of  Lemma \ref{lem: basic mu} (c) shows
\begin{eqnarray*}
\EXP{t'} &=& \int_{\alpha=0}^T \Pr (t' \ge \alpha) d\alpha 
  =  \int_{\alpha=0}^T  \left(1 - \mu(A(\alpha)\right)^{n-1} d\alpha \\
&=&  \int_{\alpha=0}^T  e^{-(n-1) \mu(A(\alpha))} \left(1+ O((n-1) \mu^2(A(\alpha))) d \alpha \right)\\
  &=&  \int_{\alpha=0}^T  e^{-(n-1)\alpha} \left(1+ O((n-1) \alpha^2)\right)d \alpha   = O\left(\frac 1 {n-1}\right).
\end{eqnarray*}
Thus
$\EXP{X(t')} = O(1).$
\end{proof}

\begin{proof} of Sweep Lemma (Lemma \ref{lem: Sweep}):

From the setup in  Lemma \ref {lem: Mirror}, for all $t,$ all points in $B\setminus B(t)$ are dominated by all points in $A(t)$. By the definition of $t'$,  $A(t')\cap S_n$ contains (exactly) one point.
Thus no point in $(B\setminus B(t'))\cap S_n$  can be maximal, i.e., 
$$MAX(S_n) \cap (B\setminus B(t'))  = 0.$$
The proof follows from 
$$\EXP{|MAX(S_n) \cap B| } = \EXP{|MAX(S_n) \cap B(t')| } \le \EXP{|S_n \cap B(t')|} = O(1).$$
\end{proof}

\begin{proof} of Lemma \ref{lem: scaling}:

Let $S_n=\{u_1,\ldots,u_n\}$ 
  be  chosen from $\bfd$.
Recall that the process of choosing  point $u$ from $\bfd$ is to choose $w$ from $\Ball p$,  $v$ from $\Ball q$ and return  $u=w+ \delta v$.  Choosing a point $u'$ from $\bfd'$ is the same except that it returns 
$p'= v+ \frac 1 \delta w= \frac 1 \delta u_i$. 
Thus the distribution of choosing $S_n=\{u_1,\ldots,u_n\}$ from $\bfd$ is exactly the same as choosing
$S_n=\{\frac 1 \delta u_1,\ldots, \frac 1 \delta u_n\}$ from $\bfd'.$

Finally, note that dominance is invariant under multiplication by a scalar, i.e., $p_i$ dominates $p_j$ if and only if
$\frac 1 \delta u_i$ dominates $\frac 1 \delta u_j$.  Thus $|\MAX(S_n)|$ and $|\MAX(S'_n)|$ have  the same distribution and $\EXP{|\MAX(S_n)|}=\EXP{|\MAX(S'_n)|}.$
\end{proof}

The proof of Lemma  \ref{lem: limiting}  will need an observation that will be reused multiple times in the analysis of $\Ballpq \infty q$ and is therefore  stated first, in its own lemma.

\begin{Lemma}
\label{lem: side rec}
Recall $P(v)$ from Definition \ref{def: PDdef}.
Fix $q \in [1,\infty)$ and set $\bfd = \Ballpq  \infty q$.  Let $\Delta \in (2\delta,1/2).$  
Define regions (Fig.~\ref{fig: small delta}(a))
 \begin{eqnarray*}
T(\Delta) &=& P((z,0))  \cap P((0,z)),\\
 R_1(\Delta) &=& P((z,0)) \setminus T(\Delta),\\
R_2(\Delta) &=& P((0,z)) \setminus T(\Delta).
\end{eqnarray*} 
(a) Then  
\begin{equation}
\label{eq:Biqul1}
\EXP{|\MAX(S_n \cap  R_1(\Delta))|} =O(\ln n).
\end{equation}

\par\noindent (b) Furthermore,  if $ \mu(R_1(\Delta))  = \Omega\left(1/\sqrt n\right)$, then\footnote{$R \Bigm|  B$  denotes random variable $R$ {\em conditioned} upon event $B.$}
\begin{equation}
\label{eq:Biqul2}
\EXP{|\MAX(S_n \cap  R_1(\Delta))|  \Bigm| S_n \cap T(\Delta) = \emptyset }  =   \Omega ( \ln n).\\
\end{equation}

\par\noindent (c)
Parts (a) and (b) 
 remain correct if $R_1(\Delta)$ is replaced by$ R_2(\Delta)).$
\end{Lemma}
\begin{proof}
Suppose $u,v \in R_1(\Delta)$ satisfy  $u.x=v.x$.  Note that $v + \delta B_q$ is just 
$u + \delta B_q$  shifted vertically $v,x-u.x$.
Then, because $\Delta \in (2\delta,1/2),$
$$\Area((u+ \delta B_q) \cap B_\infty)  = \Area((v+ \delta B_q) \cap B_\infty).$$ 
(In Fig.~\ref{fig: small delta}(a) this is illustrated by noting that the two white cross-hatched areas are vertically shifted versions of each other.)

Thus, from 
Lemma \ref {lem: measure integral},  $f(u) = f(v).$  Equivalently,  this can be written as,
``if  $v \in R_1(\Delta)$ then  $f(v) = (g(v.x))$ for some function $g(x)$''.

The other observation  needed is that the distribution of $S_n \cap R_1(\Delta)$   is equivalent to the one generated by 
\begin{enumerate}
\item Choosing random variable $X$ from a binomial distribution 
$B\left(n, \mu(R_1(\Delta)\right).$
\item  For $v \in  R_1(\Delta)$, setting
$ \bar f(v) = \frac {f(v)} {\mu(R_1(\Delta))}$ to be  the probability density function for choosing a point $v$ from $\Ballpq \infty q,$ {\em conditioned on knowing} that $v \in R_1(\Delta).$  
\item  Choosing $X$ points (in $R_1(\Delta)$)   from the distribution defined by $\bar f(v).$
\end{enumerate}

In  particular,  point (2) implies  that the distribution on $R_1(\Delta)$ defined by $\bar f(v)$ 
is of the form  $\bar f(v) = \bar g(v.x)$, where $\bar g(v.x) = \frac {g(v.x)} {\mu(R_1(\Delta))}$ is only dependent upon $v.x$ and not $v.y.$
Thus $\bar f(v)$ denotes a distribution in which the $x$ and $y$ coordinates  are independent of each other.

 As stated in the introduction, the number of maxima for $X$ points chosen from such a distribution behaves exactly as if the points are chosen from $\Ball \infty$.  Thus, if $X$ points are chosen using $\bar f(u)$,  the expected number of maxima among them will be $\Theta(\ln X).$  Since $X \le n$
this immediately implies (a).

\medskip

To prove (b)  assume that the only information known is  that $S_n \cap T(\Delta)=\emptyset.$  Conditioned on this event, the new density for
$v \in (B_\infty + \delta B_q)\setminus T(\Delta)$ is
$\hat f(v) = \frac {f(v)} {1 - \mu(T(\Delta))} > f(v).$ In particular, this implies that
$\hat \mu(R_1(\Delta)) = \Omega(1/\sqrt n).$

Using the same argument as in the proof of (a), 
{\em  conditioned on $S_n \cap T(\Delta)=\emptyset$}, the distribution of 
$S_n \cap R_1(\Delta)$,   is equivalent to   the one generated by 
\begin{enumerate}
\item Choosing random variable $X$ from a binomial distribution 
$B\left(n, \hat \mu(R_1(\Delta)\right)).$
\item  For $v \in  R_1(\Delta)$, setting
$ \bar f(v) = \frac {\hat f(v)} {\hat \mu(R_1(\Delta))}$.   $\bar f(v)$ is the conditional  probability density function for choosing a point $v$ from $\Ballpq \infty q$ conditioned on knowing that $v \in R_1(\Delta)$ and $S_n \cap T(\Delta)=\emptyset$.
\item  Choosing $X$ points (in $R_1(\Delta)$)   from the distribution defined by $\bar f(p).$
\end{enumerate}

In  particular,  point (2) implies  that the distribution on $R_1(\Delta)$ defined by $\bar f(v)$ 
is of the form  $\bar f(v)  = \frac {g(v.x)} {(1 - \mu(T(\Delta))) \mu(R_1(\Delta))}$ is only dependent upon $v.x$ and not $v.y.$
Thus $\bar f(v)$ denotes a distribution in which the $x$ and $y$ coordinates  are independent of each other.
Again, as in the proof of (a), this implies that, if  if $X$ points are chosen using $\bar f(u)$,  the expected number of maxima among them will be $\Theta(\ln X).$  
Thus
$$
\EXP{|\MAX(S_n \cap  R_1(\Delta))|  \Bigm| S_n \cap T(\Delta) = \emptyset }  =   \Theta\left(\sum_{i=1}^n (\ln i)  \Pr(X = i)\right).
$$
Finally, recall that $X$ was drawn from a binomial $B(n,p)$ distribution, where 
$$p = \hat \mu( R_1(\Delta))> \mu( R_1(\Delta)) = \Omega(1/\sqrt n).$$
 But this immediately imples that $\Theta\left(\sum_{i=1}^n (\ln i)  \Pr(X = i)\right) = \Theta(\sqrt n)$, completing the proof of (b).

By symmetry, the  proof of (c)  is exactly the same as the proofs of (a) and (b).
 \end{proof}

\begin{proof} of Lemma \ref{lem: limiting}: See Fig.~\ref{fig: small delta}.

The proofs of the cases (a) $p \not = \infty$ and (b) $p= \infty$ are done separately. 

\begin{figure}[t]
\centerline{
\includegraphics[width = 4.8in]{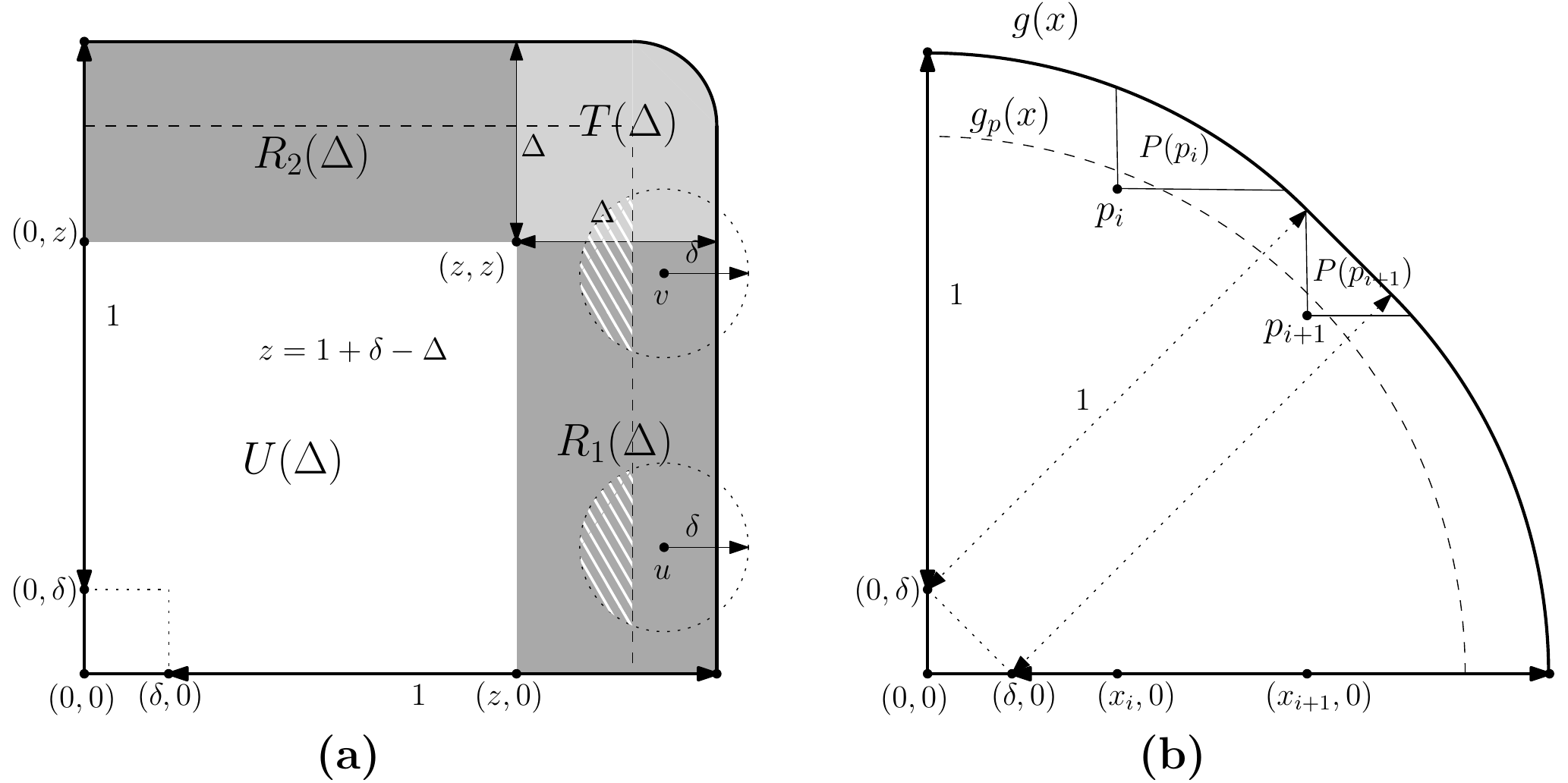}}
\caption{ Illustration of the proofs of Lemmas \ref{lem: limiting} and \ref {lem: side rec}. Note that only the support in the upper-right quadrant is shown. (a) illustrates 
the proof of Lemma \ref {lem: side rec} which is used in  the proof of  Lemma \ref{lem: limiting}, when  $p = \infty$.
 (b) illustrates the upper bound proof in Lemma \ref{lem: limiting} when $p\not = \infty$.
 Note that for clarity, diagrams are not drawn to school.  In (a),  $\delta = O(1/\sqrt n)$ and thus the part of $B_p + \delta B_q$ outside of $B_p$ is extremely thin.}
\label{fig: small delta}
\end{figure}

\medskip

\par\noindent\underline{(a) $p = \infty:$}  
Corollary \ref{cor: Quadrants} states that $\EMN = \EXP{|\MAX(S_n) \cap Q_1|} + O(1)$.  We therefore only need to analyze $ \EXP{|\MAX(S_n) \cap Q_1|}.$

Define  $T(\Delta),$  $R_1(\Delta)$ and $R_2(\Delta)$ as in  Lemma \ref {lem: side rec}.
Also set 
$$U(\Delta) = (B_\infty + \delta B_q) \setminus (T(\Delta) \cup R_1(\Delta) \cup R_2(\Delta)).$$
By decomposition and symmetry
\begin{eqnarray}\EXP{|\MAX(S_n) \cap Q_1|} &=& \EXP{| \MAX(S_n) \cap R_1(\Delta)|} + \EXP{ |\MAX(S_n) \cap R_2(\Delta)|}  \hspace*{.5in}  \ \nonumber\\
&& \quad +  \EXP{| \MAX(S_n) \cap T(\Delta)|} +  \EXP{| \MAX(S_n) \cap U(\Delta)|} \label{eq:Biqup4} \\
&=& 2 \EXP{| \MAX(S_n) \cap R_1(\Delta)|} \hspace*{.5in}  \  \nonumber \\
&&\quad +  \EXP{ |\MAX(S_n) \cap T(\Delta)|} +  \EXP{ \MAX(S_n) \cap U(\Delta)|}. \nonumber
\end{eqnarray}

\par\noindent\underline {\bf Upper Bound:}

Let $\Delta =  c  \sqrt \frac {\ln n} n +\delta$  for constant $c$ to be chosen later.  From Lemma \ref{lem:easy mu}(d)
$\mu(T(\Delta)) = \Theta\left( \frac {\ln n} n\right)$,  where the constants implicit in the $\Theta$ grow to $\infty$ as $c$ increases.

This immediately implies that 
\begin{equation}
\label{eq:Biqup1}
\EXP{| \MAX(S_n) \cap T(\Delta)|}  \le \EXP{|S_n \cap T(\Delta)|} = n \mu(T(\Delta))  = O(\ln n).
\end{equation}
From Lemma \ref{lem: basic mu} (c), for any $c'>0$ we can choose $c$ large enough so that
$$\Pr(S_n \cap T(\Delta)= \emptyset) = \left(1 - \mu(T(\Delta)) \right)^n = O(n^{-c'}).$$
Let $c'=1.$

Next note that  since any point in  $T(\Delta)$  dominates ALL points in $U(\Delta),$  $S_n \cap T(\Delta) \not = \emptyset$ implies  $|MAX(S_n) \cap U(\Delta)| =0.$
 Since $\MAX(S_n) \le n$,
 \begin{equation}
\label{eq:Biqup2}
\EXP{| \MAX(S_n) \cap U(\Delta)|}  \le  n \Pr(S_n \cap T(\Delta)= \emptyset)  = O(1).
\end{equation}

Finally, note that from Lemma \ref{lem: side rec}(a)
 \begin{equation}
\label{eq:Biqup3}|\MAX(S_n) \cap R_1(\Delta)| \le    |\MAX(S_n \cap R_1(\Delta))| = O(\ln n).
\end{equation}

Plugging Eqs.~\ref{eq:Biqup1}, \ref{eq:Biqup2} and \ref{eq:Biqup3} into Eq.~\ref{eq:Biqup4} proves
$$ \EMN = O(\ln n).$$

\par\noindent\underline {\bf Lower Bound:}

Let $\Delta =  \frac 1 {\sqrt n}  +2\delta$.  From Lemma \ref{lem:easy mu}(c), 
$\mu(T(\Delta)) = O(1/n).$   From  Lemma \ref{lem:easy mu}(c) (d), $\mu(R_1(\Delta)) = \Theta(1/ \sqrt n).$

Thus, from Lemma \ref{lem: basic mu}, for large enough $n,$
$$\Pr( S_n  \cap T(\Delta) = \emptyset)  =  \left(1 - \mu(T(\Delta))\right)^n  = \Omega(1).$$
Next observe that the  only points outside of $R_1(\Delta)$   that can dominate a point in $R_1(\Delta)$ are points in $T(\Delta).$
Thus, if $ S_n \cap T(\Delta) = \emptyset$ then
$$ \MAX(S_n) \cap R_1(\Delta) =\MAX(S_n \cap  R_1(\Delta)).$$
This combined with  Lemma \ref{lem: side rec}(b) yields
\begin{eqnarray*}
\EMN &\ge & \EXP{|\MAX(S_n) \cap R_1(\Delta) |}\\
         &\ge & \EXP{|\MAX(S_n) \cap R_1(\Delta) |  \Bigm| S_n \cap T(\Delta) = \emptyset } \cdot \Pr\left(  S_n \cap T(\Delta) = \emptyset  \right)\\
             &= & \EXP{|\MAX(S_n \cap R_1(\Delta)) |  \Bigm| S_n \cap T(\Delta) = \emptyset } \cdot \Pr\left(  S_n \cap T(\Delta) = \emptyset  \right)\\     
         &=& \Omega(\ln n) \cdot \Omega(1) = \Omega(\ln n).
         \end{eqnarray*}

\medskip
\par\noindent\underline{(b) $p \not = \infty:$}\\

\medskip

\par\noindent\underline {\bf Upper Bound:}

Define 
$\bar B_p= B_p((0,0), 1- \delta)$ and $\BarBallp$ as the uniform distribution over $\bar B_p.$

 Let $S'_m= \{v'_1,v'_2,\ldots,v'_n\}$ be a set $m$ points chosen from 
$\BarBallp.$  Since scaling a $L_p$ ball this way does not change the distribution of the number of maxima 
$\EXP{\MAX(S'_m)} = \Theta(\sqrt m).$

Now let $S_n$ be $n$ points chosen from $\Ballpq p q.$
 Then
\begin{eqnarray*}
\EXP{|\MAX(S_n)|}  &=& \EXP{|\MAX(S_n) \cap \bar B_p|}  +  \EXP{|\MAX(S_n) \cap \left(\left(B_p + \delta B_q \right)\setminus \bar B_p\right)|}\\
       		   &\le& \EXP{|\MAX(S_n \cap B_p)|}
   + \EXP{|S_n \cap \left(\left(B_p + \delta B_q \right)\setminus \bar B_p\right)|}\\
			& \le& \EXP{|\MAX(S_n \cap \bar B_p)|}   
                            +n \mu   \left(
					\left(B_p + \delta B_q \right) \setminus \bar B_p
					\right)\\
			& =& \EXP{|\MAX(S_n \cap \bar B_p|)}    + O(\sqrt n),
\end{eqnarray*}
where the last line uses the fact that
\begin{eqnarray*}
\mu   \left(\left(B_p + \delta B_q\right) \setminus B_p \right) &=&
O\left(\Area\left(\left(B_p + \delta B_q\right) \setminus B_p \right)\right)
  + O\left(\Area\left(B_p \setminus \bar B_p \right)\right)\\
&=& O(\delta)  + O(\delta) =  O\left(\frac  1 {\sqrt n}\right).
\end{eqnarray*}

From the triangle inequality, if $v \in \bar B_p$ then $v + \delta B_q \subset B_p.$ Thus
$\{u \in B_p \,:\, v-u \in \delta B_q\} = B_q(v,\delta)$ and
$$f(v) = \frac 1 {a_p a_q} \frac  {\Area(\{u \in B_p \,:\, v-u \in \delta B_q\})} {\delta^2}
= \frac 1 {a_p a_q} \frac  {\Area(B_q(v,\delta))} {\delta^2}
= \frac 1 {a_p}.$$

Now note that $S_n \cap \bar B_p$   has an equivalent distribution to 
\begin{enumerate}
\item Choosing random variable $X$, from a binomial distribution 
$B\left(n, \mu(\bar B_p)\right).$
\item  For $v \in \bar B_p$, setting
$\bar f(v) = \frac {f(v)} {\mu(\bar B_p)}$.   $\bar f(v)$ is the conditional  probability density function for choosing a point $v$ from $\Ballpq p q$ conditioned on knowing that $v \in B_p.$  
\item  Choosing $X$ points (in $\bar B_p$)   from the distribution defined by $\bar f(p).$
\end{enumerate}
But  since $f(v)$ is constant for $v \in \bar B_p$, $\bar f(v)$ is constant as well.  This means that $S_n \cap \bar B_p$  has the same distribution as $X$ points chosen uniformly from the  ball $\bar B_p.$
Conditioning on $X$ this implies
$$\EXP{|MAX(S_n \cap \bar B_p)| \,\bigm|\,  X} = \Theta\left(\sqrt X\right).$$
But $X \le n$ so 
$$ \EXP{|\MAX(S_n \cap B_p)|}  = O(\sqrt n)$$
and thus
$$\EXP{|\MAX(S_n)|}  \le  \EXP{|\MAX(S_n \cap B_p|)}    + O(\sqrt n) = O(\sqrt n).$$

\par\noindent\underline{\bf Lower Bound}  See Fig.~\ref{fig: small delta}(a).

Let $g_p(x) = \left(1  - x^{q}\right)^{1/q}$ be the equation of the upper boundary of $B_p$ for  $0 \le x \le 1.$
The Theorem of the Mean states that if $x' < x,$ then 
$$g_p(x) - g_p(x') = g'(z) (x-x'),$$
for some $z \in [x',x]$.  In particular,  since $g'_p(z)$ is bounded for $z \in [1/3,2/3]$  $$ \mbox{if } \  x,x' \in [1/3,2/3]   \quad \Rightarrow  \quad g(x) = g(x') + \Theta (x-x').$$

Now let $g(x)$ be the curve of the upper boundary of $B_p + \delta B_q$. By construction $g(x)$ is monotonically decreasing for $x \in [0,1]$ and is concave down. Because of the construction, we also know that $g(x)$ is within a distance $\Theta(1/\sqrt n)$ of $g_p(x)$ in the following sense:  there exists a $c>0$ such that
$$ g_p(x) \le   g(x) \le g_p(x) + c/{\sqrt n}
\quad\mbox{and}\quad
g^{-1}_p(y) \le   g^{-1}(y) \le g^{-1}_p(y) + c/{\sqrt n}.$$

Now set $\beta = (3+c) /\sqrt n$,  $t =\lfloor 1/3 \beta\rfloor$ and, for  $i=0,1,\ldots,t-1$  
\begin{eqnarray*}
x_i&=& \frac 1 3 + i \beta,\\
y_i &=& g_p(x_i)-\frac 1 {\sqrt n},\\
p_i &=&  (x_i, y_{i+1}).
\end{eqnarray*}
Note that
 by construction, $p_i \in B_p$ and, furthermore  $g_p(x_i) -y_{i+1}  = \Theta(1/\sqrt n)$  and $g^{-1}(y_i+1) - x_i = \Theta(1/\sqrt n).$
Thus  $\Area(P(p_i)) \cap B_p = \Omega(1/n)$ and, thus, by Lemma \ref {lem:easy mu} (c),
$\mu(P(p_i)) = \Omega(1/n).$  Since, again by construction,  
if $i \not = j$,  $P(p_i) \cap P(p_j)=\emptyset$  we can apply Lemma \ref{lem: lb}  to derive
$$\EXP{|MAX(S_n)|} = \Omega(t) = \Omega(\sqrt n),$$
and are done. \end{proof}

\section{Analysis of $\Ballpq 1 1$}
\label{sec: B1B1}
This section derives cell (ii)(c) in Theorem \ref{thm: main}, that is,  if $n$  points are chosen from $\bfd = \Ballpq 1 1$ and  $\frac 1  {\sqrt n} \le \delta \le 1$ then $\EMN = \Theta \left(  \frac {n^{1/3}}  { \delta^{1/3}}  \right)$.
Applying Lemma  \ref {lem: scaling} for   $1  \le \delta \le   \sqrt n $     yields cell (ii)(d), i.e., that   $\EMN =  \Theta \left(   \delta^{1/3}  n^{1/3} \right).$

Corollary \ref {cor: Quadrants} states that 
$$\EMN =   \EXP {|Q_1 \cap \MAX(S_n)|}  + O(1),$$
so our analysis will be restricted to the upper-right  quadrant $Q_1$.
Our approach will be to
\begin{enumerate}
\item State a convenient expression for $\mu(u)$ (proof delayed until later).
\item Derive a lower bound using Lemma \ref {lem: lb} by defining an appropriate  pairwise disjoint collection of dominant regions 
\item Derive an upper bound by partitioning $D$ into appropriate regions and applying the sweep Lemma.
\end{enumerate}

\begin{figure}[t]
\begin{center}
\includegraphics[width=2.2in]{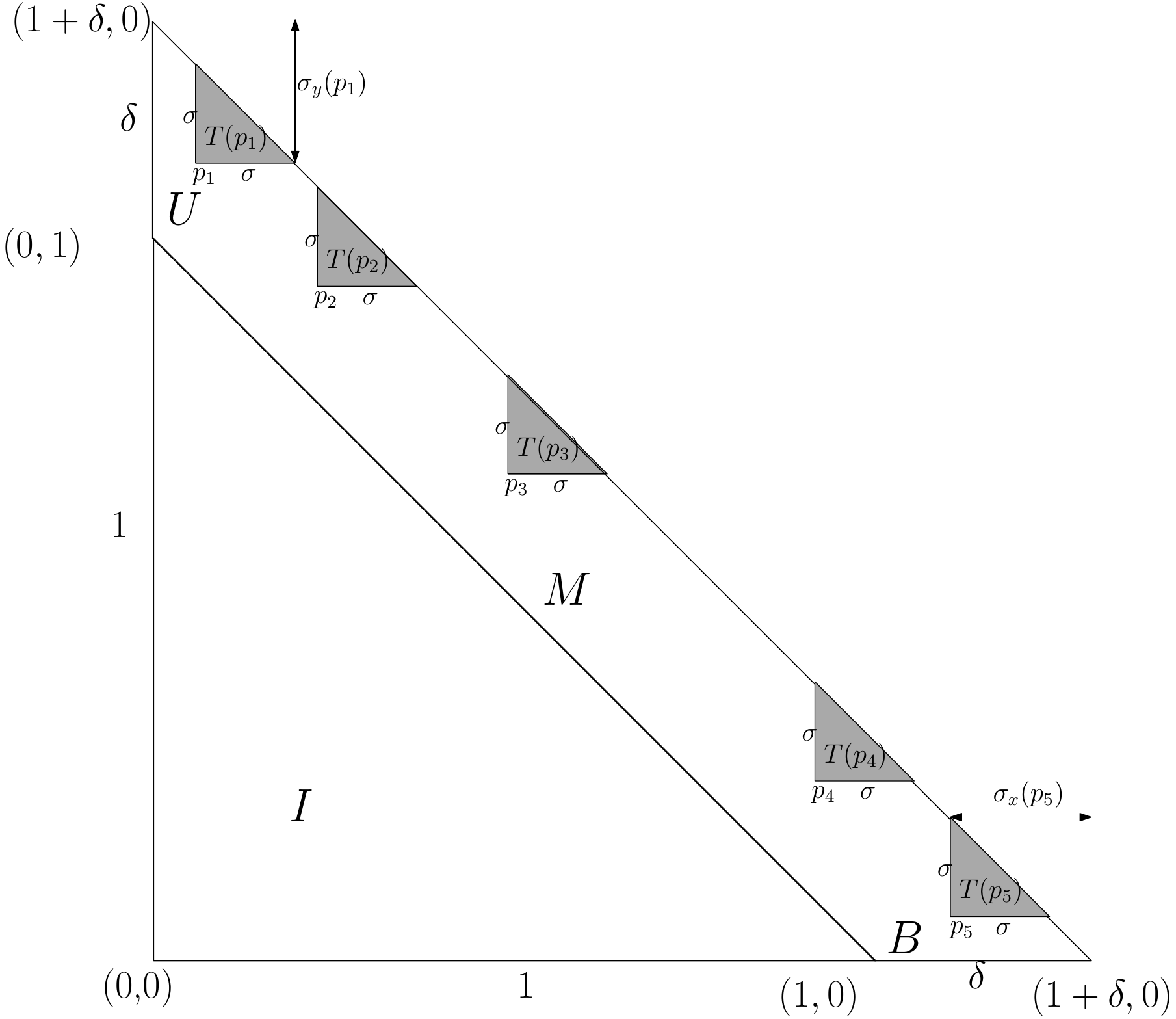}  \quad \includegraphics[width=2.2in]{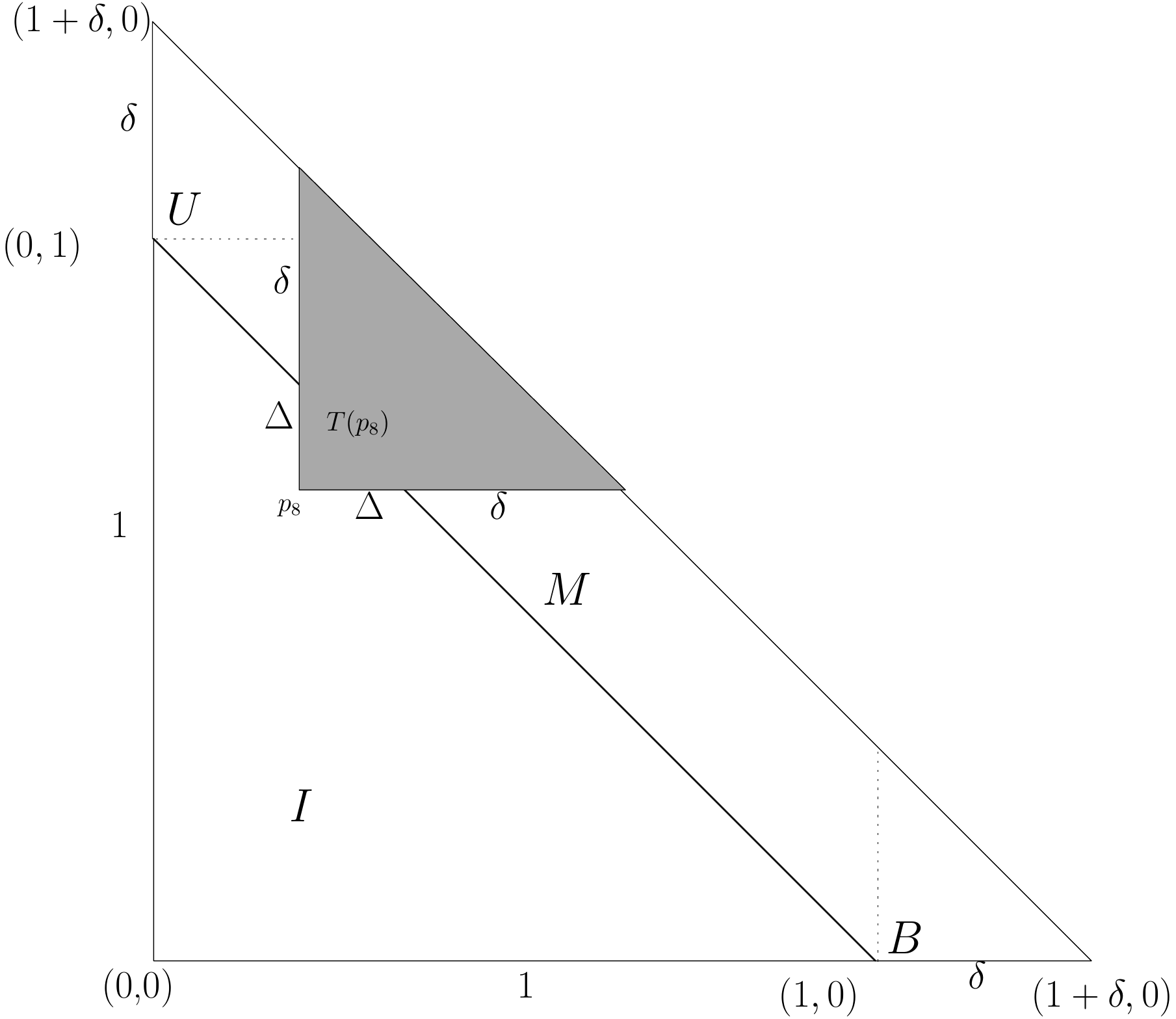}
\end{center}
\caption{Illustration of Definition  \ref{def:Regionsa}  and  Lemma \ref{lem:mainB1B1}.} %
\lab{fig:B1_Regionsd}
\rule{5.5in}{0.5pt}
\end{figure}

\begin{Definition}
\lab{def:Regionsa}
Let $D_1= (B_1 + \delta B_1) \cap Q_1.$  Set 
\begin{align*}
I   =&  D_1 \cap B_1,    & U  =&  D_1  \cap \left\{ u \in \Re^2 \,:\,  u.y \ge 1\right\}, \\
M =&  D_1 \setminus (I \cup U \cup B), & B  =&  D_1  \cap \left\{ u \in \Re^2 \,:\,  u.x \ge 1\right\}.\\
\end{align*}
\end{Definition}

In the following Lemma, we deal with  $D = B_1 + \delta B_1$ and $D_1 = D \cap Q_1$ as given in Definition \ref{def:Regionsa}, and for all $p \in D_1$ we consider $P(p)$ introduced in Definition \ref{def: PDdef} (see Fig.~\ref{fig:B1_Regionsd})
Notice $P(p)$ is 
the isosceles right triangle with base $p$ and hypotenuse flush with the upper-right border of $B_1+ \delta B_1.$
For notational convenience, to emphasise that $P(p)$ is a triangle,  set $T(p) = P(p).$   In particular notice that  $T(p)$ is a dominant region.
Further set 
$$
\begin{array}{lcll}
\sigma(p)  &=& 1 + \delta - p.x -p.y \quad & \mbox{side length of triangle $T(p)$},\\
\Delta_x(p) &=&  1 +\delta -p.x &\mbox{$p$'s $x$-distance to rightmost point in $B_1 + \delta B_1$},\\
\Delta_y(p) &=&  1 +\delta -p.y & \mbox{$p$'s $y$-distance to highest point in $B_1 + \delta B_1$}. \\
\end{array}
$$

\begin{Lemma}
\label{lem:mainB1B1}
Let $D$, $D_1$ and for every  $p \in D_1$, $T(p)$  defined as above. For  $p$ chosen from $\bfd = \Ballpq 1 1$, the formula for calculating  $\mu(T(p))$ differs by location of $p,$ as below:  
\begin{enumerate}
\item If $p \in U$   $\Rightarrow  \mu(T(p)) = \Theta\left( \Delta_y(p)\ \frac {  \sigma^3(p)} {\delta^2}\right) = O\left( \frac {\sigma^3(p)} {\delta}\right)$.
\item If $p \in B$    $\Rightarrow  \mu(T(p)) =\Theta\left( \Delta_x(p)\  \frac {\sigma^3(p)} {\delta^2}\right) = O\left( \frac {\sigma^3(p)} {\delta}\right).$
\item If $p \in M$   $ \Rightarrow  \mu(T(p)) =\Theta\left( \frac {\sigma^3(p)} {\delta}\right)$.
\item If $p \in I$ $ \Rightarrow  \mu(T(p))= \Theta(\sigma^2(p))$.
\end{enumerate}
\end{Lemma}
In  Fig.~\ref{fig:B1_Regionsd},   $p_1$ is an example for case~1, $p_5$  for case~2,
$p_2, p_3$ and $p_4$  for case~3, and $p_8$ for case~4.

Before proving the Lemma, we  use it to derive the upper and lower asymptotic bounds on  $\EMN$. 
\begin{Lemma}
\label{lem: B1B1LB}
Let $S_n$ be $n$ points chosen from the distribution $\bfd = \Ballpq 1 1$ with $\frac 1 {\sqrt n} \le \delta \le 1$. Then
$$\EMN = \Omega \left(  \frac {n^{1/3}}  { \delta^{1/3}}  \right).$$ 
\end{Lemma}

\begin{proof}
\begin{figure}[t]
\begin{center}
 \includegraphics[width=2.3in]{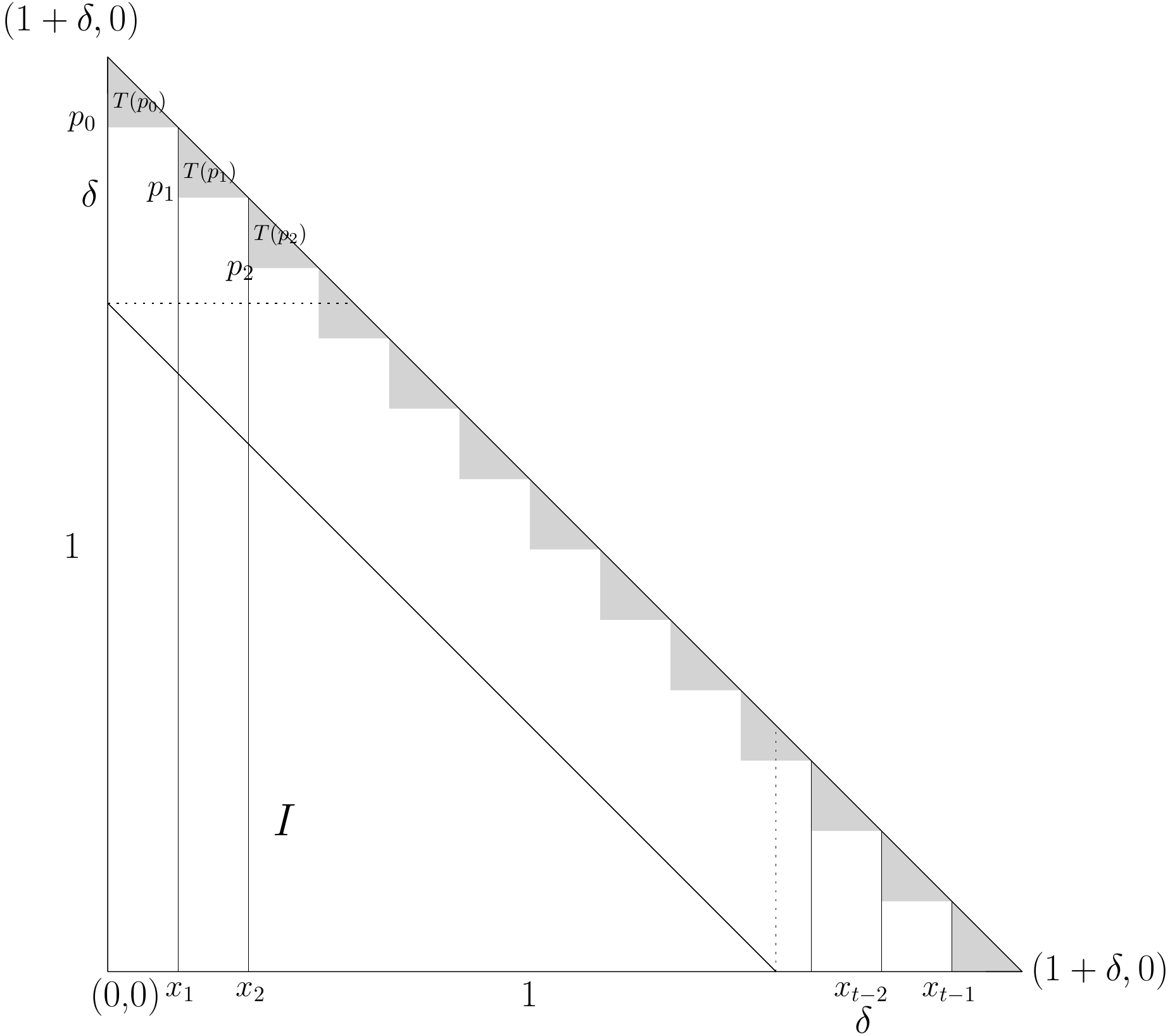}
\end{center}
\caption{The lower bound construction for $\Ballpq 1 1$.  Note that the $T(p_i)$ (in gray) are pairwise disjoint dominant regions.  Their horizontal and vertical sides all have length $\sigma$ which is chosen to force $\mu(T(p_i)) = \Theta(1/n).$
} %
\lab{fig:B1_Regionsa}
\rule{5.5in}{0.5pt}
\end{figure}
The proof creates 
 $m=\Omega \left(  \frac {n^{1/3}}  { \delta^{1/3}}  \right)$ points $p_i$  such that, $\forall  i\not =j,$  $T(p_i) \cap T(p_{j}) = \emptyset$  and, 
$\forall i,\,  \mu(T(p_i) = \Theta(\frac 1 n)$.  See Figure \ref{fig:B1_Regionsa}. Since the $T(p_i)$ are all dominant regions, 
Lemma \ref {lem: lb}
 then immediately imples that $\EXP{M_n} = \Omega(m) = \Omega \left(  \frac {n^{1/3}}  { \delta^{1/3}}  \right)$.

Start by fixing  $\sigma= \left(  \frac \delta n \right)^{1/3}$, setting $m' = \lfloor \frac {1 + \delta} \sigma \rfloor $ and,  for $ i = 0,\ldots m'$  defining
$$ x_i = i \sigma,
\quad \quad
y_i = 1+\delta- x_i - \sigma,
\quad\quad  p_i = (x_i,y_i).
$$

Note that $m' \sigma \le 1 + \delta$ so, if  $i \le m'$,  $p_i \in D_1$.   
Note that  
$$\sigma(p_i) = 1 + \delta - x_i - y_ i = \sigma.$$
Thus, by the construction,  
if $i \not =j$ then $T(p_i) \cap T(p_{j}) = \emptyset$.

Next note that, by definition,  
\begin{equation}
\label{eq:B1B11na}
\frac  {\sigma^3} \delta = \frac 1 n.
\end{equation}

Finally note that because  $ \frac 1 {\sqrt n} \le \delta$,  $\sigma^3 \le \frac \delta n \le \delta^3$ so  $\sigma \le \delta.$ 
Thus, for  $i \le m'-1$,    $p_i(x) \in B \cup U \cup M$. 
Now let $i$ satisfy   $ \frac {m'} 3 \le  i \le  \frac {2 m'} 3$. Then 
 \begin{equation}\label{eq:B1B1LowerRestrict}
  \frac 1 3  (1 + \delta) \le x_i \le  \frac 2 3 (1 + \delta),
  \end{equation}
and one of the following must be true
\begin{itemize}
\item $p(x_i) \in M$.
\item $p(x_i) \in B$.    Then  $\Delta_x(p_i) = 1 + \delta - x_i  \ge \frac 1 3 (1 + \delta) \ge \frac 1 3.$\\ Also, by definition,
$\Delta_x(p_i) \le \delta \le 1.$\\
Together these  imply   $\frac {\Delta_x(p_i)} \delta = \Theta(1)$.
\item $p(x_i) \in U$.  A symmetric analysis 
 shows  $\frac {\Delta_y(p_i)} \delta = \Theta(1)$.
\end{itemize}
Applying Cases 1-3 in Lemma \ref {lem:mainB1B1} then shows that for all $i$  satisfying Eq.~\ref{eq:B1B1LowerRestrict}, $\mu(T(p_i)) = \Omega(1/n).$

There are $m = m/3$ values of $i$ satsifying Eq.~\ref{eq:B1B1LowerRestrict}, and thus, from 
Lemma~\ref{lem: lb}.

$$\EMN = \Omega\left( \frac {m'} 3 \right) = \Omega\left( \left(\frac{n}{\delta}\right)^{1/3}\right).$$

 {\em Note:  The motivation  for the restriction (\ref{eq:B1B1LowerRestrict}) is that, if $\delta$ is close to $1$, then almost no $p_i$ would be in $M$.  It is therefore  necessary to analyze the measure of $T(p_i)$ for $p_i \in B,U.$  But,  in those cases, if  $i$ is close to $0$ or $m'$  then $ \mu(T(p_i)$ could be $o(1/n)$.  The proof therefore needs to bound away from those corners.}
\end{proof}

The upper bound is more technical.
\begin{Lemma}
\label{lem: B1B1UB}
Let $S_n$ be $n$ points chosen from the distribution $\bfd = \Ballpq 1 1$ with $\frac 1 {\sqrt n} \le \delta \le 1$. Then
$$\EMN = O \left(  \frac {n^{1/3}}  { \delta^{1/3}}  \right).$$
\end{Lemma}

\begin{figure}[t]
\begin{center}
\includegraphics[width=2.3in]{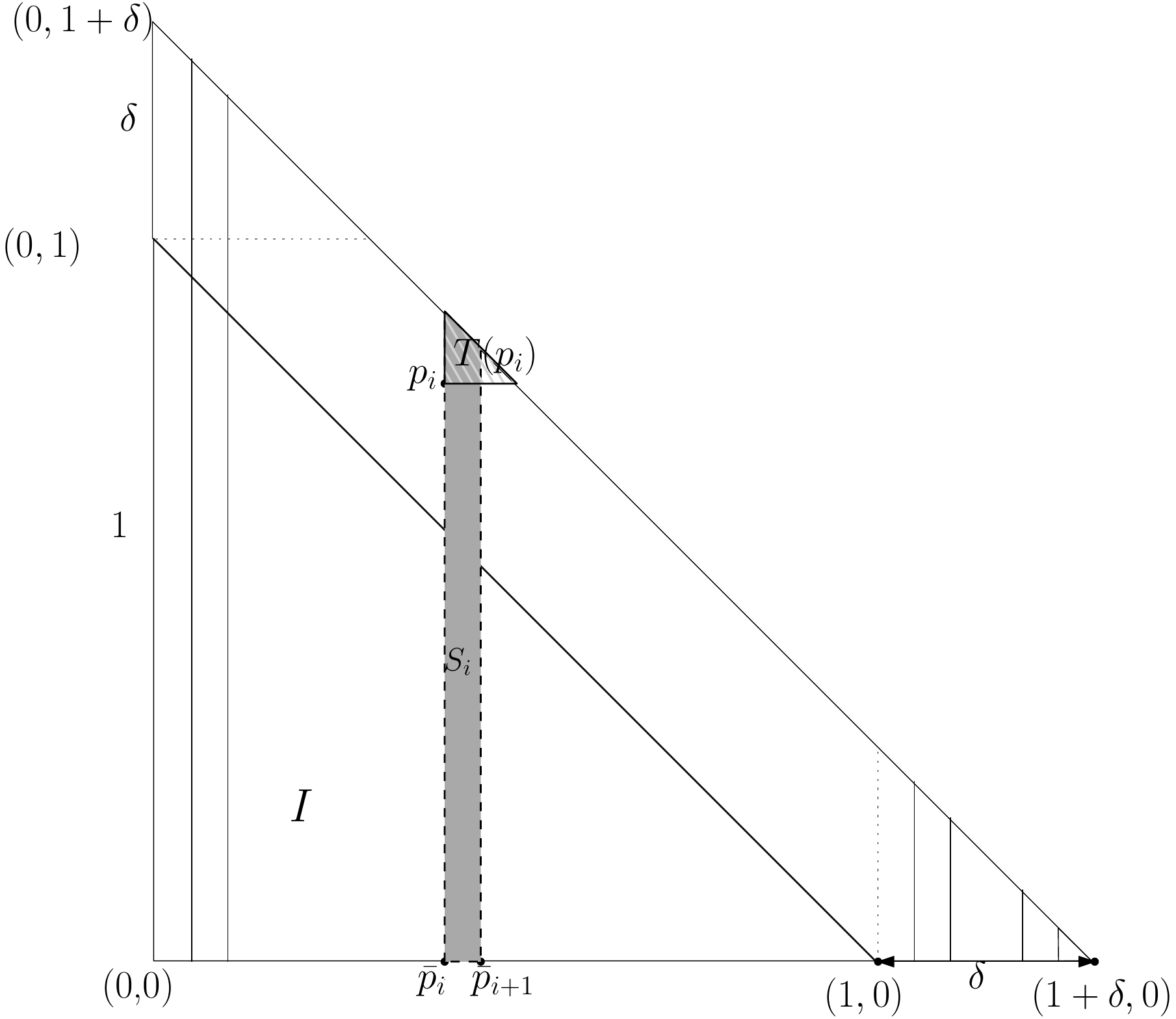}  \quad \includegraphics[width=2.3in]{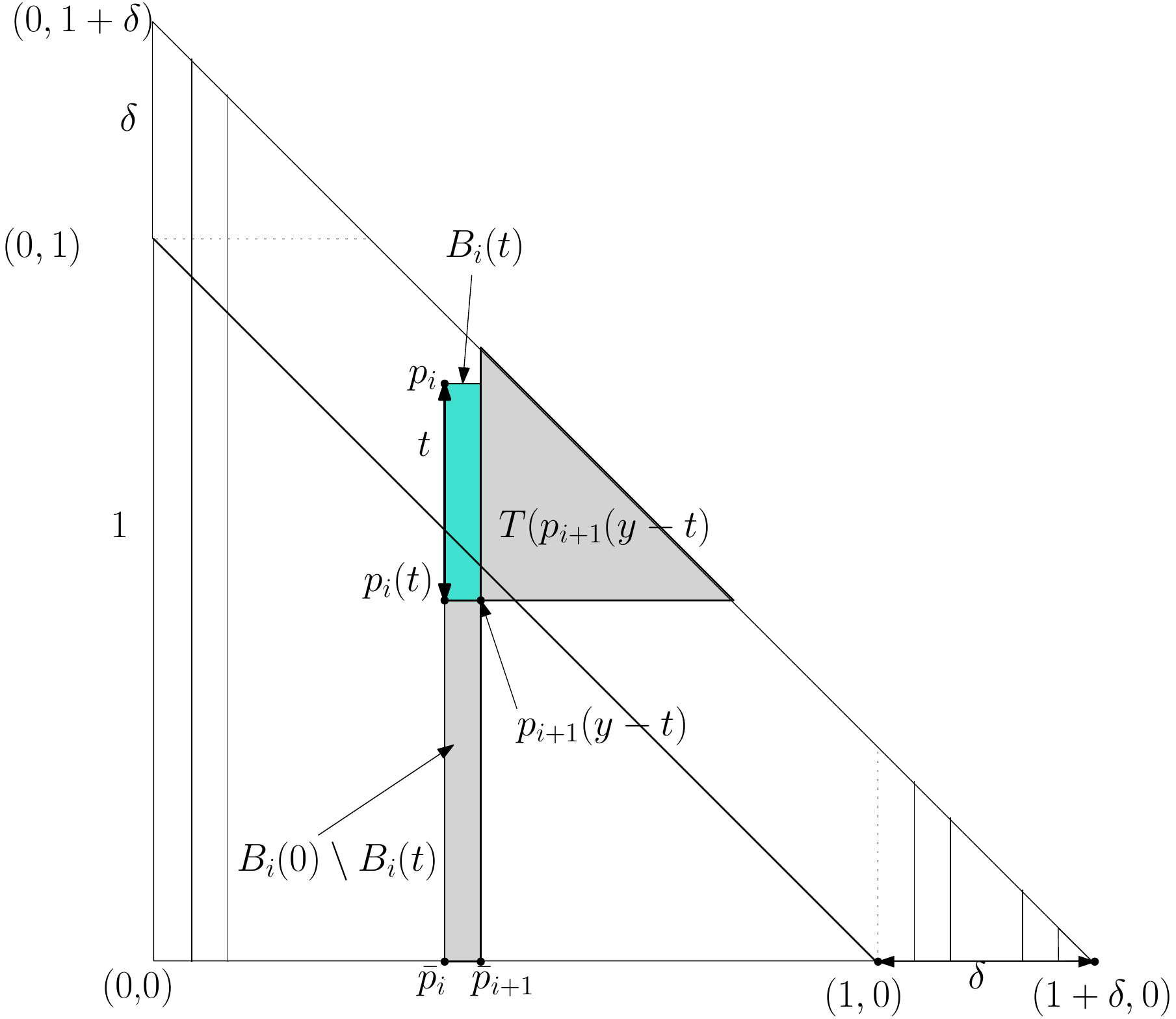}
\end{center}
\caption{Illustration of the proof of Lemma \ref {lem: B1B1UB}.  $S_i$ denotes $\Strip_i.$ } %
\lab{fig:B1_Regionsc}
\rule{5.5in}{0.5pt}
\end{figure}

\begin{proof}
The proof will assume that $\delta > \frac 8  {\sqrt n}.$  The case 
follows directly from Lemma \ref{lem: limiting}.

As noted,  $\EMN = \EXP{|\MAX(S_n) \cap D_1|} +O(1).$ Observe that, by the $x,y$ symmetry of $\bfd$,
$$ \EXP{|\MAX(S_n) \cap B|} =  \EXP{|\MAX(S_n) \cap U|}.$$
Thus
$$\EMN = \EXP{|\MAX(S_n) \cap (U \cup M \cup I)|}.$$

The   upper bound  approach will be,  for an  appropriate value $\sigma$,   to partition $U \cup M \cup I$ 
into $m = \Theta (1/\sigma)$ vertical {\em strips} of width $\sigma$ and show that each strip contains only $O(1)$ expected maxima.

\medskip

Set $\sigma= \left(  \frac \delta n \right)^{1/3}$.    Note that this implies
\begin{equation*}
\label{eq:B1B1UBsigbd}
\sigma =  \left(  \frac \delta n \right)^{1/3}   \le \left(  \delta  \cdot \frac  {\delta^2} {64} \right)^{1/3} = \frac  \delta 4.
\end{equation*}
Now set\footnote{Without loss of generality, we  assume that $1/\sigma $ is an integer.  If not we replace $\sigma$  by the  largest  $\sigma' \le \sigma$ such that $1/\sigma'$ is an integer.  This perturbation will not change $\sigma$ by more than a factor of $2.$}
 $m = \frac 1 \sigma$  and,  for $ i = 0,\ldots m$  define
$$ x_i = i \sigma,
\quad \quad
y_i = 1+\delta- x_i - 2\sigma,
$$
and
$$p_i(y) = (x_i,y),
\quad 
p_i=(p_i(y_i)) = (x_i,y_i),   
\quad  
\bar p_i = (x_i,0).
$$

Next, for $0 \le i < m$,  define
$$\Strip_i =  (U \cup M \cup I)  \cap \{u \in \Re^2 \,:\,  x_i \le u.x < x_{i+1}\}.$$
By definition, 
$$ U \cup M \cup I  = \bigcup_{i=0}^{m-1} \Strip_i.$$

Now define
\begin{eqnarray*}
B_i(t) &= & \Strip_i \cap \{u \,:\, y_i-t \le u.y \le y_i\},\\
A_i(t) &=& T(p_{i+1}(y_i-t)).
\end{eqnarray*}

To proceed note that since $i \le m-1$ and $\sigma \le \delta/4,$
$$y_i = 1+\delta- x_i - 2\sigma = 1 + \delta - (i+2)\sigma \ge 0.$$
Thus 
\begin{eqnarray*}
\Strip_i 
 & \subseteq& T(p_i)  \ \cup \  \left(\Strip_i \cap \{u \,:\, 0 \le u.y \le y_i\}\right)\\
&=& T(p_i) \cup B_i(y_i),
\end{eqnarray*}
where $B_i(y_i) \not=\emptyset.$
From 
Lemma  \ref {lem: basic mu} (b)
$$\EXP{|MAX(S_n) \cap T(p_i)|}  \le \EXP{|S_n| \cap T(p_i)|} = O(n \mu(T(p_i)),$$ 
so
$$\EXP{|MAX(S_n) \cap \Strip_i |}  \le  \EXP{|MAX(S_n) \cap B_i(y_i)|} + n \mu(T(p_i).$$

Finally, observe that all points in $B_i(y_i)\setminus B_i(t)$ are dominated by all points in $A_i(t)$.  
Thus Lemma \ref {lem: Sweep}  implies  that if, $\forall t,$  $\mu(B_i(t)) = O(\mu(A_i(t)))$ then
$$ \EXP{MAX(S_n) \cap B_i(y_i)} = O(1).$$
But
$$B_i(t) \subset T(p_i(y-t)) \quad  \Rightarrow  \quad  \mu(B_i(t)) \le  \mu(T(p_i(y-t))).$$
So,  proving 
\begin{equation} \label{eq:B1B1UPmain}
\forall t,\,   \mu(T(p_i(y-t)))  = \Theta(\mu(T(p_{i+1}(y_i-t)))),
\end{equation}
it will suffice to prove 
 $\EXP{MAX(S_n) \cap B(y_i)} = O(1)$.

Summarizing, the above discussion states that, if  Eq.~\ref{eq:B1B1UPmain} is true, then 
 \begin{eqnarray}
\label{B1B1UPfull}
\EXP{M_n} &=& \sum_{0 \le i \le m-1}\EXP{|MAX(S_n) \cap \Strip_i |} \nonumber\\
                 &= & O\left(\sum_{i=0}^{m-1}  \bigl(1 + n \mu(T(p_i))\bigr)\right)\\
                &=& O\left(m   + \sum_{i=0}^{m-1} n \mu(T(p_i))\right).\nonumber
\end{eqnarray}

The analysis requires the following simple observations that are true for $\forall i < m $ and $\forall t \le y_i:$

\begin{align*}
\sigma(p_i)  &= 1 + \delta - x_i - y_ i = 2 \sigma,   &	\sigma(p_{i+1}(y_i)) &= \sigma(p_i) - \sigma   = \sigma,\\
\sigma(p_i(y_i-t)) &= \sigma(p_i) + t  = 2 \sigma+ t, & \quad 	\sigma(p_{i+1}(y_i-t)) &= \sigma(p_{i+1}) + t  = \sigma+ t,	
\end{align*}
\begin{align*}
\Delta_x(p_i) &= \Delta_x(p_i(y-t)),				& 	\Delta_y(p_i) &= \Delta_y(p_{i+1}(y_i)),\\
\Delta_x(p_{i+1})  &= \Delta_x(p_i) + \sigma,	&	\Delta_y(p_{i+1}(y_i-t)) &= \Delta_y(p_{i+1}) + t,
\end{align*}
and
\begin{equation}
\label{eq:B1B1UB Delta Bounds}
\begin{array} {rcccl}
(m-i) \sigma &\le & \Delta_x(p_i) &\le  & (m-i+1)\sigma,\\
(i+2) \sigma  & \le & \Delta_y(p_i) & \le & (i+3) \sigma.\\
\end{array}
\end{equation}

Since $\sigma \le \delta/4$  the observations above yield
$$\sigma(p_i) = 2 \sigma  \le  \frac \delta 2, $$  so  
$p_i \not \in I$ and thus $p_i \in M \cup U$. 
Cases 1 or  3 in  Lemma \ref{lem:mainB1B1} then apply and 
$$\mu(T(p_i)) = O\left(\frac {\sigma^3} \delta \right) = O \left(\frac 1 n\right).$$
It remains to prove Eq.~\ref{eq:B1B1UPmain}.  This will require splitting into multiple cases depending upon the value of $x_i$.

\begin{itemize}
\item $\delta  \le x_i \le  1:$  If  $t \le \delta - 2 \sigma$ then  $p_i(y_i-t) \in M$;\\ If  $\delta - 2 \sigma \le t$ then $p_i(y_i-t) \in I$.  This yields
\begin{equation*}
\label{B1B1 UP pi b}
\mu(T(p_i(y_i-t))
=
\left\{
\begin{array}{ll}
\Theta\left (  \frac {(2 \sigma + t)^3} {\delta} \right) &  \mbox{if  $0 \le t \le \delta - 2 \sigma,$}\\
\Theta\left ((2 \sigma + t)^2\right)  &  \mbox{if   $ \delta - 2 \sigma  \le t$.}
\end{array}
\right.
\end{equation*}

\item  $ x_i \le  \delta:$  
If  $t \le \delta - (i+2) \sigma$ then  $p_i(y_i-t) \in U$;\\  If  $ \delta - (i+2) \sigma < t \le \delta - 2 \sigma$ then $p_i(y_i-t) \in M$;\\  
If  $ \delta - 2 \sigma  \le t$ then $p_i(y_i-t) \in I.$

  Using (\ref{eq:B1B1UB Delta Bounds})  and working through the details gives
\begin{equation*}
\label{B1B1 UP pi c}
\mu(T(p_i(y_i-t)))
=
\left\{
\begin{array}{ll}
\Theta\left (  (i+2) \sigma \frac {(2 \sigma + t)^3} {\delta^2} \right) &  \mbox{if  $0 \le t \le \delta - \sigma (i+2)$},\\
\Theta\left(\frac {(2 \sigma + t)^3} {\delta} \right) &  \mbox{if  $\delta - \sigma (i+2) \le t \le  \delta - 2 \sigma$},\\
\Theta\left ((2 \sigma + t)^2\right)  &  \mbox{if  $\delta - 2 \sigma  \le t$.}
\end{array}
\right.
\end{equation*}
\end{itemize}

We now work through similar calculations for  $p_{i+1}(y_i-t)$ to prove
\begin{itemize}
\item $\delta - \sigma  \le x_{i+1} \le  1:$
\begin{equation*}
\label{B1B1 UP pip1 b}
\mu(T(p_{i+1}(y_i-t))
=
\left\{
\begin{array}{ll}
\Theta\left (  \frac {( \sigma + t)^3} {\delta} \right) &  \mbox{if  $0 \le t \le \delta -  \sigma$},\\
\Theta\left (( \sigma + t)^2\right)  &  \mbox{if   $ \delta -  \sigma  \le t $}.
\end{array}
\right.
\end{equation*}
\item  $ x_{i+1}< \delta - \sigma$:

\begin{equation*}
\label{B1B1 UP pip1 c}
\mu(T(p_{i+1}(y_i-t)))
=
\left\{
\begin{array}{ll}
\Theta\left (  (i+3) \sigma \frac {(\sigma + t)^3} {\delta^2} \right) &  \mbox{if  $0 \le t \le \delta - \sigma (i+2)$},\\
\Theta\left(\frac {( \sigma + t)^3}{\delta} \right) &  \mbox{if  $\delta - \sigma (i+2) \le t \le  \delta - \sigma$},\\
\Theta\left ((\sigma + t)^2\right)  &   \mbox{if   $ \delta -  \sigma  \le t.$}
\end{array}
\right.
\end{equation*}
\end{itemize}
We now  prove Eq.~\ref{eq:B1B1UPmain}.

Note that,  $i+3 = \Theta(i+2)$  and $\sigma + t  = \Theta(2 \sigma +  t)$ with the constants not dependent upon $i$ or $t$.  We therefore consider the items in those pairs as interchangeable with each other.

\begin{itemize}
\item $ \delta - 2 \sigma  \le x_i < x_{i+1} \le 1:$ We divide $t$ into three ranges and 
 write the values of the functions in each of those ranges.
$$
\begin{array}{|l|c|c|c|}
\hline
t \in  & [0,\delta - 2 \sigma] &  [\delta - 2 \sigma, \delta - \sigma] & [\delta - \sigma,y_i]\\[0.1in]\hline
\mu(T(p_i(y_i-t))  =  &  \Theta\left (  \frac {(2 \sigma + t)^3} {\delta} \right)    & \Theta\left ((2 \sigma + t)^2\right)    &   \Theta\left ((2 \sigma + t)^2\right) \\[0.1in] \hline
\mu(T(p_{i+1}(y_i-t))  =& \Theta\left (  \frac {( \sigma + t)^3} {\delta} \right)   &  \Theta\left (  \frac {( \sigma + t)^3} {\delta} \right)  & \Theta\left (( \sigma + t)^2\right)  \\\hline
\end{array}
$$
Note that in the first and third ranges the values in the two rows are within  a constant factor of each other.
For the middle range recall that $\sigma \le \delta/4$ so, if  $\delta - 2 \sigma \le t < \delta - \sigma $ then $\sigma+t = \Theta(\delta)$ 
so in the middle range as well the two rows are within  a constant factor of each other, proving   Eq.~\ref{eq:B1B1UPmain}  for the complete range of $t$.

\item $x_{i+1} < \delta -2\sigma:$ We divide $t$ into four  ranges and 
 write the values of the functions in each of those ranges.
$$
\hspace*{-.4in}
\small
\begin{array}{|l|c|c|c|c|}
\hline
t \in  &  [0, \delta - \sigma (i+2)]&  [\delta -  \sigma(i+2), \delta - 2\sigma] & [\delta - 2\delta-\sigma]& [\delta-\sigma,y_i]\\[0.1in]\hline
\mu(T(p_i(y_i-t))  =  &     \Theta\left (  (i+2) \sigma \frac {(2 \sigma + t)^3} {\delta^2} \right) &  \Theta\left(\frac {(2 \sigma + t)^3} {\delta} \right) &  \Theta\left (2(\sigma + t)^2\right)  & \Theta\left (2(\sigma + t)^2\right)\\[0.1in] \hline
\mu(T(p_{i+1}(y_i-t))  =&\Theta\left (  (i+3) \sigma \frac {(\sigma + t)^3} {\delta^2} \right)   & \Theta\left(\frac {( \sigma + t)^3} {\delta} \right)  & \Theta\left(\frac {( \sigma + t)^3} {\delta} \right)& \Theta\left ((\sigma + t)^2\right) \\[0.1in]\hline
\end{array}
$$
Note that in the first, second and fourth  ranges the values in the two rows are within  a constant factor of each other.
For the 3rd range  range,  recall again that $\sigma \le \delta/4$ so, if  $\delta - 2 \sigma \le t < \delta - \sigma $ then $\sigma+t = \Theta(\delta)$ 
so in the middle range as well the two rows are within  a constant factor of each other, proving   Eq.~\ref{eq:B1B1UPmain}  for the complete range of $t$.

\end{itemize}

We have just  proven    that Eq.~\ref{eq:B1B1UPmain} is valid for all cases. We can thus apply 
Eq.~\ref{B1B1UPfull} to prove
 \begin{eqnarray*}
\EXP{M_n}         &=& O\left(  m+ \sum_{i=0}^{m-1} n \mu(T(p_i))\right)\\
			&=&  O\left(  m   + m  n \Theta(1/n) \right) = O (m ) =  O\left(\left(\frac n \delta\right)^{1/3}\right).
\end{eqnarray*}
and the proof is complete.

\end{proof}

We now return to prove Lemma  \ref{lem:mainB1B1}.

\begin{proof} of Lemma  \ref{lem:mainB1B1}.

Recall that the goal of the proof is to analyze $\mu(T(p))$ for $p \in D_1.$ For notational simplicity, we set $\sigma = \sigma(p).$
Recall
\begin{equation}\label{eq:B1B1int}
\mu(T(p)) =\int_{u\in T(p)}  f(u) du ,
\end{equation}
where $f(u)$ is the density as defined by Definition \ref{def: den}.

Next set 
$$T'(p) =T(p) \cap  \left\{u \,:\,  u.x + u.y \le 1 + \delta - \sigma/2\right\}.$$
Note that $T'(p)$ is the upper right isosceles triangle with lower-left corner $p$ and  sidelengths $\sigma/2.$  Since $T'(p) \subseteq T(p)$,  $\mu(T'(p)) \le \mu(T(p)).$

Finally, let
\begin{eqnarray*}
U(p)  &=&  \{u \in B_1 \,:\, (u + \delta B_1) \cap T(p) \not=0\},\\
U'(p)  &=&  \{u \in B_1 \,:\,  T'(p) \subseteq u + \delta B_1 \}.
\end{eqnarray*}
Notice $U(p)$ is the preimage of $T(p)$ in $\Ball 1$ while $U'(p)$ is the the intersection of the preimages of all of the points in  $T'(p).$

From Eq.~\ref{eq: fdef}   in Lemma  \ref{lem: measure integral},
\begin{equation}
\forall u \in T(p),\quad  f(u)  = O\left( \frac{\Area(U(p)) } {\delta^2}  \right) .
\end{equation}
Plugging into Eq.~\ref{eq:B1B1int} then gives
\begin{equation}\label{B1B1UB}
\mu(T(p)) \le   \Area(T(p)) \max_{u \in T(p)} f(u)
= O\left(    \frac {\Area(U(p)) \sigma^2} {\delta^2}  \right).
\end{equation}
In the other direction we   similarly find that
$$\forall u \in T'(p),\quad  f(u) 
 = \Omega\left( \frac{\Area(U'(p)) } {\delta^2}  \right).$$
Again  plugging into Eq.~\ref{eq:B1B1int} 
gives 
\begin{equation}
\label{B1B1LB}\mu(T(p)) \ge \mu(T'(p)) \ge   \Area(T'(p)) \min_{u \in T'(p)} f(u)
= \Omega\left(    \frac {\Area(U'(p)) \sigma^2} {\delta^2}  \right).
\end{equation}

We now work through the four cases in the Lemma, deriving the appropriate values of $\Area(U(p))$ and $\Area(U'(p))$ and plugging them into
Eqs.~\ref{B1B1UB} and \ref{B1B1LB} to complete the proof.

\par\noindent\underline{\bf Case 3: $ p \in M$.}
\\

\begin{figure}[t]
\begin{center}
\includegraphics[width=5cm]{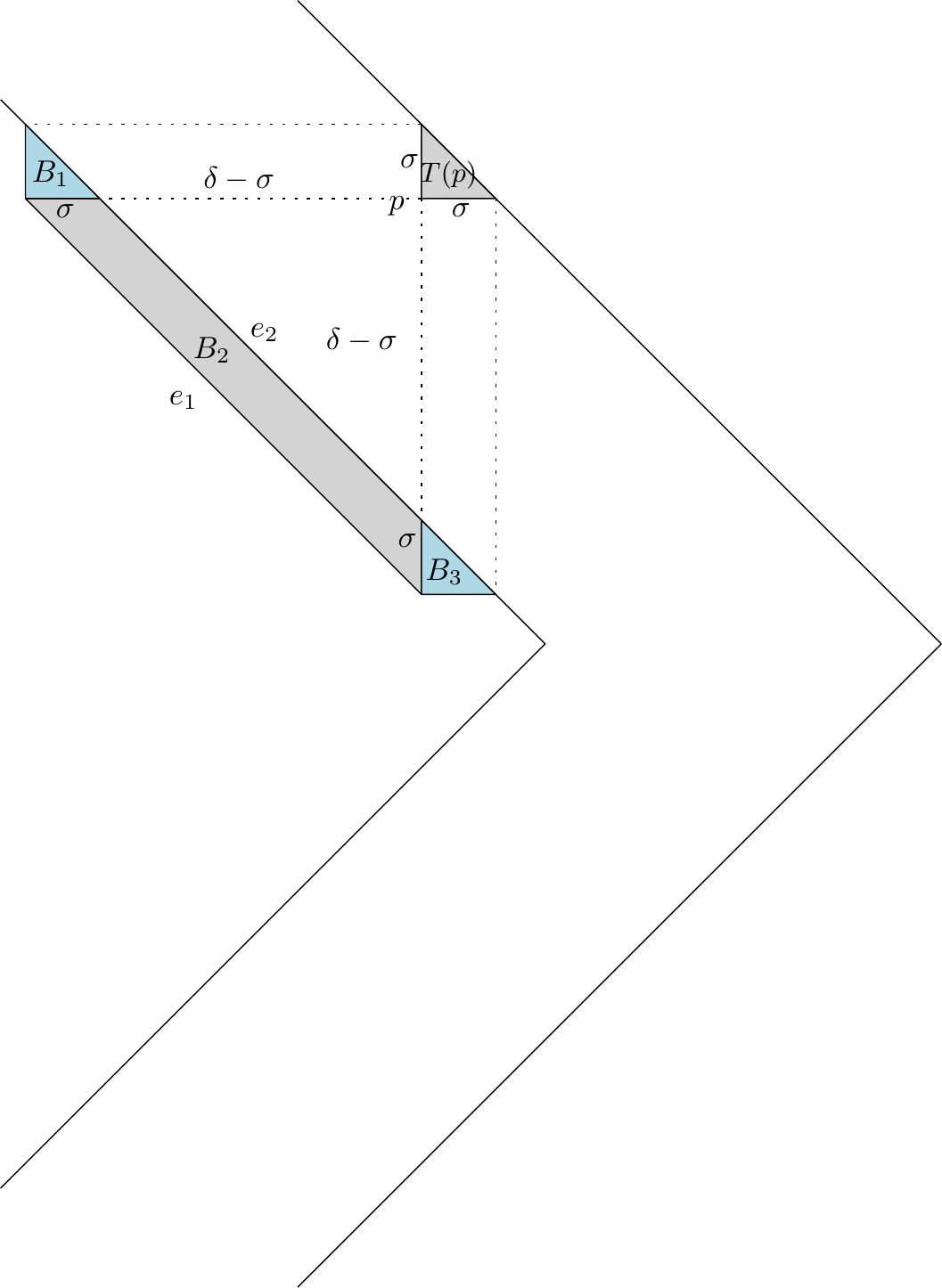}
\hspace*{.3in}
\includegraphics[width=5cm]{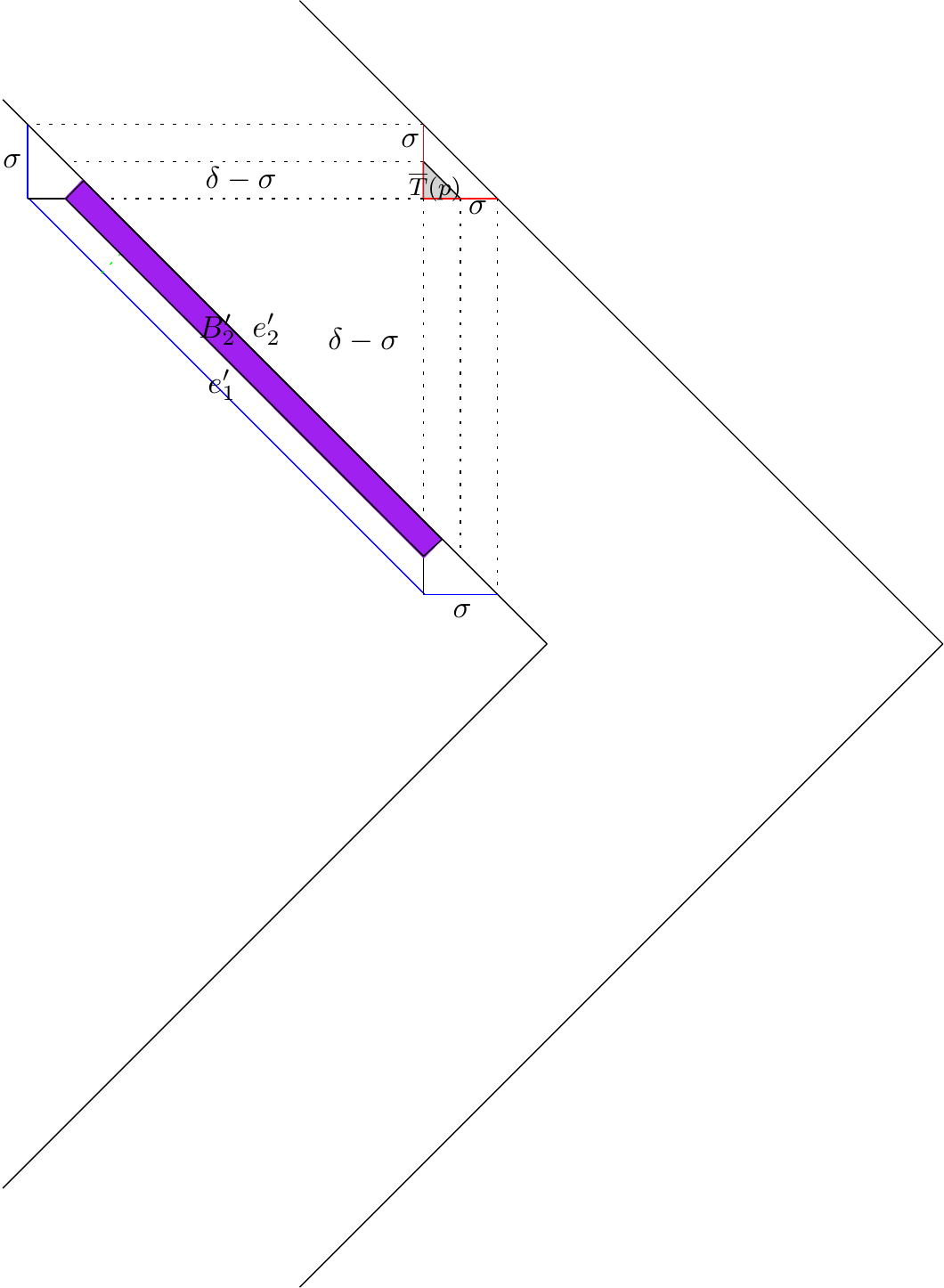}
\end{center}
\caption{Case 3 in which $T(p) \subseteq M$. } %
\lab{fig:Mid_case}
\rule{5.5in}{0.5pt}
\end{figure}

First assume that $T(p) \subset M,$ e.g.,  point $p_3$ in Figure \ref{fig:B1_Regionsd}.

 Consider the left diagram  in Figure \ref{fig:Mid_case}.  
Simple geometric arguments show that $U(p) = B_1 \cup B_2 \cup B_3$. Then  $\sigma \le \delta $ immediately yields
$$\Area(U(p)) = \Theta(\sigma \delta).$$

Consider the  right diagram  in Figure \ref{fig:Mid_case}.   This illustrates  the fact (again derivable from direct geometric arguments) that 
 $U'(p) = B'_2$ and 
$$\Area(U'(p)) = \Theta(\sigma \delta).$$
Plugging these formulas into Eqs.~\ref{B1B1UB} and \ref{B1B1LB} we get  
$$\mu(T(p)) = \Theta\left(    \frac {\sigma^3} \delta  \right).$$

  \begin{figure}[t]
\begin{center}
\includegraphics[width=5cm]{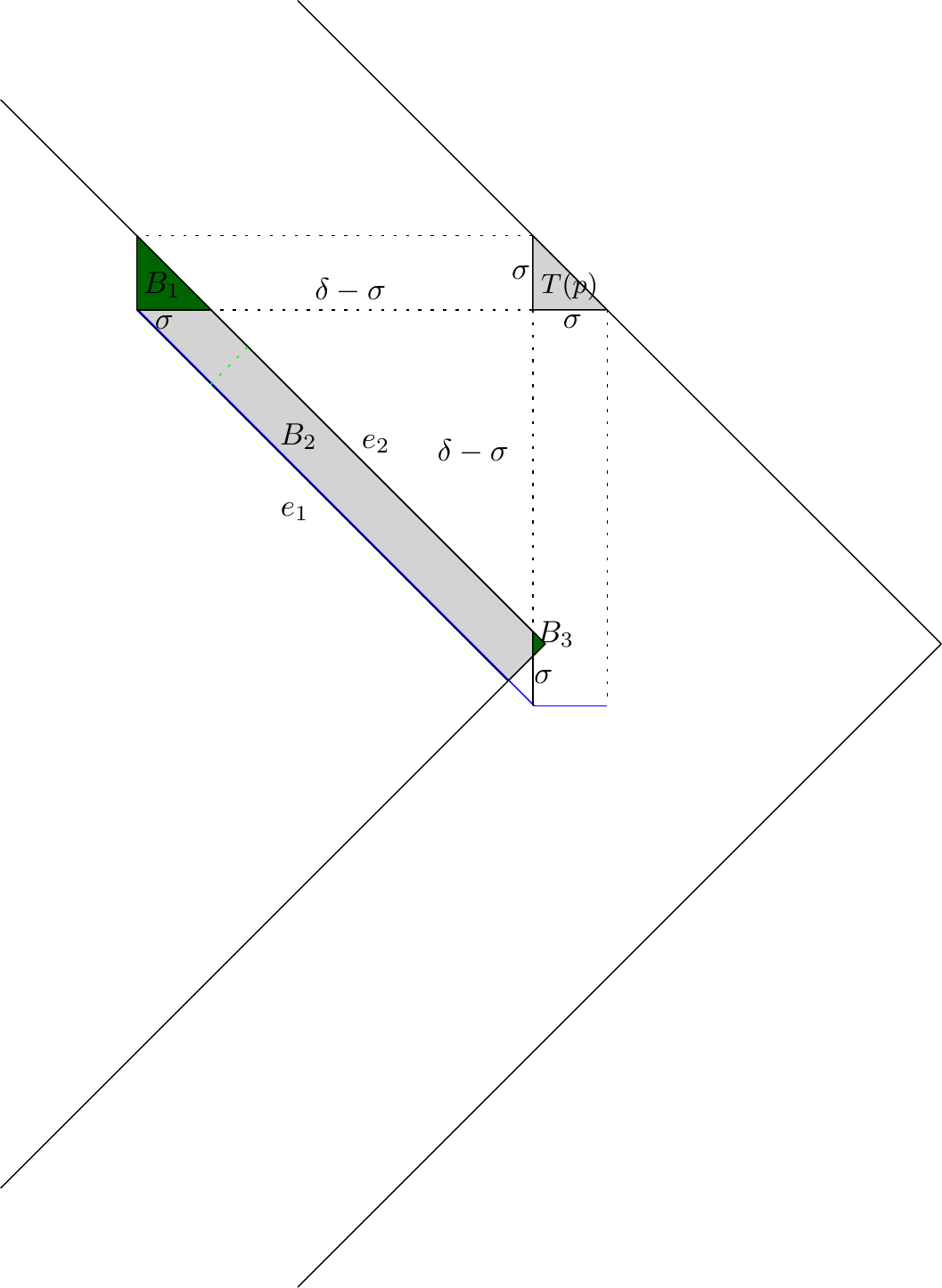}
\hspace*{.3in}
\includegraphics[width=5cm]{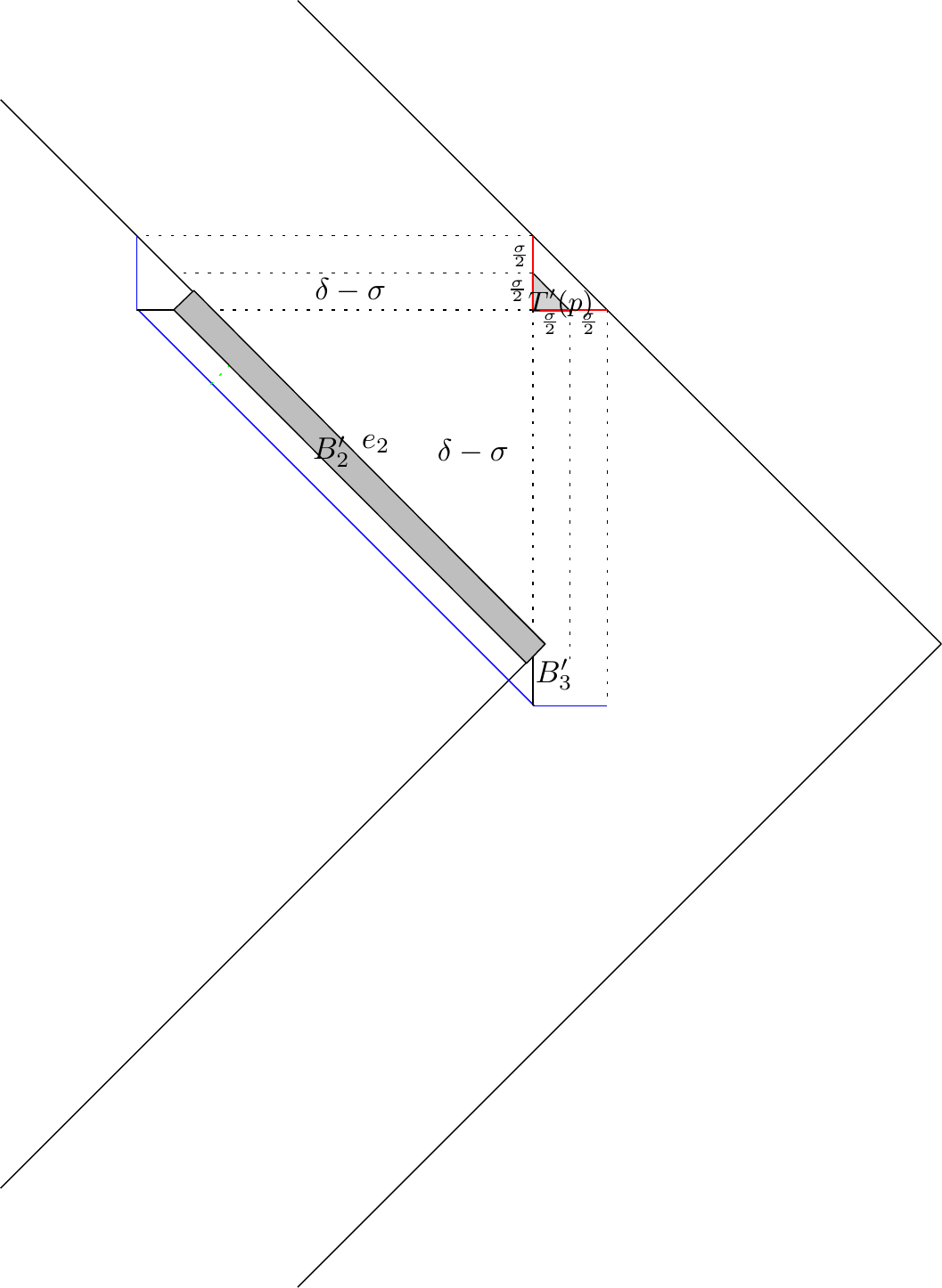}
\end{center}
\caption{Case 3 in which $T(p)\not \subseteq M$. } %
\lab{fig:Mid_case_2}
\rule{5.5in}{0.5pt}
\end{figure}

Assume that  $p \in M$ but $T(p) \not \subseteq M$, e.g., 
  points $p_2, p_4$ in Figure \ref{fig:B1_Regionsd}.
   Again,  simple geometric calculations (see Figure \ref{fig:Mid_case_2})  show that
  $$\Area(U(p)) = \Theta(\sigma \delta),
  \quad\mbox{and}\quad
  \Area(U'(p)) = \Theta(\sigma \delta).$$
(The intuition is that the areas are at least half the size of the preimages in the first subcase).  
Using again Eqs.~\ref{B1B1UB} and \ref{B1B1LB} we get 
$$\mu(T(p)) = \Theta\left(    \frac {\sigma^3} \delta  \right).$$

\par\noindent\underline{\bf Case 2:  $ p \in B.$}\\
  \begin{figure}[t]
\begin{center}
\includegraphics[width=5cm]{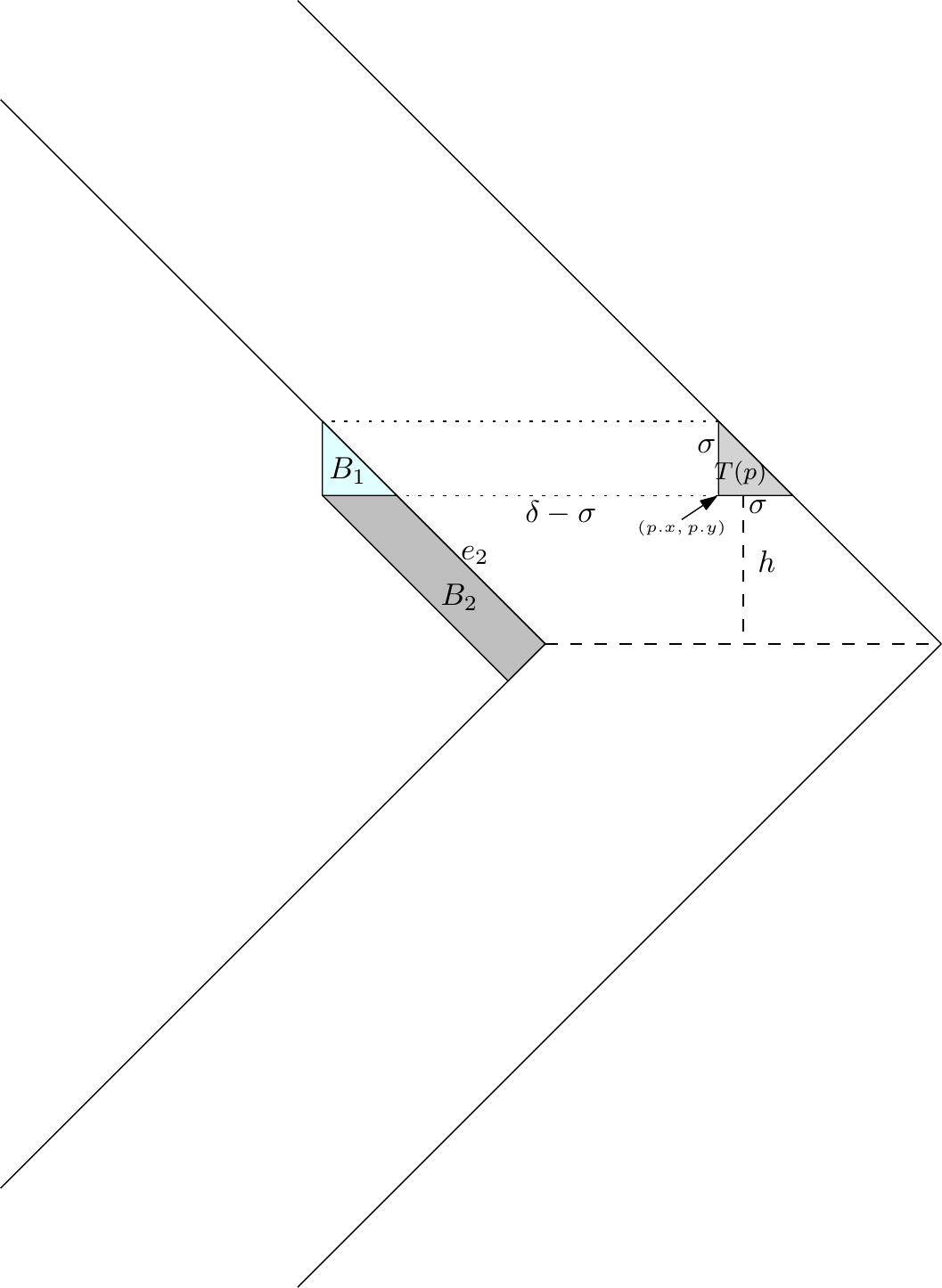}
\hspace*{.3in}
\includegraphics[width=5cm]{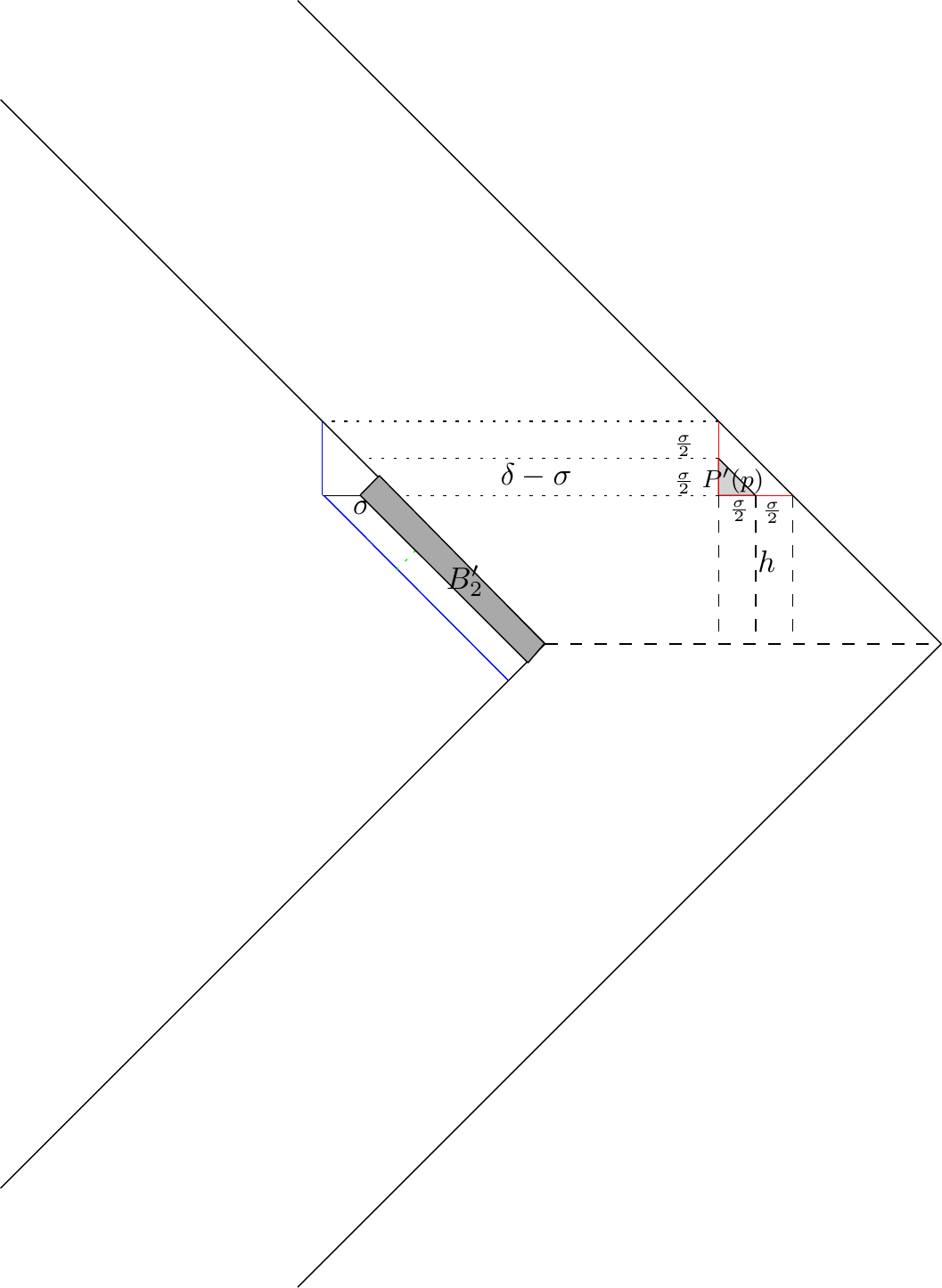}
\end{center}
\caption{Case 2} %
\label{fig:Bottom case v2}
\rule{5.5in}{0.5pt}
\end{figure}
An example is $p_5$   in Fig.~\ref{fig:B1_Regionsd}.

  In this case,  $1 < p.x$.    Set $h= p.y.$ 
  Note that,  by definition,   $h + \sigma \le \delta.$   

next note that, in Fig.~\ref{fig:Bottom case v2} , $U(p) = B_1 \cup B_2$ which,  by calculation satisfies
$$\Area(U(p)) = \Theta((h + \sigma)  \sigma).$$  
Similarly,  in the other direction, similar geometric arguments also  show
$$\Area(U('p)) = \Theta((h + \sigma)  \sigma).$$  
Plugging these formulas into Eqs.~\ref{B1B1UB} and \ref{B1B1LB}  
 and using the fact that $h+\sigma = \Delta_x(p)$ immediately gives
$$\mu(T(p) = \Theta\left(   \Delta_x(p) \frac {  \sigma^3} {\delta^2}  \right).$$

\par\noindent\underline{\bf Case 1: $p \in U.$} \\
An example is $p_1$  in Figure \ref{fig:B1_Regionsd}.
The proof for this case is totally symmetric to that of Case 2 and is therefore omitted.
  
  \medskip
  
  \par\noindent\underline{\bf Case 4:  $p \in I$.} \\
  
    \begin{figure}[t]
\begin{center}
\includegraphics[width=9cm]{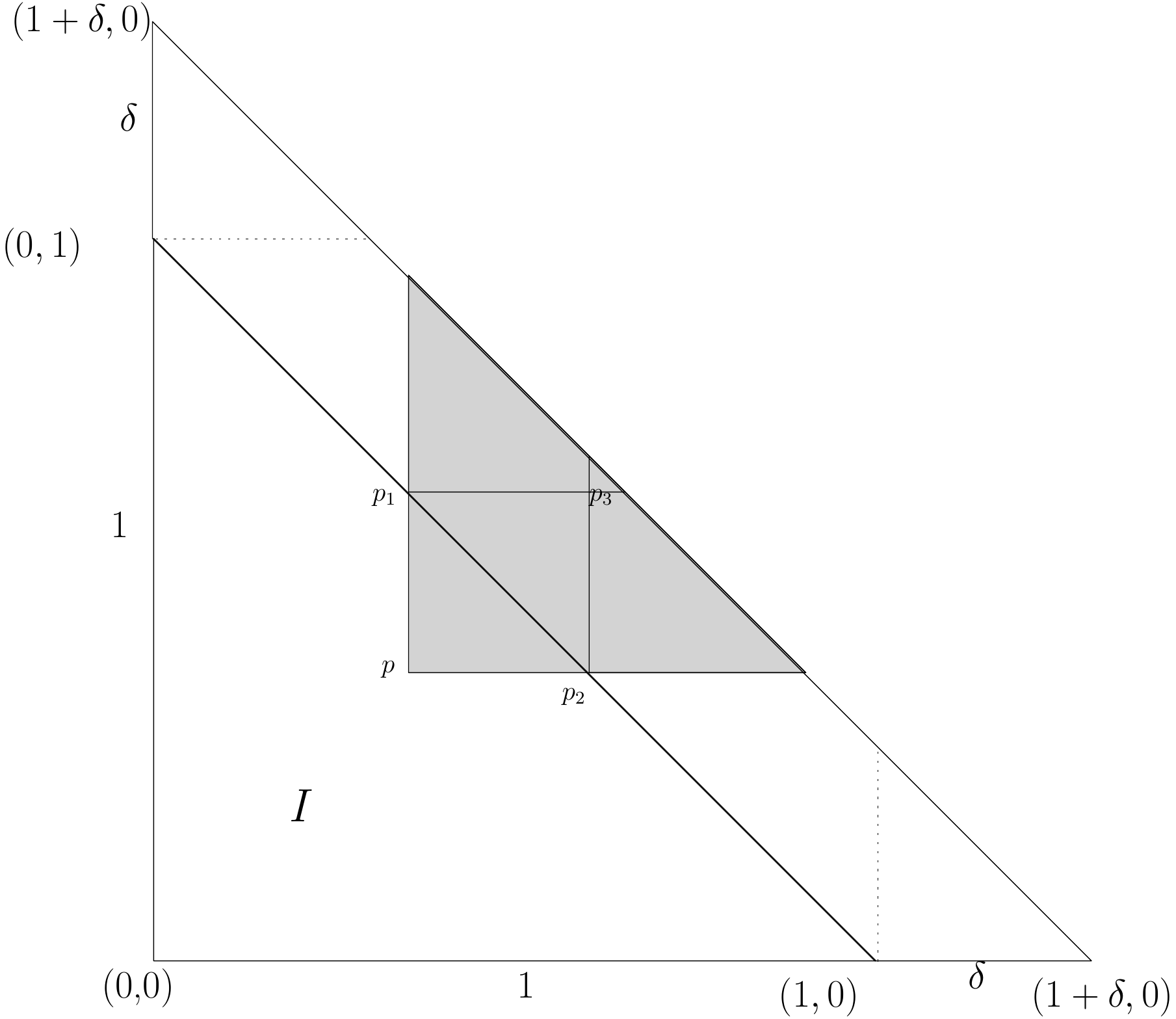}
\end{center}
\caption{Case 4} %
\label{fig:case v4}
\rule{5.5in}{0.5pt}
\end{figure}
An example is $p_8$   in Fig.~\ref{fig:B1_Regionsd}.

Set $\Delta =   1 - (p.x+p.y)   = \sigma(p) - \delta.$
We will prove that $\mu(T(p)) = \Theta\left( \Delta^2+ \delta^2 \right).$ Since this is 
$\Delta^2+ \delta^2=\Theta(\sigma^2(p))$ this completes the proof.

Given $p \in I$ now set
$$ p_1 =(p.x+ \Delta, p.y),\quad 
p_2 = (p.x, p.y+ \Delta),\quad
p_3 = (p.x+ \Delta, p.y+ \Delta).
$$ 
See Fig.~\ref{fig:case v4}.
Note that while $p_3$ is not necssarily in $\Ball 1 + \delta \Ball 1$, 
$p_1,p_2$ are on the boundary of $\Ball 1$ so $\sigma(p_1) = \sigma_p(2) = \delta$ and from Case 3
$$\mu(T(p_1)) = \mu(T(p_2)) =\Theta( \delta^2).$$
Let $T$ be the triangle with vertices $p,p_1,p_2$ and   $T'$  the triangle with vertices $p_1,p_2,p_3.$
Since $T \subset B_1$, from  Lemma \ref{lem:easy mu}(d) (with $\kappa =1$)
$$\mu(T) = \Theta(\Area(T)) = \Theta(\Delta^2).$$
From Lemma \ref{lem:easy mu}(c),
$$\mu(T') = O(\Area(T')) = O(\delta^2).$$

Finally note that $T(p) \subseteq  T(p_1) \cup T(p_2) \cup T \cup T'$ so
$$ \mu(T(p)) \le  \mu(T(p_1)) + \mu(T(p_2)) + \mu(T) + \mu(T') = \Theta ( \Delta^2 + \delta^2).$$
From the other direction,  $T$ and $T(p_1)$ are disjoint  with $T \cup T(p_1) \subset T(p)$ so 
$$\mu(T(p)) \ge  \mu(T(p_1)) + \mu(T) = \Theta(\Delta^2 + \delta^2).$$
Thus 
$$\mu(T(p))   =  \Theta(\Delta^2 + \delta^2).$$
\end{proof}
\section{Analysis of $\Ballpq 2 2$}
\label{sec: B2B2}
This section derives cell (iii)(c) in Theorem \ref{thm: main}, that is,  if $n$  points are chosen from $\bfd = \Ballpq 2 2$ and  $\frac 1  {\sqrt n} \le \delta \le 1$ then $\EMN = \Theta \left(  \frac {n^{2/7}}  { \delta^{3/7}}  \right)$.
Applying Lemma  \ref {lem: scaling} for   $1  \le \delta \le   \sqrt n $     yields cell (iii)(d), i.e., that   $\EMN =  \Theta \left(   \delta^{3/7}  n^{1/3} \right).$

As before, Corollary \ref {cor: Quadrants} states that 
$$\EMN =   \EXP {|Q_1 \cap \MAX(S_n)|}  + O(1)$$
so our analysis will be restricted to the upper-right  quadrant $Q_1$.
Our approach will be to
\begin{enumerate}
\item State a convenient expression for $\mu(u)$ (proof delayed until later).
\item Derive a lower bound using Lemma \ref{lem: lb} by defining an appropriate  pairwise disjoint collection of dominant region.
\item Derive an upper bound by partitioning $D$ into appropriate regions and applying the sweep Lemma.
\end{enumerate}

This section will need the following special definitions
\begin{Definition}
\lab{def:B2B2 Regionsa}
Let $D_1= (B_2+ \delta B_2) \cap Q_1$ be the support in the first quadrant.  The support restricted to the   the first and second octants is
$$D_{1,1} = D_1 \cap \{(x,y) \in \Re^2 \,:\,  y \le x\},\quad
D_{1,2} = D_1 \cap \{(x,y) \in \Re^2 \,:\,  x \le y\}.\
$$
For all $v \in D$ define
$$\sigma(v) = 1 + \delta - ||v||,\quad
\theta(v) = \arctan \left( \frac {v.y} {v.x}\right).
$$
$\sigma(v)$ is the distance from $v$ to the boundary of $D_1$ and $\theta(v)$ is the angle formed by the $x$ axis by the line connected the origin to $v$.
\end{Definition}

 By construction, the induced measure is radially symmetric, i.e.,  $f(v)= g(||v||),$ for some function $g.$   More specifically, 

\begin{Lemma}
\lab{lem:B2B2 measure}
Let $\bfd = \Ballpq 2 2$.  
The density of $\bfd$ satisfies
$$\forall v \in D_1,\quad   f(v) =
\left\{
	\begin{array} {cc}
 \Theta  \left(  \left(  \frac  {\sigma(v)} {\delta}    \right)^{3/2}\right)   & \mbox{ if $\sigma(v) \le \delta$},\\
		\Theta (1)  & \mbox{ if $\delta \le \sigma(v) \le 1 + \delta$}.
	\end{array}
\right.
$$
\end{Lemma}
The proof of the Lemma is deferred to the end of this section.

Given the above, the lower bound is easy to derive:
\begin{Lemma}
\label{lem: B2B2LB}
Let $S_n$ be $n$ points chosen from the distribution $\bfd = \Ballpq 2 2$ with $\frac 1 {\sqrt n} \le \delta \le 1$. Then
$$\EMN = \Omega \left(  \frac {n^{2/7}}  { \delta^{3/7}}  \right).$$ 
\end{Lemma}

\begin{proof} 
\begin{figure}[t]
\begin{center}
\includegraphics[width=10cm]{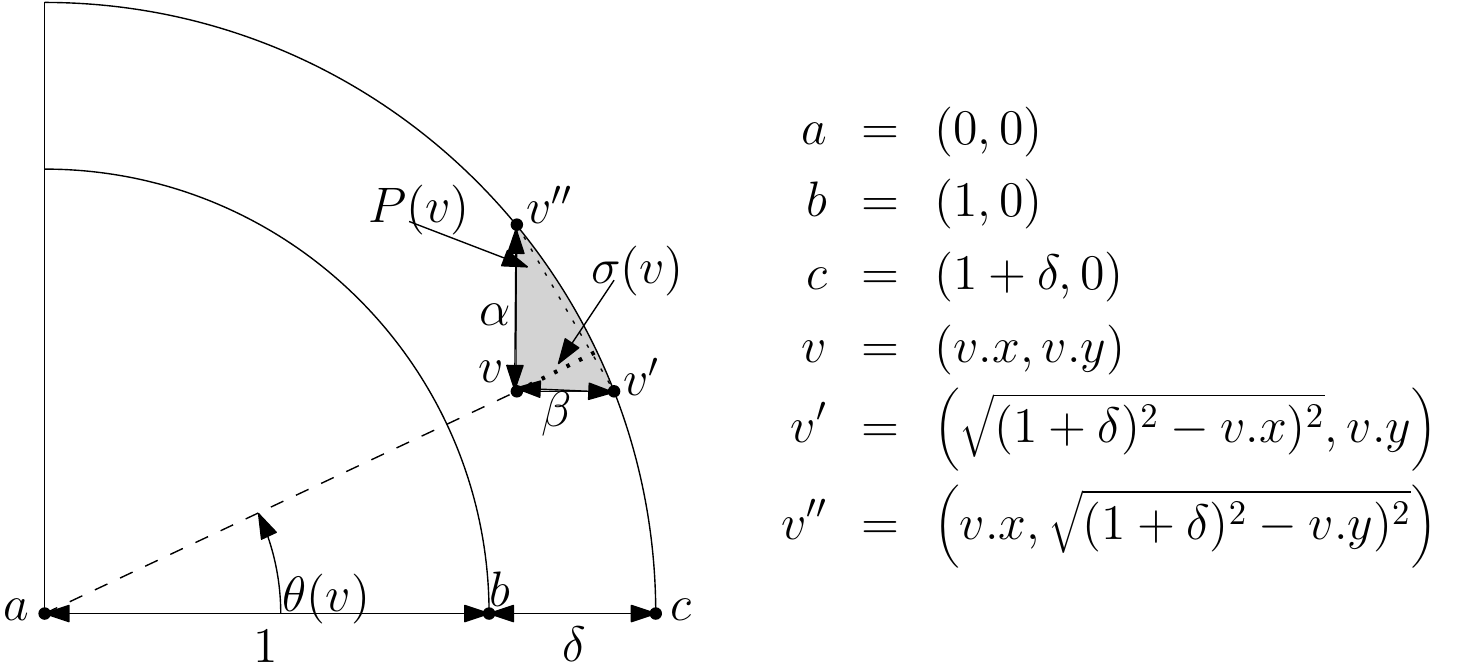}
\end{center}
\caption{Illustration of  the  proof of  Lemma \ref{lem: B2B2LB}.  } %
\label{fig:B2 B2 Pp}
\rule{5.5in}{0.5pt}
\end{figure}

Let $v \in D_1$. 
The remainder of the proof assumes  that
(i) $\sigma(v) \le \delta$ and  (ii) $\theta(v) \in \left[\frac \pi 6, \frac \pi 3\right]$  (this could be replaced by any arbitrary fixed angular range bounded away from $0$ and $\pi/2$ radians).

Refer to  Figure \ref {fig:B2 B2 Pp} for the below.
Recall $P(v)$ as 
introduced in Definition  \ref {def: PDdef}. 
  Let $v'$, $v''$ be the two other vertices of $P(v)$ with $v.x = v''.x$, 
$v.y =  v'.y$,  $\alpha(v) = v''.y - v.y$ and $ \beta(v)  = v'.x - v.x$.  
Note that  $\pi/6 \le \theta(v) \le \pi/3$   immediately implies
$\beta(v) = \Theta(\sigma(v))$ and $\alpha(v) = \Theta(\sigma(v)).$ Let $c > 0$ be such that $\beta(v) \le c \sigma(v)$ for all such $v.$

Since $P(v)$ contains the triangle with vertices $v,v',v''$ and is constained in the rectangle that has those three points and $(v'.x,v''.y)$ as corners,
 $ \frac {\alpha(v) \beta(v) } 2 \le \Area(P(v)) \le  \alpha(v) \beta(v)$.  Thus
$$\Area(P(v)) = \left( {\frac  {\sigma(v) }\delta  }\right)^{3/2}.$$


Let $u \in P(v).$  Note that $\sigma(u) \le \sigma (v)$, so    Lemma \ref{lem:B2B2 measure}  immediately implies $f(u) = O(f(v))$ and 
\begin{eqnarray*}
\mu(P(v)) = \int_{u \in P(v)} f(u) du 
&\le&  \Area(P(v)) \max_{u \in P(v)} f(u) \\
&=& O\left( \sigma^2(v)
 \left( {\frac  {\sigma(v) }\delta  }\right)^{3/2}
 \right) 
=\Theta
\left(
\frac  {\sigma^{7/2}(v)}  {\delta^{3/2}}
\right).
\end{eqnarray*}

Set 
$$\bar P(v) = P(v) \cap\left \{u \in \Re^2 \,:\,  \sigma(u) \ge \frac 1 4 \sigma(v) \right\}.$$

Using basic geometric arguments it is straightforward that  
$$\Area(\bar P(v)) =  \Theta(\Area(P(v)) = \Theta(\sigma^2(v)).$$

From  Lemma \ref{lem:B2B2 measure}, for every $u\in \bar P(v),$
$f(v') = \Theta(f(v))$ and thus

\begin{eqnarray*}
\mu(P(v)) = \int_{u \in P(v)} f(u) du 
&\ge&  \int_{u \in \bar P(v)} f(u) du \\
&\ge&  \Area(\bar P(v)) \min_{u \in \bar P(v)} f(u) \\
&=& \Theta \left( \sigma^2(v)
 \left( {\frac  {\sigma(v) }\delta  }\right)^{3/2}
 \right) 
=\Theta
\left(
\frac  {\sigma^{7/2}(v)}  {\delta^{3/2}}
\right).
\end{eqnarray*}

We have thus shown that, for all $v$ satisfying the two conditions, 
$$\mu(P(v)) = \Theta
\left(
\frac  {\sigma^{7/2}(v)}  {\delta^{3/2}}
\right).
$$

Set $\sigma = \frac  {\delta^{3/7}} {n^{2/7}}.$  Note that, for $\frac 1 {\sqrt n} \le \delta ,$  this implies  $\sigma \le \delta$ and $v$ satisfies the first condition.  Thus, if  $\theta(v) \in \left[\frac \pi 6, \frac \pi 3\right]$ and     $\sigma(v) = \sigma$ then 
\begin{equation}
\lab{eq:B2B2 1n prob}
\mu(P(v)) = \Theta
\left(
\frac  {\sigma^{7/2}(v)}  {\delta^{3/2}}
\right)
= \Theta\left (\frac 1 n\right).
\end{equation}
We  construct a collection of such $v$.
Set $x_\ell = (1 + \delta) \cos (\pi/3)$ and  $x_r = (1 + \delta) \cos (\pi/6)$.  Note that $x_r - x_\ell= \Theta(1).$

Let $m = \lfloor \frac  {x_r - x_\ell} {c \sigma}\rfloor$.  Note that since $x_r - x_\ell= \Theta(1), $  $m = \Theta\left(\frac 1 \sigma \right) = \Theta\left(     \frac   {n^{2/7}}  {\delta^{3/7}}  \right).$
Now set, $x_0=x_l$ and, for  $i=1,\ldots,m,$
\begin{eqnarray*}
x_i &=& x_l + i \sigma,\\
v_i &=&  (x_i,  \, \sqrt {(1+ \delta- \sigma)^2 - (x_{i})^2}).
\end{eqnarray*}

By construction,   $\sigma(v_i) = \sigma$,    $\theta(v_i) \in \left[\frac \pi 6, \frac \pi 3\right]$ and thus
$\mu(P(v_i)= \Theta(1/n)$.  Furthermore, each $P(v_i)$ is a dominant region and, 
since by construction,  $\forall i,$ $x_i + \beta(v_i)  \le x_{i+1},$ 
the $P(v_i)$ are pairwise disjoint.  Then, Lemma \ref{lem: lb} then immediately proves 
 the required 
$$\EXP{M_n}= \Omega(t) = \Omega \left(     \frac   {n^{2/7}}  {\delta^{3/7}}  \right).$$
\end{proof}

The upper bound proof will require an additional definition and lemma.

 \begin{figure}[t]
\begin{center}
\includegraphics[width=7cm]{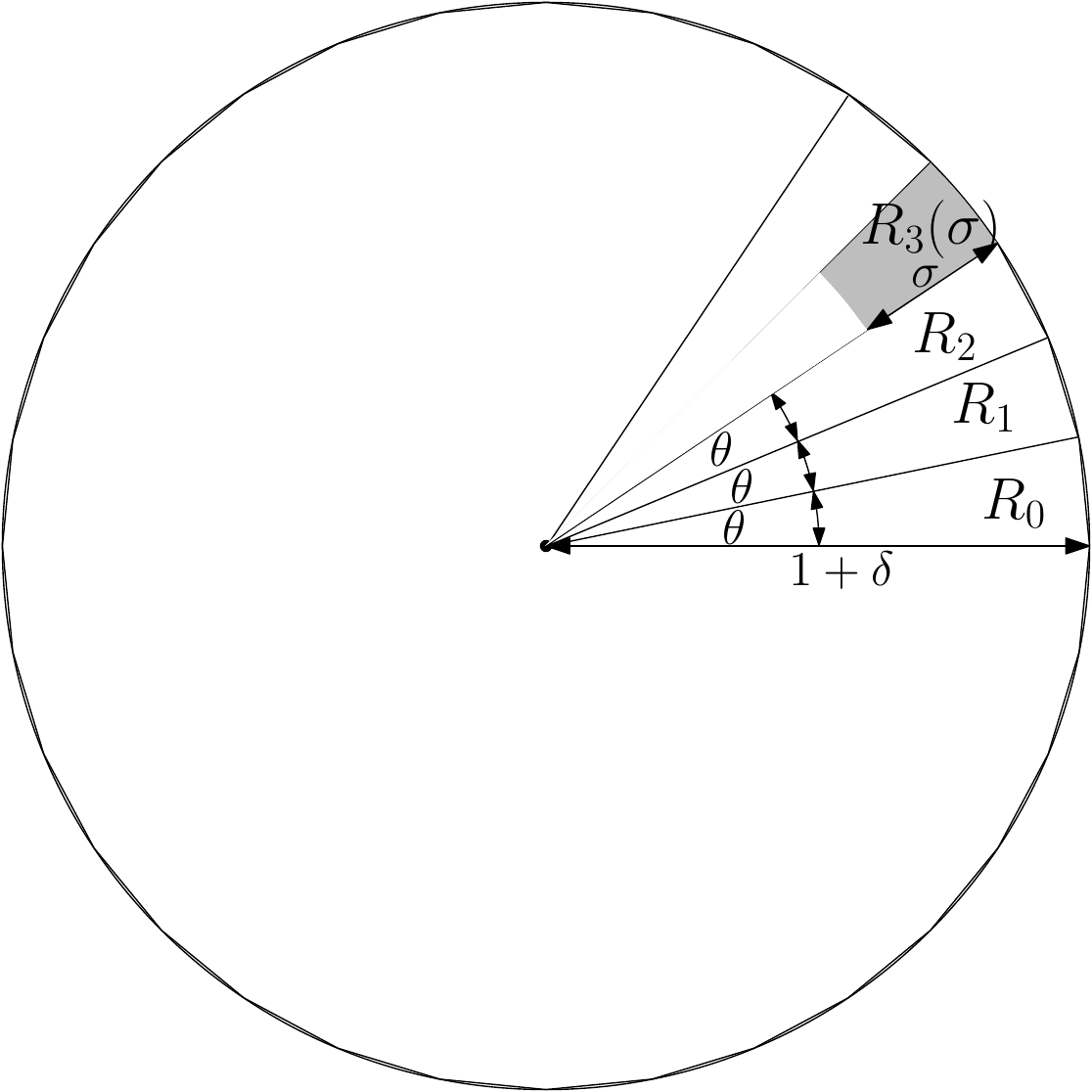}
\end{center}
\caption{Illustration of Definition \ref {def: B2B2Raddef}.  $D=B_2+ \delta B_2$  is partitioned into $t$ radially symmetric sectors $R_i.$  $R_i(\sigma)$ are the points in $R_i$ within $\sigma$ of the boundary of $D$. } %
\label{fig:B2 B2 UB V}
\rule{5.5in}{0.5pt}
\end{figure}

\begin{Definition}
\label{def: B2B2Raddef}
Let $t$ be a positive integer, 
$\theta = \frac {2 \pi} t$ and
$\theta_i = i \theta,$  $i =0,1,\ldots , t-1.$  Now partition  $D=B2 + \delta B_2$ into $t$ sectors, $i =0,1,\ldots, t-1,$
$$R_i =  D \cap  \{ v \in \Re^2 :\,  \theta_i \le \theta(v) \le \theta_{i+1}\},$$
and further define
$$R_i(\sigma) = \{v \in R_i \,:\, 1 + \delta -\sigma \le ||v|| \}.$$
\end{Definition}

\begin{Lemma}
\lab{lem:muri}
$$\mu(R_i(\sigma) )
=
\left\{
\begin{array}{cc}
\Theta
 \left( \frac {\sigma^{5/2}} {t \delta^{3/2}}   \right) & \mbox{ if $0 \le \sigma \le \delta$}, \\[0.1in]
\Theta\left( \frac \sigma t \right)  & \mbox{ if $\delta  \le \sigma \le 1 +  \delta$}.
\end{array}
\right.
$$
\end{Lemma}

\begin{proof}

Change  into polar coordinates and write $v= (x,y) =  (||v|| \cos \gamma,||v||\sin \gamma)$. Recall  from 
Lemma \ref {lem:B2B2 measure} that  $f(v) = g(1+\delta - ||v||)$ so we can integrate to get

\begin{eqnarray*}
 \mu(R_i(\sigma))    &=& \int\int_{(x,y)  \in R_i(\sigma)}  f((x,y))\, dx dy \\
                         &=& \int_{\gamma =\theta_i}^{\theta_{i+1}} \int_{r=1+\delta-\sigma}^{1 + \delta}  r\, g(1 + \delta -r)\,dr \\
                           &=& \int_{\gamma =\theta_i}^{\theta_{i+1}} \int_{r=0}^\sigma (1+\delta -\sigma)\, g(r)\, dr\\
                          &=& \frac {2 \pi} t \int_{r=0}^\sigma (1+\delta -r)\, g(r)\, d r.
\end{eqnarray*}
For all 
$r \le \delta,$   $  g(r) = \Theta  \left(  \left(  \frac  {r } {\delta}    \right)^{3/2}\right)$ so

$$\mu(R_i(\sigma))  = \Theta
 \left(      \frac  1 t     
\int_{r=0} ^{\sigma} \frac {r^{3/2}} {\delta^{3/2}} d r
\right)
=
\Theta
 \left( \frac {\sigma^{5/2}} {t \delta^{3/2}}
\right).
$$

Recall from  Lemma \ref {lem:B2B2 measure}  that, for   $ \delta \le \sigma \le 1 + \delta$,  $f(v) = \Theta(1),$ so 
$$\mu(R_i(\sigma))  = \Theta
 \left(      \frac  1 t     
\int_{r=0} ^{\delta} \frac {r^{3/2}} {\delta^{3/2}} d r
+  \frac  1 t   \int_{r=\delta} ^{\sigma}  dr
\right)
 = \Theta \left(     \frac  t  \delta\right)  + \Theta\left( \frac{\sigma - \delta} t     \right)
= \Theta\left( \frac  \sigma t\right).
$$
\end{proof}

As before, the upper bound is more technical.
\begin{Lemma}
\label{lem: B2B2UB}
Let $S_n$ be $n$ points chosen from the distribution $\bfd = \Ballpq 1 1$ with $\frac 1 {\sqrt n} \le \delta \le 1$. Then
$$\EMN = O \left(  \frac {n^{2/7}}  { \delta^{3/7}}  \right).$$
\end{Lemma}

\begin{proof}
Chose integer $t$ in Definition \ref{def: B2B2Raddef} so that
$$t = \min \left\{
t' \ge  \frac {n^{2/7}}  { \delta^{3/7}}
\quad\mbox{and}\quad
\mbox{$t'$ is a multiple of 8}\right\}.
$$
Recall from Corollary \ref{cor: Quadrants} that 
$$\EMN 
= \Theta
\left(
 \EXP{| \MAX(S_n) \cap D_{1,1}|}
 \right),$$
where $D_{1,1} = D \cap O_1$ is the support of $\bfd$ in the first octant.

Next note that because $t$ is a multiple of 8, $D_{1,1} = \bigcup_{i=1}^{t/8} R_i.$  Thus
$$\EXP{| \MAX(S_n) \cap D_{1,1}|} = \sum_{i=1}^{t/8} \EXP{| \MAX(S_n) \cap R_i|}.$$
We now  use the Sweep lemma to show  that for all $i,$ $\EXP{| \MAX(S_n) \cap R_i|} =O(1)$. The proof will then follow.

 \begin{figure}[t]
\begin{center}
\includegraphics[width=9cm]{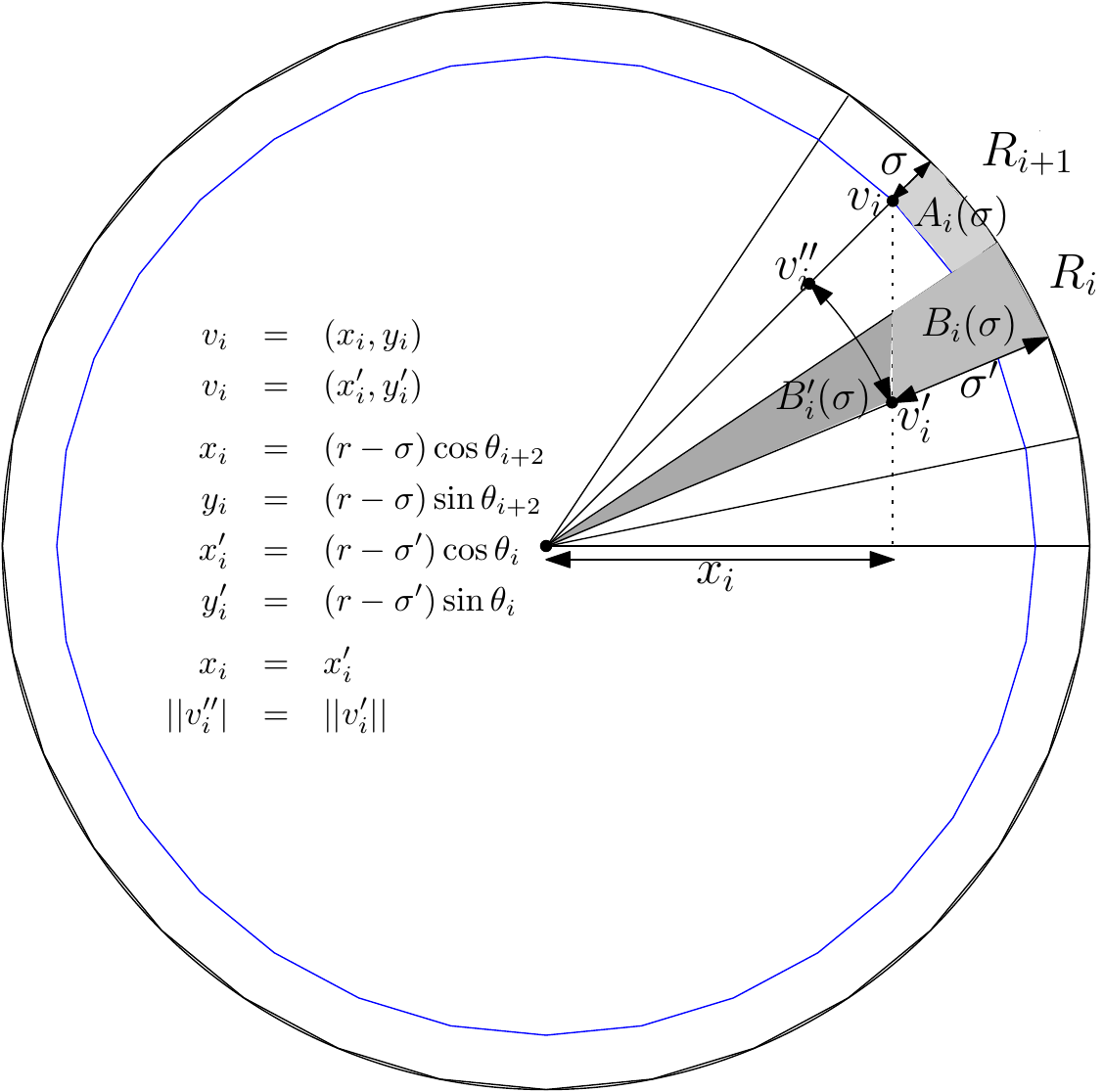}
\end{center}
\caption{ Illustration of the proof of Lemma \ref{lem: B2B2UB}.} %
\label{fig:B2 B2 UB VI}
\rule{5.5in}{0.5pt}
\end{figure}

Consider Figure  \ref{fig:B2 B2 UB VI}.  Set $r = 1+ \delta$ to be the full radius of  $B_2 + \delta B_2$.
Let $A_i(\sigma) = R_{i+1}(\sigma).$  Set 
$$v_i =( (r - \sigma) \cos \theta_{i+2},\,  (r - \sigma) \sin \theta_{i+2}).$$
This is the leftmost point of $A_i(\sigma)$ on the boundary line between $R_{i+2}$ and $R_{i+3}.$
Now drop a vertical line from $v_i$ to the $x$-axis and let $v'_i$ be the point at which it intersects  the boundary between $R_i$ and $R_{i-1}.$ Let $\sigma'= \sigma(v_i').$    By construction

$$(r - \sigma) \cos \theta_{i+2} = v.x = v'.x = (r -\sigma') \cos \theta_{i}.$$
But
\begin{eqnarray*}
\cos \theta_{i+2} &=&  \cos \left(\theta_i + \frac 2 t\right)\\
			    &=&  \cos \theta_i \cos (2/t) + \sin \theta_i \sin (2/t)\\
                               &=&  (1 + O(1/t) \cos \theta_i  + O(1/t)\\
                               & =&  \cos \theta_i    + O(1/t).
\end{eqnarray*}
Because the $R_i$ are in the first octant, $\theta_i$ is bounded away from $\pi/2$  so $\cos \theta_i$ is bounded away from $0.$  Thus, 
$\sigma' - \sigma = O(1/t).$

Define 
$$B'_i (\sigma)=  \{v \in R_i \,:\,  v.x \le v_i.x\},\quad    B_i(\sigma) = R_i \setminus B'_i(\sigma).$$
By construction
\begin{itemize}
\item  $\forall \sigma,$ every point in $A_i(\sigma)$ dominates every point in $B'_i(\sigma).$
\item $\forall \sigma,$  $B_i(\sigma) \subset R_i(\sigma')$ \quad $\Rightarrow$ \quad  $\mu(B_i(\sigma)) \le  \mu(R_i(\sigma'))$.
\item  $\bar \sigma = \frac 1 t $   \quad $\Rightarrow$ \quad $\bar\sigma' = \frac 1 t + O(\frac 1 t ) =\Theta(\bar \sigma)$.
\item  $\frac 1 {\sqrt n} \le \delta \quad \Rightarrow\quad  \frac 1 {n^2} \le \delta^4  \quad \Rightarrow  \quad\frac {\delta^3} {n^2} \le \delta^7
\quad \Rightarrow  \quad \bar  \sigma  = \frac 1 t \le \frac {\delta^{3/7} }{n^{2/7}} \le \delta.$
$$\Rightarrow \mu(B_i(\bar\sigma)) \le  \mu(R_i(\bar\sigma')) = \Theta \left( \frac {\bar \sigma^{5/2}} {t \delta^{3/2}}\right) = \Theta(1/n).$$
\end{itemize}
  Then 
  $$\EXP{|MAX(S_n) \cap B_i(\bar \sigma)|} \le \EXP{|S_n \cap B_i(\bar \sigma)|}
  = O(n \mu(B_i(\bar \sigma))) = O(1).$$
 
 Since   $R_i = B'_i(\sigma) \cup B_i(\sigma)$,  this implies
\begin{eqnarray*}
\EXP{|MAX(S_n) \cap R_i|}  &\le&  \EXP{|MAX(S_n) \cap B'_i(\bar \sigma)|}   + \EXP{|MAX(S_n) \cap B_i(\bar \sigma)|}  \\
&=&      \EXP{|MAX(S_n) \cap B'_i(\bar \sigma)|}   + O(1).		
\end{eqnarray*}

Thus,  to show that $\EXP{|MAX(S_n) \cap R_i|} =O(1)$ it suffices to prove that
$\EXP{|MAX(S_n) \cap B'_i(\bar \sigma)|} = O(1),$  which we will now do via the sweep lemma,  using $\sigma$ as the sweep parameter.

Set 
$$
\begin{array}{ccccccl}
A &=& R_{i+1}, &                   \quad & A(\sigma) &=& A_i(\sigma),\\
B &=& B'_{i}({\bar \sigma}), & \quad & B(\sigma) &=& B_i(\sigma) \cap B'_i(\bar \sigma).
\end{array}
$$
Note that,  by the previous discussion
\begin{itemize}
\item  Every point in $B \setminus B(\sigma) $ is dominated by every point in $A_i(\sigma).$
\item If   $\sigma \le \bar \sigma$ $\Rightarrow $  $B(\sigma) = \emptyset$.
\item If $\sigma > \bar \sigma $$ \Rightarrow$  $B(\sigma)  \subset B_i(\sigma) \subset R_i(\sigma')$.
\end{itemize}
Thus,  if  $\sigma \le \bar \sigma$ then
$$ \mu(B(\sigma)) = 0 = O(\mu(A(\sigma)),$$
while if $ \sigma > \bar \sigma$ then
\begin{eqnarray*}
\mu(B(\sigma))  &= & O (\mu(R_i(\sigma'))) \\
                             &= & O (\mu(R_{i+1}(\sigma'))) \\
                           &=  & O (\mu(R_{i+1}(\sigma + O(\bar \sigma) )))  \\
                           &=  & O (\mu(R_{i+1}(\sigma))), 
  \end{eqnarray*}
where the last inequality comes from        plugging in the values from Lemma \ref{lem:muri}.  
This explicitly satisfies the conditions of the Sweep Lemma 
and thus,
$\EXP{|MAX(S_n) \cap B'_i(\bar \sigma)|} = O(1),$  
so the proof of the upper bound  is completed.
\end{proof}

\begin{proof} of Lemma \ref {lem:B2B2 measure}.

   \begin{figure}[t]
\begin{center}
\includegraphics[width=9cm]{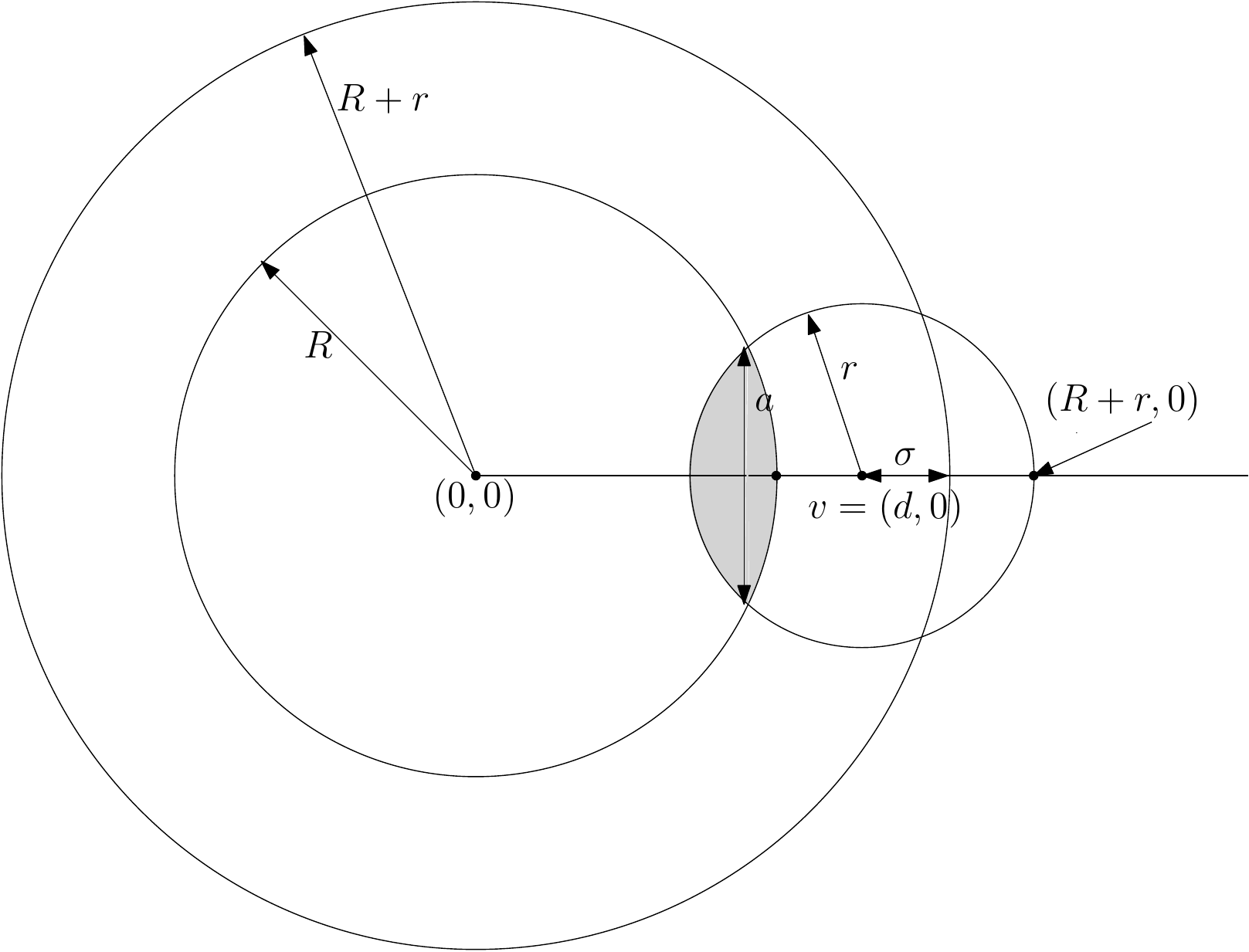}
\end{center}
\caption{Illustration of proof of Lemma \ref{lem:B2B2 measure}.  $R=1,$ $r= \delta$. The gray area is $P'(v)$, the preimage of $v$ in $B_2,$  i.e., the intersection
of $B_2$ and $B_2(v,\delta).$
} %
\label{fig:B2 B2 Measure}
\rule{5.5in}{0.5pt}
\end{figure}
Set $\sigma = \sigma(v).$ The radial symmetry of $\bfd$ follows from the radial symmetry of $\Ball 2$.  Because of the radial symmetry we may assume that $v =(1+\delta-\sigma,0),$ i.e., $v$ lies on the non-negative $x$ axis.  
First note that if $\sigma \ge \delta,$ then  $ v \in B_2$ so, by Lemma \ref{lem:easy mu}(b) (with $\kappa =1$)
$f(v) = \Theta(1).$ 
 If $\sigma > 1 + \delta,$  then $v \not\in B_2 + \delta B_2$.  The remainder of the proof therefore assumes  $0 \le \sigma < \delta.$

From Eq.~\ref{eq: fdef} in Lemma \ref{lem: measure integral},
\begin{equation}
\label{eq:B2B2mf}
f(v) = \Theta\left( \frac {\Area(P'(v))} {\delta^2} \right)
\end{equation}
where 
$$P'(v)  = B_2 \cap B_2(v,\delta),$$
i.e. the  preimage of $v$, is the intersection of two circular balls. 
The first has radius $R=1$ with center$ (0,0)$ and the second has radius $r=\delta$ with center $v=(d,0)$,  
where $d=1+\delta-\sigma.$ 
The intersection region is  shaped like an asymmetric lens (Fig.~\ref{fig:B2 B2 Measure}).
The width of this lens is  $R - (d-r) = \sigma.$
Let $a$ denote the  height of the lens.

It is known, e.g., \cite{weissteinCirc}, that such an intersection satisfies
\begin{equation}
\lab{eq:adef}
a = 
\frac 1 d \sqrt {(-d+r-R)(-d-r+R)(-d+r+R)(d+r+R)}.
\end{equation}
Since  $\sigma \le \delta =r \le 1$, $ d = \Theta(1)$ and 
$$
\begin{array}{ccccccc}
-d +r -R &=&  (-1 - \delta + \sigma) + \delta -1 &=& -2 + \delta &=&  - \Theta(1),\\
-d -r +R &=& (-1 - \delta + \sigma) - \delta +1 &=& -2\delta + \sigma &=&  - \Theta(\delta),\\
-d +r +R &=& (-1 - \delta + \sigma) + \delta +1 &=&  \sigma, \\
d +r +R &=&  (1 + \delta-  \sigma) + \delta +1 &=& 2 + 2 \delta  + \sigma &=&   \Theta(1).
\end{array}
$$
Plugging  these values into (\ref{eq:adef})  yields  $a = \Theta(\sqrt {\sigma \delta}).$
Because $P'(v)$ is the intersection of two circles,  $P'(v)$ is convex and thus contains the quadrilateral 
defined by the  four corners 
$$\left((1-\frac \sigma 2, 0\right),  \, \left( 0, \frac \alpha 2\right),\,  \left(1 + \frac \sigma 2, 0\right),\, \left(0, \frac \alpha 2\right).
$$
$P'(v)$  is also, by definition, contained in the rectangle 
$$ \left\{(x,y) \,:\,  1-\sigma/2 \le x \le 1+\sigma/2  \quad   \mbox{and}  \quad - a/2 \le y \le a/2)  \right\}.$$
Thus, 
$\Area(P'(v)) = \Theta(\sigma a) = \Theta(\sigma^{3/2} \sqrt \delta)$. Plugging into Eq.~\ref{eq:B2B2mf} yields

$$f(v) = \Theta
\left( 
\frac  {\Area(P'(v))} {\delta^2}
\right)
= \Theta
\left(  \left(
\frac  {\sigma} {\delta}
\right)^{3/2}\right).
$$
\end{proof}
\section{Analysis of $\Ballpq 1 2$}
\label{sec: B1B2}
This section derives cell v(c) and v(d) in Theorem \ref{thm: main}, that is,  if $n$  points are chosen from 
$\bfd = \Ballpq 1 2$ and  
$\frac 1  {\sqrt n} \le \delta \le n^{1/26}$ then 
$\EMN = \Theta \left(  \frac {n^{2/7}}  { \delta^{3/7}}  \right)$,  while if
$ n^{1/26} \le \delta \le \sqrt n $ then 
$\EMN = \Theta \left(  \frac {n^{1/4}}  { \delta^{1/2}}  \right)$.
Applying Lemma  \ref {lem: scaling}   gives a full analysis for the case $\bfd = \Ballpq 2 1$ as well.

From Corollary \ref{cor: Quadrants} it suffices to analyze
$\EXP{\MAX(S_n \cap O_1)},$  the maxima in the first octant.

  \begin{figure}[t]
\begin{center}
\includegraphics[width=9cm]{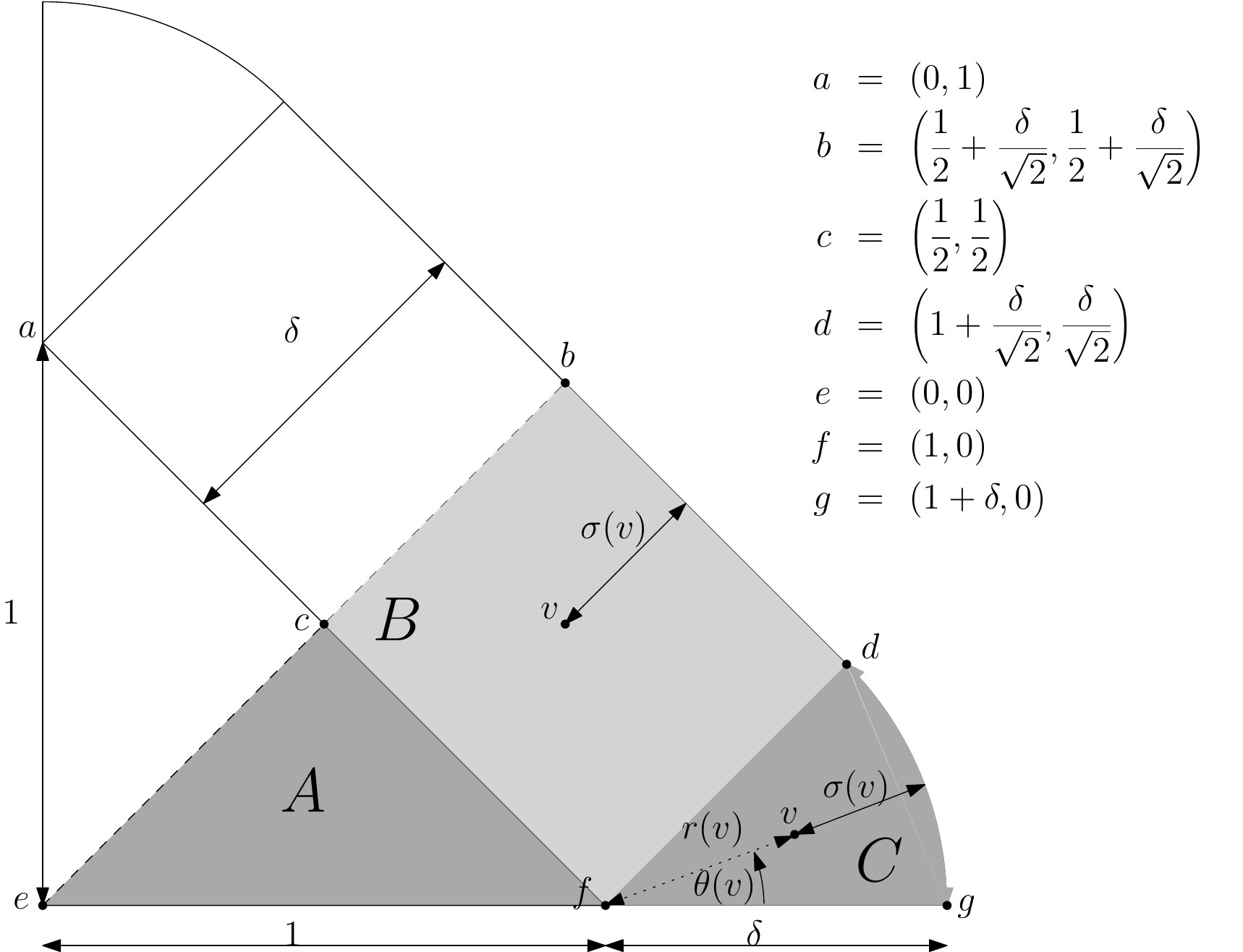}
\end{center}
\caption{Illustration of $Q_1$, the 
 three regions into which the first octant is partitioned  for  $\bfd = \Ballpq 1 2$ and associated variables.  The example illustrates the case $\delta < 1.$  If $\delta >> 1$ then $C$ could be much larger than  $A$ and $B$.
} %
\label{fig:B2 B2 f1}
\rule{5.5in}{0.5pt}
\end{figure}

\begin{Definition}
\lab{def:B1B2 Regionsa}
Let $D= B_1+ \delta B_2$ be the support of $\bfd$ and $D_{1,1} = D \cap O_1$ the support restricted to the first octant.  
Partition $D_{1,1}$ into $A,B,C$ as follows
\begin{eqnarray*}
A &=& D_{1,1} \cap    B_1,\\
B &=&  (D_{1,1} \setminus A) \cap \{u \in \Re^2\,:\,  u.x-u.y \le 1\},\\
C &=&  (D_{1,1} \setminus A) \cap \{u \in \Re^2\,:\,  u.x-u.y \ge 1\}.
\end{eqnarray*}
For $v \in A \cup B $ define
$$\sigma(v) =  \frac 1 {\sqrt 2} \left(1 + \sqrt 2 \delta - (x+y)\right).$$
For $v \in C$ define
\begin{eqnarray*}
r(v) &=&  \sqrt {(v.x-1)^2 + (v.y)^2},\\
\sigma(v) &=& \delta - r(v),\\
\theta(v) &=& \arctan \left( \frac  {v.x-1} {v.y}\right).
\end{eqnarray*}

$$\sigma(v) = 1 + \delta - ||v||,\quad
\theta(v) = \arctan \left( \frac {v.y} {v.x}\right).
$$
Note that in both cases $\sigma(v)$ is the distance from $v$ to the boundary of $D_{1,1}.$ If $v \in C$, $\theta(v)$ is the angle formed with  the $x$ axis by the line connecting $(1,0)$  to $v$.
\end{Definition}

Because it is the  convolution of two very different distributions the density of $\bfd$ is  does not have a clean description.  The next lemma encapsulates  properties that can be used  to derive the behavior of $\EMN$.

\begin{Lemma}
\lab{lem:B1B2 measure}
Let $\bfd = \Ballpq 1 2$  
and $v \in D_{1,1}.$

\begin{itemize}
\item  If $v \in A$, then 
$$f(v) =\left\{
\begin{array}{ll}
\Theta (1) & \mbox{if $ \delta \le 1$},\\
\Theta\left( \frac  1 {\delta^2} \right)& \mbox{if $ \delta > 1$}.
\end{array}
\right.
$$
\item If $v \in B$, then
\begin{equation}
\label{eq: B1B2 measureB}
f(v)
=\left\{
\begin{array}{ll}
\Theta\left(\frac 1 {\delta^2} \right) & \mbox{if  $\sigma(v) = \Omega(1)$},\\
\Theta \left(\frac  {\sigma(v)} {\delta^2} \right) & \mbox{if $\delta \sigma(v) = \Omega(1)$ and  $\sigma(v) = O(1)$},\\
\Theta  \left(  \left( \frac  {\sigma(v)} {\delta}  \right)^{3/2} \right)& \mbox{if $\delta \sigma(v) = O(1)$ and  $\sigma(v) = O(1).$}
\end{array}
\right.
\end{equation}

\item If $v \in C$, then

\begin{enumerate}
\item If $\theta(v) =\frac \pi 4$ then $f(v)$ is  as defined in Eq.~\ref{eq: B1B2 measureB}
\medskip

\item Fix $d > 0$ and  set $\gamma(v)  = \frac \pi 4 - \theta(v)$. If   $\sigma(v)< d \delta$, then 
\begin{equation}
\label{eq: B1B2 measureC}
f(v) = 
 \left\{
 \begin{array}{ll}
 \Theta\left( \frac {\min(\sigma(v),1) \cdot \min\left( \sqrt {\delta \sigma(v)},1)\right)} {\delta^2}\right) & \mbox{if  $ \sigma(v)  = \Omega(\delta (\gamma(v))^2)$},\\
 \Theta\left( \frac {\min(\sigma(v),1) \cdot \min\left( \frac  {\sigma(v)}  {\gamma(v)},1\right)} {\delta^2}\right) & \mbox{if  $ \sigma(v)  = O(\delta (\gamma(v))^2)$}.\\
 \end{array}
 \right.
 \end{equation}
\end{enumerate}
\end{itemize}

\end{Lemma}

This immediately implies 

\begin{Corollary}
\label{cor: B1B2theta}
Fix $d >0.$ If $v \in C$  and $\sigma\le d \delta$  then $f(v) = \Theta\left(g(\theta(v),\sigma(v))\right)$,  where $g(\theta,\sigma)$ satisfies
\begin{itemize}
\item $\forall \sigma,$  if  
$\theta_1 \le \theta_2$ then 
$g\left(\theta_1,\sigma\right)  = O(g\left(\theta_2,\sigma\right))$
\item 
$\forall \theta,\quad   g\left(\theta,2 \sigma\right)  = O(g\left(\theta,\sigma\right)) $,\\
where the constant in the $O()$  is independent of $\sigma, \delta.$ 
\end{itemize}
\end{Corollary}

The lower bound is easy to derive:
\begin{Lemma}
\label{lem: B1B2LB}
Let $S_n$ be $n$ points chosen from the distribution $\bfd = \Ballpq 1 2$ with $\frac 1 {\sqrt n} \le \delta \le \sqrt n$. Then
$$\EMN = 
\left\{\begin{array}{ll}
\Omega \left(  \frac {n^{2/7}}  { \delta^{3/7}}  \right), & \mbox{if  $\frac 1 {\sqrt n} \le \delta \le n^{1/26}$},\\
 \Omega \left( \sqrt \delta   {n^{1/4}} \right), & \mbox{if  $n^{1/26} \le \delta \le \sqrt n$}.\\
\end{array}
\right.
$$
\end{Lemma}

\begin{proof} 

\par\noindent\underline{(a) $\EXP{|\MAX(S_n \cap (A \cup B))|}$  when  $ \frac 1 {\sqrt n}\le \delta \le n^{1/5}:$}

Fix $\sigma$  with value to be determined later. Set $m = \lfloor 1/\sigma \rfloor$.\\
For $ i = 0,\ldots m-1$  define
$$ x_i =\frac 1 2 + \frac \delta {\sqrt 2} +  i \sigma,
\quad \quad
y_i = \frac 1 2 + \frac \delta {\sqrt 2} -  x_i - \sigma,
\quad\quad  p_i = (x_i,y_i).
$$
The $p_i$ are $m$ equally spaced points along the line $\sigma(x) = 1 + \sqrt 2 \delta-  \sigma$, which is parallel and distance $\sigma$ from the  line $x+y = 1 + \sqrt 2 \delta$ on the boundary of $D.$   Thus  the $P(p_i)$ are isosceles triangles.
Note that  $\forall  i\not =j,$  $P(p_i) \cap P(p_{j}) = \emptyset$. Since all $P(p_i)$ are dominant regions, if we could show  
$\forall i,\,  \mu(P(p_i) = \Theta(\frac 1 n)$
Lemma \ref {lem: lb}
implies that $\EXP{M_n} = \Omega(m)$.

Set 
$$\hat P(v) = P(v) \cap \{v' \in A \cup B \,:\,  \sigma(v') \ge \sigma(v)/2\}.$$
Then, from Eq.~\ref{eq: B1B2 measureB},  
$$\forall v' \in P(v),\,  f(v') = O (f(v)
\quad\mbox{and}\quad
\forall v' \in \hat P(v),\,  f(v') = \theta(f(v)).$$
Since $\Area(P(p_i)) = \frac 1 2 \sigma^2$ and $ \Area(\hat P(p_i) = \frac 1 8 \sigma^2$, this gives
$$\mu(P(p_i)) = \Theta\Bigl(\Area(P(p_i))\, f((p_i))\Bigr) =  \Theta\Bigl(\sigma^2  f((p_i))\Bigr).$$ 

%

Now set $\sigma = \frac {\delta^{3/7}} {n^{2/7}}.$
Then $1 \le \delta \le n^{1/5}$ implies
$$
\sigma  =\frac {\delta^{3/7}}{n^{2/7}} \le n^{\frac {3} {35} - \frac {10} {35}} \le 1
\quad\mbox{and}\quad {\delta \sigma} = { \frac  {\delta^{10/7}} {n^{2/7}}}\le 1.
$$
This satisfies the final condition in Eq.~\ref{eq: B1B2 measureB}, so
$$\mu(P(p_i))  =  \Theta\Bigl(\sigma^2  f((p_i))\Bigr) =  
\Theta\left(\sigma^2  \left( \frac  {\sigma} {\delta}  \right)^{3/2} \right)
= \Theta\left(\frac  {\sigma^{7/2}} {\delta^{3/2}}\right)
= \Theta\left(\frac   {\delta^{3/2}} n  \frac 1  {\delta^{3/2}}\right) = \Theta\left(\frac 1 n\right).$$

Then Lemma \ref {lem: lb}
 implies that 
$$\EXP{M_n} = \Omega(m)  =  \Omega\left( \frac 1 n \right) = \Omega\left( \frac {n^{2/7}}  {\delta^{3/7}}  \right).$$

\par\noindent\underline{(b) $\EXP{|\MAX(S_n \cap C|}$ when 
$ 1 \le \delta  \le \sqrt n:$}

Set $\sigma = \frac {\sqrt \delta} {2 n^{1/4}}.$

Let $v \in C$ be such that $\sigma(v) \le \sigma$ and $\theta(v) \le \pi/8$ so $\pi/8 \le \gamma(v) \le \pi/ 4.$
Note that this implies $\sigma(v) \le \min (1,\delta/2)$ and $\gamma(v) = \Theta(1)$ so, from Eq.~\ref{eq: B1B2 measureC}
$$f(v) = \Theta\left(\frac {(\sigma(v))^2} {\delta^2}\right).$$

The construction is very similar of the proof of  Lemma \ref {lem: B2B2LB} so we only sketch the details.

First let $v \in C$ be such that $\sigma(v) = \sigma$ and $\theta(v) \le \pi/8.$

Recall $P(v)$ as 
introduced in Definition  \ref {def: PDdef}. 
  Let $v'$, $v''$ be the two other vertices of $P(v)$ with $v.x = v''.x$, 
$v.y =  v'.y$,  $\alpha(v) = v''.y - v.y$ and $ \beta(v)  = v'.x - v.x$.  
Note that  $\pi/8 \le \theta(v) \le \pi/4$   immediately implies
$\beta(v) = \Theta(\sigma(v))$ and $\alpha(v) = \Theta(\sigma(v)).$ Let $c > 0$ be such that $\beta(v) \le c \sigma(v)$ for all such $v.$

Following the same steps as in the derivation of Eq.~\ref{eq:B2B2 1n prob},
$$\mu(P(v)) = \Theta\left(
(\sigma(v))^2 \frac {(\sigma(v))^2} {\delta^2}
\right)
=
\Theta\left(
 \frac {\sigma^4} {\delta^2}
\right)
= \Theta\left(
 \frac {1} {n}
\right)
.$$

Set $x_\ell = 1 + \delta \cos (\pi/8)$ and  $x_r = 1 + \delta\cos (\pi/4)$.  Note that $x_r - x_\ell=  \Theta( \delta).$

Let $m = \lfloor \frac  {x_r - x_\ell} {c \sigma}\rfloor$.  Since $x_r - x_\ell= \Theta(1), $  $m = \Theta\left(\frac \delta \sigma \right) = \Theta\left(     \sqrt \delta n^{1/4}  \right).$
Now set, $x_0=x_l$ and, for  $i=1,\ldots,m,$
\begin{eqnarray*}
x_i &=& x_l + i \sigma,\\
v_i &=&  (x_i,  \, \sqrt {(\delta- \sigma)^2 - (x_{i})^2}).
\end{eqnarray*}

By construction,   $\sigma(v_i) = \Theta(\sigma)$,    $\theta(v_i) \in \left[\frac \pi 8, \frac \pi 4\right]$ and thus
$\mu(P(v_i))= \Theta(1/n)$.  Furthermore, each $P(v_i)$ is a dominant region and, 
since by construction,  $\forall i,$ $x_i + \beta(v_i)  \le x_{i+1},$ 
the $P(v_i)$ are pairwise disjoint.  Lemma \ref{lem: lb} then   
 proves  the required 
$$\EXP{M_n}= \Omega(m) = \Omega\left(     \sqrt \delta n^{1/4}  \right).$$

To complete the proof of the lemma simply note that
$$
\sqrt \delta n^{1/4} \le  \frac {n^{2/7}} {\delta^{3/7}} 
 \quad \Leftrightarrow \quad 
\delta^{\frac {26} {28}} = \delta^{\frac {13} {14}} = \delta^{\frac 1 2 + \frac 3 7} \le n^{\frac 2 7 - \frac 1 4} = n^{\frac 1 {28}}
 \quad \Leftrightarrow \quad    \delta \le n^{\frac 1 {26}}.
$$
\end{proof}

As usual, the   upper  bound is  more technical.
\begin{Lemma}
\label{lem: B1B2UB}
Let $S_n$ be $n$ points chosen from the distribution $\bfd = \Ballpq 1 2$ with $\frac 1 {\sqrt n} \le \delta \le 1$. Then
\begin{equation}
\label{eq:B1B2UPstate}
\EMN = 
\left\{\begin{array}{ll}
O \left(  \frac {n^{2/7}}  { \delta^{3/7}}  \right), & \mbox{if  $\frac 1 {\sqrt n} \le \delta \le n^{1/26}$},\\
O \left( \sqrt \delta   {n^{1/4}} \right), & \mbox{if  $n^{1/26} \le \delta \le \sqrt n$}.\\
\end{array}
\right.
\end{equation}
\end{Lemma}

\begin{proof}
The proof will be split into two parts:  the first an upper bound on $v \in A \cup B$ and the second an upper bound on $v \in C.$
This will proceed via case analyses that show
$$(a) \ \EXP{\left| \MAXSN \cap (A \cup B) \right|} =
\left\{
\begin{array}{ll}
O\left(\frac   {n^{2/7}}  {\delta^{3/7}} \right) & \mbox{If $\frac 1{ \sqrt n} \le \delta \le \frac 1 {\sqrt 2}$}, \\
O\left( \frac {n^{2/7}} {\delta^{3/7}} \right)   & \mbox{If $\frac 1 {\sqrt 2} \le \delta \le \ n^{1/5}$},\\
O\left( \frac {n^{1/3}} {\delta^{2/3}}\right)  & \mbox{If $ n^{1/5}\le \delta \le \frac {\sqrt n} {\log n}$},\\
O(\log^4 n)  &\mbox{If $  \frac {\sqrt n} {\log n} \le \delta  \le \sqrt n$},
\end{array}
\right.
$$
and
$$(b) \ \EXP{\left| \MAXSN \cap C\right|} =
\left\{
\begin{array}{ll}
O\left(\frac   {n^{2/7}}  {\delta^{3/7}} \right) & \mbox{If $\frac 1{ \sqrt n} \le \delta \le 1$}, \\
O\left(\sqrt \delta n^{1/4}\right)  & \mbox{If $1  \le \delta \le  \sqrt n$}.
\end{array}
\right.
$$
Since $\frac {n^{1/3}} {\delta^{2/3}} = O\left(\sqrt \delta n^{1/4}\right) $  if $ n^{1/5} \le \delta$  and  
$\frac   {n^{2/7}}  {\delta^{3/7}} > \sqrt \delta n^{1/4}$ if and only if $ \delta < n^{1/26}$,  the proof  will follow.
\medskip

\par\noindent\underline{(a) $ v \in A \cup B:$}

First note, that from Lemma \ref{lem:B1B2 measure}, $\forall v,v' \in A \cup B,$
\begin{eqnarray}
\mbox{If $\sigma(v') \le \sigma(v)$}  & \Rightarrow& f(v') = O((f(v)). \label{eq:B1B2UP1}\\
\mbox{ If $\frac 1 2 \sigma(v) \le \sigma(v') \le \sigma(v)$}   & \Rightarrow& f(v') = \Theta((f(v)).\label{eq:B1B2UP2}
\end{eqnarray}

Now set $\sigma_b= \frac {\delta^{3/7}} {n^{2/7}}$ and $\sigma_c  =\frac {\delta^{2/3}}{n^{1/3}}.$

If $ \frac 1 {\sqrt n}  \le \delta \le n^{1/5},$ then 
$$
\sigma_b  =\frac {\delta^{3/7}}{n^{2/7}} \le n^{\frac {3} {35} - \frac {10} {35}} \le 1
\quad\mbox{and}\quad {\delta \sigma_b} = { \frac  {\delta^{10/7}} {n^{2/7}}}\le 1.
$$
This satisfies the final condition in Eq.~\ref{eq: B1B2 measureB} so, if $\sigma(v) = \Theta(\sigma_b)$,
\begin{equation}
\label{eq:B1B2UP8}
 \Theta\Bigl(\sigma^2(v)  f((v)\Bigr) =  
\Theta\left(\sigma_b^2  \left( \frac  {\sigma_b} {\delta}  \right)^{3/2} \right)
= \Theta\left(\frac  {\sigma_b^{7/2}} {\delta^{3/2}}\right)
= \Theta\left(\frac   {\delta^{3/2}} n  \frac 1  {\delta^{3/2}}\right) = \Theta\left(\frac 1 n\right).
\end{equation}

If $ n^{1/5} < \delta$, then  combining with the constraint  
$\delta \le  \sqrt n$ implies
$$
\sigma_c  =\frac {\delta^{2/3}}{n^{1/3}}  \le 1
\quad\mbox{and}\quad 
{\delta \sigma_c} = { \frac  {\delta^{5/3}} {n^{1/3}}}> 1.$$
This satisfies the middle   condition in Eq.~\ref{eq: B1B2 measureB} so, if $\sigma(v) = \Theta(\sigma_c)$,
\begin{equation}
\label{eq:B1B2UP9}
\Theta\Bigl(\sigma^2(v)  f((v))\Bigr) =  
\Theta\left(\sigma^2_c   \frac  {\sigma^3_c} {\delta^2}   \right)
= \Theta\left(\frac   {\delta^{2}} n  \frac 1  {\delta^{2}}\right) = \Theta\left(\frac 1 n\right).
\end{equation}

The analysis will be similar to that used to prove Lemma \ref{lem: B1B1UB}.  For an  appropriate value $\sigma$,   $A \cup B$ 
will be partitioned into $m = \Theta (1/\sigma)$ vertical {\em strips} of width $\sigma$ and, using the sweep Lemma,  showing  that each strip contains only $O(1)$ expected maxima.

Due to changes in the geometry of the support and distribution of  $\Ballpq 1 2$ as $\delta$ grows, the analysis is  split  into three  cases:  (a)(i) $\frac 1{ \sqrt n} \le \delta \le \frac 1 {\sqrt 2}$,  (a)(ii)  $\frac 1 {\sqrt 2} \le \delta \le \frac {\sqrt n} {\log n}$ and
(a)(iii) $  \frac {\sqrt n} {\log n} \le \delta  \le \sqrt n.$

\bigskip

\par\noindent\underline{(a)(i)  $\frac 1{ \sqrt n} \le \delta \le \frac 1 {\sqrt 2}:$}


\medskip
\begin{figure}[t]
\begin{center}
\includegraphics[width=10cm]{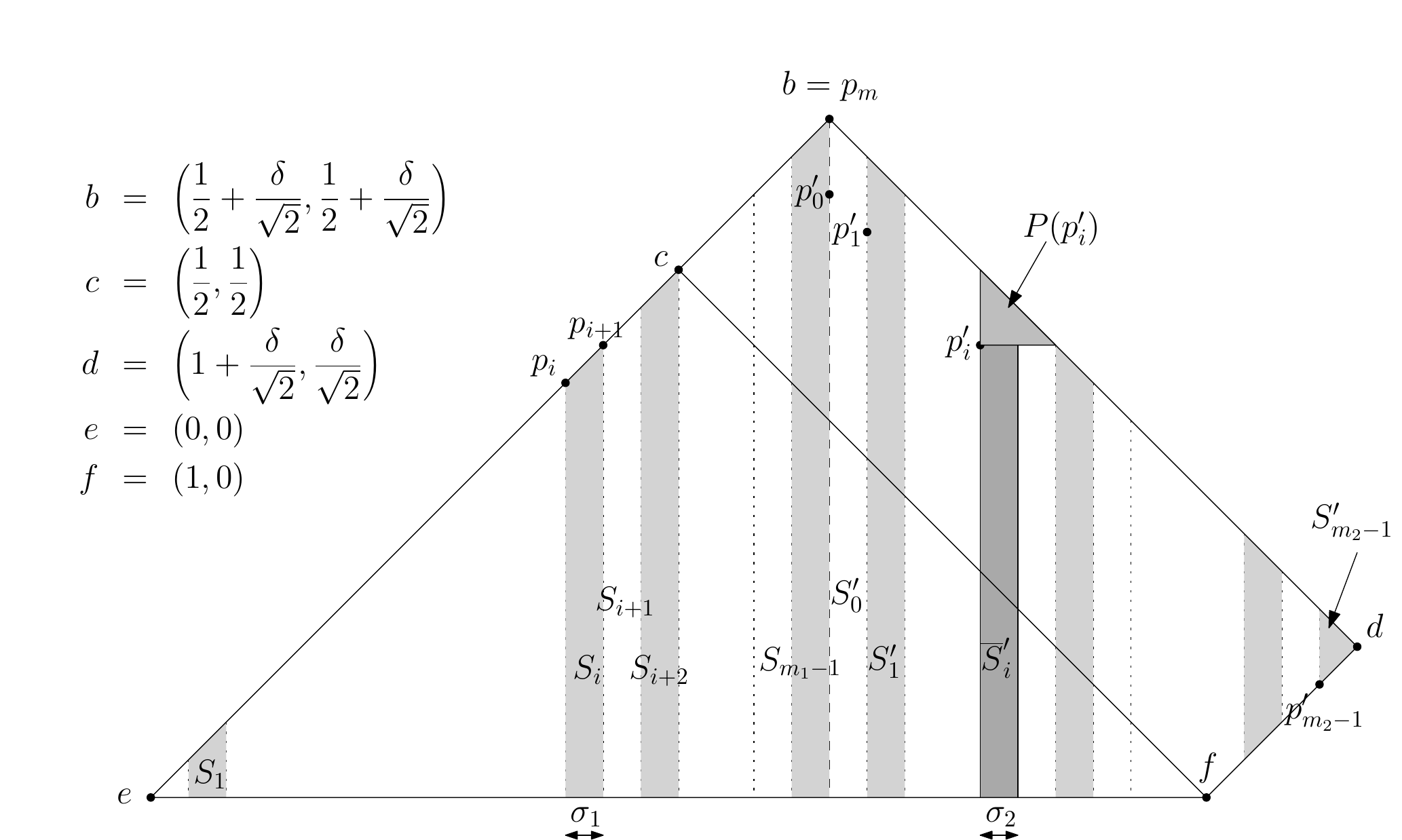}
\end{center}
\caption{Illustration of  the derivation of Lemma \ref{lem: B1B2UB} for case (a)(i).
$S_i,$ $S'_i$ and $\overline {S'}_i$ denote, respectively, $\Strip_i$, $\Stripp_i$ and $\Stripph_i$.
 } %
\label{fig:B2 B2 f7}
\rule{5.5in}{0.5pt}
\end{figure}
Fix $\sigma = \sigma_b$  and define
$$m_1 =  \left \lceil
\frac 1   \sigma  \left(\frac 1 2 + \frac \delta {\sqrt 2} \right)
\right\rceil
\quad\mbox{and}\quad
m_2  =\left \lceil  \frac 1  {2  \sigma}\right\rceil,
$$
$$ \sigma_1 = \frac 1 {m_1}\left(\frac 1 2 + \frac \delta {\sqrt 2} \right)
\quad\mbox{and}\quad
\sigma_2 = \frac 1 {2 m_2}.
$$

By construction $\sigma = \Theta\left(  \sigma_1 \right)= \Theta\left(  \sigma_2\right)= \Theta\left(  \sigma_b\right).$

Define (illustrated in  Fig.~\ref{fig:B2 B2 f7})
\begin{align*}
\forall i &\le  m_1, 	&   x_i &=  i \sigma_1 , 		&   y_i &=  i \sigma_1, 	& p_i&=(x_i,y_i) \\
\forall i & \le  m_2, 	   &x'_i &=  \frac 1 2 + \frac \delta {\sqrt 2} + i \sigma_2, & y'_i &= 1 + \sqrt 2 \delta - x_i - 2 \sigma_2, &  p'_i&=(x'_i,y'_i), 
\end{align*}

Note that $x_{m+1} = x'_0.$ Also define 
\begin{align*}
\forall i &\le  m_1-1, 	 & \Strip_i &= \left\{\{u \in (A \cup B)\,:\,  u.x \in  [x_i, x_{i+1}]\right\}, &&\\
\forall i & \le  m_2-1, 	 & \Stripp_i &=    \{u \in A \cup B)\,:\, u.x \in  [x'_i, x'_{i+1}]  \}, &\Stripph_i &=  \{u \in \Stripp_i \,:\,   u.y \le y'_{i}\}.
\end{align*}
Also note that  $\Strip_i$ and $\Stripp_i$   partition $A \cup B$ so
$$
\EXP{|\MAX(S_n \cap (A \cup B) |}
\le
\sum_{i=0}^{m_1-1} \EXP{|\MAX(S_n) \cap \Strip_i| } + \sum_{i=0}^{m_2-1} \EXP{|\MAX(S_n) \cap \Stripp_i| }.
$$

By definition $\Stripp_i \subseteq \Stripph_i \cup P(p'_i).$  Also, since $\sigma(p'_i) = \Theta(\sigma_b),$ 
Eq.~\ref{eq:B1B2UP8} gives  $\mu(P(p'_i)) = O(\sigma_b^2 f(p'_i)) = O(1/n)$  and thus
\begin{eqnarray*}\EXP{| \MAX(S_n \cap  \Stripp_i| )} 
&\le&  \EXP{| \MAXSN \cap  \Stripph_i| } + \EXP{| \MAXSN\cap  P(p'_i)| )}\\
&\le &  \EXP{| \MAXSN \cap  \Stripph_i| } + n \mu(P(p'_i)) \\
&\le &  \EXP{| \MAXSN \cap   \Stripph_i| } + O(1). \\
\end{eqnarray*}
Finally, note that  $\Stripph_{m_2-1} = \emptyset.$
Thus
$$
\EXP{|\MAXSN \cap (A \cup B)| }
\le
\sum_{i=0}^{m_1-1} \EXP{|\MAXSN\cap \Strip_i| } + \sum_{i=0}^{m_2-2} \EXP{|\MAXSN \cap \Stripph_i| }  + O(m_2).
$$

For  $i\le m_1-2$ set 
$$ A(t) = \left\{u \in \Strip_{i+1} \,:\, u.y \ge  y_{i+1} -t\right\}
\quad\mbox{and} \quad B(t) =  \left\{u \in \Strip_{i}\,:\, u.y \ge  y_{i+1} -t\right\}.
$$

From Eqs.~\ref{eq:B1B2UP1}  and \ref{eq:B1B2UP2} it is straightforward to see that \\
$\forall t,\, \mu(B(t)) = O(\mu(A(t))$, so the sweep Lemma shows that $$\EXP{| \MAXSN \cap \Strip_i |} = O(1).$$

Similarly, for  $i\le  m_2-2$ set 
$$ A(t) =  \left\{u \in  \Stripp_{i+1} \,:\, u.y \le   y'_{i} -t\right\}
\quad\mbox{and} \quad B(t) = \left\{u \in \Stripph_{i} \,:\, u.y \le y'_{i}  -t\right\}.
$$
Again, from Eqs.~\ref{eq:B1B2UP1}  and \ref{eq:B1B2UP2} it is straightforward to see that 
$\forall t,\, \mu(B(t)) = O(\mu(A(t))$, so 
 the sweep lemma shows that 
 $$\EXP{|\MAXSN \cap \Stripph_i |} = O(1).$$

Combining all of the above yields

\begin{eqnarray*}
\EXP{|\MAXSN \cap (A \cup B)| }
&\le&
\sum_{i=0}^{m_1-2} \EXP{|\MAXSN\cap \Strip_i| } + \sum_{i=0}^{m_2-2} \EXP{|\MAXSN\cap \Stripp_i| }  + O(m_2)\\
&=& O(m_1) + O(m_2) + \EXP{|\MAXSN\cap \Strip_{m_1-1}|}.
\end{eqnarray*}

To complete the proof set 
$$\hspace*{-.2in} A(t) = \left \{u \in \Stripp_{0} :\, u.y \le   \frac 1 2 + \frac \delta {\sqrt 2}  -t\right\}
\quad\mbox{and} \quad B(t) = \left\{u \in \Strip_{m_1-1}\,:\, u.y \le \frac 1 2 + \frac \delta {\sqrt 2}  -t\right\}.
$$
Then a final 
application of the sweep lemma shows that $$\EXP{| \MAXSN \cap \Strip_{m_1-1}| } = O(1),$$
so
$$
\EXP{|\MAXSN \cap (A \cup B)| } = O(m_1 + m_2) = O\left( \frac 1 \sigma_b  \right) = O \left(  \frac {n^{2/7}}  { \delta^{3/7}}  \right).
$$

\medskip

\par\noindent\underline{(a)(ii)  $\frac 1 {\sqrt 2} \le \delta \le \frac {\sqrt n} {\log n}:$}

\medskip

Set
$$
b = \left(  \frac 1 2 + \frac {\delta} {\sqrt 2},  \frac 1 2 + \frac {\delta} {\sqrt 2}  \right), \quad 
d =  \left( 1+\frac {\delta} {\sqrt 2}, \frac {\delta } {\sqrt 2}  \right), \quad
h = \left(  \frac 1 2 + \frac {\delta} {\sqrt 2},   \frac {\delta} {\sqrt 2}\right).
$$
\begin{figure}[t]
\begin{center}
\includegraphics[width=10cm]{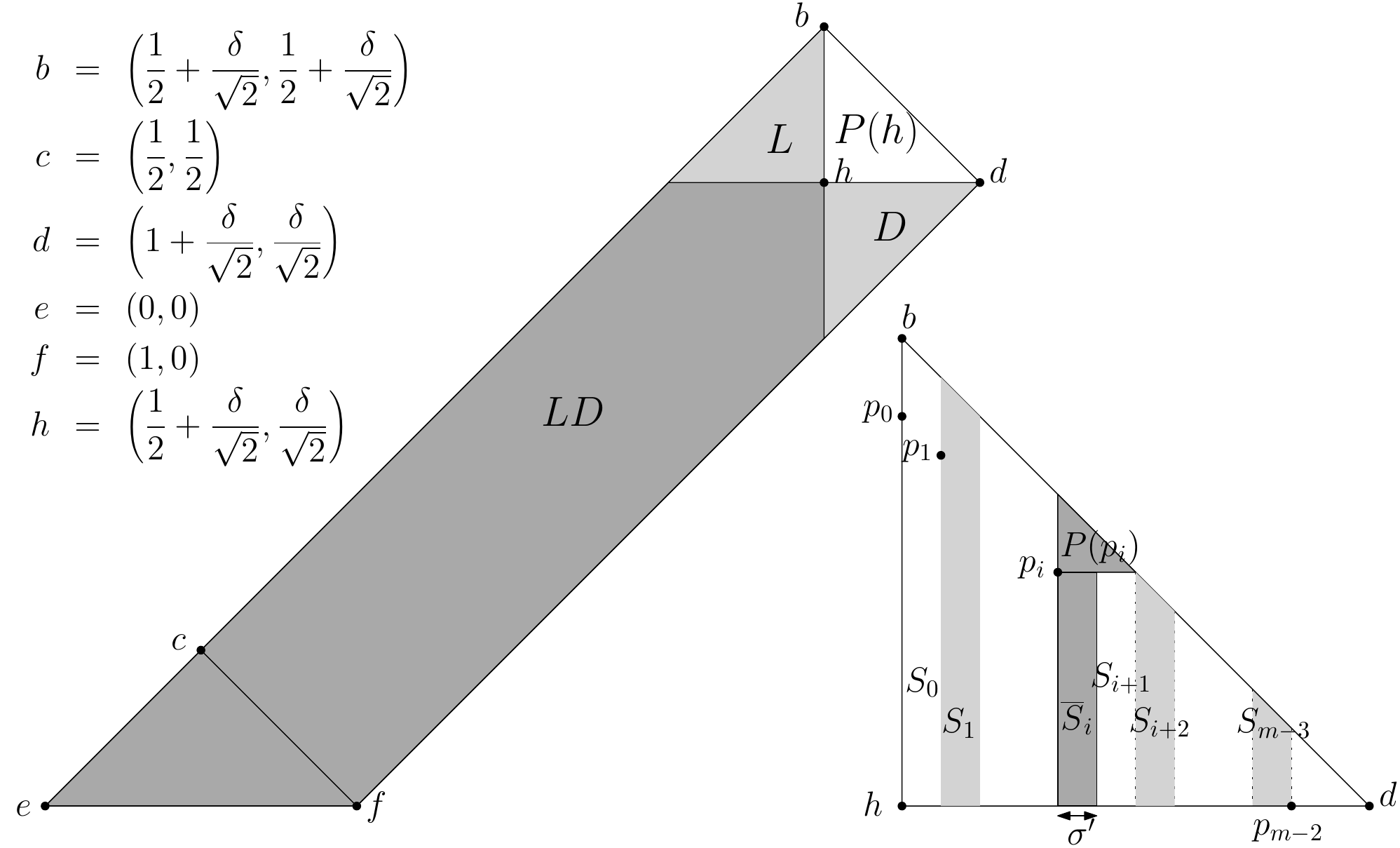}
\end{center}
\caption{Illustration of  the derivation of Lemma \ref{lem: B1B2UB} for case (a)(ii). The right hand figure is  a blown up version of 
$P(h).$ $S_i,$ and $\overline {S}_i$ denote, respectively, $\Strip_i$ and $\Striph_i$.
 } %
\label{fig:B2 B2 f8}
\rule{5.5in}{0.5pt}
\end{figure}

See Fig.~\ref{fig:B2 B2 f8}. 
The proof works in two steps.  The first is to show that 
\begin{equation}
\label{eq:B1B2UP3}
\EXP{|\MAXSN \cap (A \cup B)| } = \EXP{|\MAXSN \cap P(h)| } + O(1).
\end{equation}
The second is to prove 
\begin{equation}
\label{eq:B1B2UP4}
 \EXP{|\MAXSN \cap P(h)| } =
 \left\{
\begin{array}{ll}
O\left( \frac {n^{2/7}} {\delta^{3/7}} \right) & \mbox{if $\frac 1 {\sqrt 2} \le \delta \le n^{1/5}$},\\
O\left( \frac {n^{1/3}} {\delta^{2/3}} \right) & \mbox{if $ n^{1/5}\le \delta \le \frac {\sqrt n} {\log n}$}.
\end{array}
\right.
\end{equation}

Set 
\begin{eqnarray*}
L &=&  \left\{u \in A \cup B\,:\,  u.x \le  \frac 1 2 + \frac {\delta} {\sqrt 2},\,  u.y \ge  \frac {\delta} {\sqrt 2}    \right\},\\
D & =&  \left\{u \in A \cup B\,:\,  u.x \ge  \frac 1 2 + \frac {\delta} {\sqrt 2},\,  u.y \le  \frac {\delta} {\sqrt 2}    \right\},\\
LD &=& \left\{u \in A \cup B\,:\,  u.x \le  \frac 1 2 + \frac {\delta} {\sqrt 2},\,  u.y \le  \frac {\delta} {\sqrt 2}    \right\}.\\
\end{eqnarray*}
Since $h \in A \cup B$ and $\sigma(h) = \Theta(\delta),$ Lemma \ref{lem:B1B2 measure} implies that 
$f(h)  = \Theta\left(  \frac 1 {\delta^2} \right)$.  Then using tandard techniques  show that
$$\mu(P(h)) = \Theta \left(  \Area(P(h)) \cdot f(h) \right) = \Theta\left(  \frac 1 {\delta^2} \right).
$$
From Lemma \ref {lem: basic mu}(c)
\begin{equation}
\label{eq:B1B2UPcase2Ph}
 \Pr(  (S_n \cap P(h))= \emptyset) =  \left(1 - \mu(P(h))\right)^n
 \le  e^{- c \log^2 n}
\end{equation}
for some constant $c >0.$  Since {\em any} point in $P(h)$ dominates {\em all} points in $LD,$
$(S_n \cap P(h))\not = \emptyset$ implies $\MAX(S_n) \cap LD = \emptyset.$  Thus
\begin{equation}
\label{eq:B1B2UP5}
\EXP{|\MAXSN\cap  LD| }  \le n \Pr(  (S_n \cap P(h))= \emptyset) = o(1).
\end{equation}
\medskip
Now set 
$$A(t ) = \left\{ u \in P(h) \,:\,  u.y \le \frac 1 2 + \frac {\delta} {\sqrt 2} -t \right\},\quad
B(t) = \left\{ u \in L\,:\,  u.y \le \frac 1 2 + \frac {\delta} {\sqrt 2} -t \right\}.
$$
It is straightforward to show that $\forall t \ge 0,$  $\mu(B(t)) = O(\mu(A(t)))$, 
so the sweep Lemma yields that 
\begin{equation}
\label{eq:B1B2UP6}\EXP{|\MAXSN\cap L|} = O(1).
\end{equation}
Similarly, setting 
$$A(t ) = \left\{ u \in P(h) \,:\,  u.x \le 1 + \frac {\delta} {\sqrt 2} -t \right\},\quad
B(t) = \left\{ u \in D\,:\,  u.y \le  1 + \frac {\delta} {\sqrt 2} -t \right\},
$$
and noting that  $\forall t \ge 0,$  $\mu(B(t)) = O(\mu(A(t)))$ permits using the sweep Lemma to prove that
\begin{equation}
\label{eq:B1B2UP7}\EXP{|\MAXSN \cap D|} = O(1).
\end{equation}
Combining Eqs.~\ref {eq:B1B2UP5}, \ref{eq:B1B2UP6} and \ref{eq:B1B2UP7} proves 
Eq.~\ref{eq:B1B2UP3}.

To analyze $\EXP{|\MAX(S_n \cap (P(h))| }$,  fix $\sigma$ to be a value to be specified later and 
set 
$$m = \left\lceil  \frac  1{ 2 \sigma}   \right \rceil 
\quad\mbox{and}\quad 
\sigma' = \frac 1{ 2 m}.
$$
 Note  that
$\sigma =  \Theta(\sigma').$
Now set
$$\forall i \le  m,\quad
x_i =  \frac 1 2 + \frac \delta {\sqrt 2} + i \sigma',\quad 
y _i = 1 + \sqrt 2 \sigma - x_i - 2 \sigma,\quad  p_i=(x_i,y_i).
$$
$$\forall i  \le  m-2, \quad 	  \Strip_i =   \left\{u \in P(h)\,:\, u.x \in  [x'_i, x'_{i+1}]  \right\},\quad 
\Striph_i  =  \left\{u \in \Strip_i \,:\,   u.y \le y'_{i}\right\}.
$$

By definition $\Strip_i \subseteq \Striph_i \cup P(p_i)$ and the $\Strip_i$ partition $P(h)$ with
$\Strip_{m-1} \cup \Strip_{m-2} = P(p_{m-2}).$ Thus
\begin{eqnarray*}
\EXP{|\MAX(S_n \cap P(h)| }
&=&
\sum_{i=0}^{m-3} \EXP{|\MAXSN \cap \Strip_i| } +  \EXP{|\MAXSN\cap P(p_{m-2})| }\\
&\le & \sum_{i=0}^{m-3} \EXP{|\MAXSN\cap \Striph_i| }  +\sum_{i=0}^{m-2} \EXP{|\MAXSN \cap P(p_i)| }\\
&\le & \sum_{i=0}^{m-3} \EXP{|\MAXSN\cap \Striph_i| }  +\sum_{i=0}^{m-2} n  \mu(P(p_i)).
\end{eqnarray*}

Now, for all $i\le  m_2-3$ set 
$$ A(t) =    \left\{u \in \Strip_{i+1}\,:\, u.y \le   y_{i} -t\right\}
\quad\mbox{and} \quad B(t) =  \left\{u \in \Striph_{i}\,:\, u.y \le y'_{i}  -t\right\}.
$$
Again, from Eqs.~\ref{eq:B1B2UP1}  and \ref{eq:B1B2UP2} it is straightforward to see that 
$\forall t,\, \mu(B(t)) = O(\mu(A(t))$ so 
 the sweep Lemma shows that 
 $$\forall i\le  m_2-3,\quad \EXP{| \MAXSN \cap \Striph_i |} = O(1),$$
giving
\begin{equation}
\label{eq:B1B2UPcase1}
 \EXP{|\MAXSN \cap P(h)|}  \le O(m) +\sum_{i=0}^{m-2} n  \mu(P(p_i)).
\end{equation}
To complete the proof note that $\mu(P(p_i)) = \Theta\Bigl(\sigma^2(p_i) f(p_i)\Bigr)$.

If $ \delta \le n^{1/5}$, set 
$\sigma = \sigma_b.$ Since $\sigma(p_i) = \Theta(\sigma_b)$,   Eq.~\ref{eq:B1B2UP8}
gives  
$$\mu(P(v))  = \Theta\left(\frac 1 n\right).$$
Plugging into Eq.~\ref{eq:B1B2UPcase1} yields
$$\EXP{|\MAXSN \cap P(h)|}  = O(m)
=O\left( \frac  1 \sigma _b\right) = O\left( \frac {n^{2/7}} {\delta^{3/7}} \right).
$$

If  $\delta >  n^{1/5} $ set $\sigma = \sigma_c.$ 
Since $\sigma(p_i) = \Theta(\sigma_b)$,   Eq.~\ref{eq:B1B2UP9}
gives  
$$\mu(P(v))  = \Theta\left(\frac 1 n\right).$$
Plugging into Eq.~\ref{eq:B1B2UPcase1} yields
$$\EXP{|\MAXSN \cap P(h)|}  = O(m)
=O\left( \frac  1 \sigma _c\right) = O\left( \frac {n^{1/3}} {\delta^{2/3}} \right).
$$

\medskip

\par\noindent\underline{(a)(iii) $  \frac {\sqrt n} {\log n} \le \delta  \le \sqrt n:$}

Note that the method used for case (a)(ii) fails in this interval because  Eq.~\ref{eq:B1B2UPcase2Ph} is no longer valid.  Moreover, $\mu(p(h))$ can be so small, that $\Pr(S_n \cap P(h) = \emptyset)$ is non-negligible and thus $LD$ might contain some maximal points.

\begin{figure}[t]
\begin{center}
\includegraphics[width=12cm]{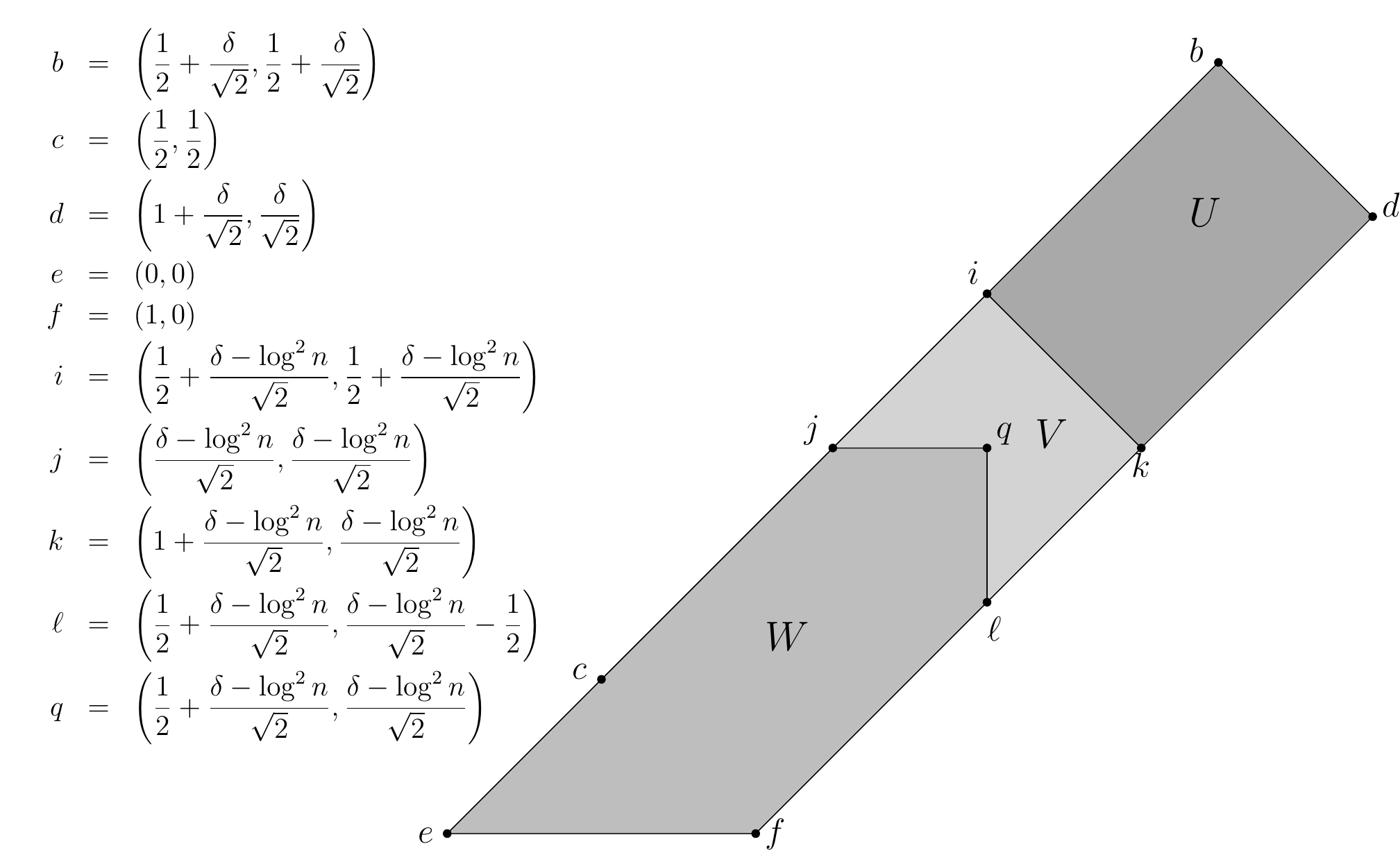}
\end{center}
\caption{Illustration of  the derivation of Lemma \ref{lem: B1B2UB} for case (a)(iii). 
 } %
\label{fig:B2 B2 f9}
\rule{5.5in}{0.5pt}
\end{figure}

Refer to Fig.~\ref{fig:B2 B2 f9} for definitions partitioning $A \cup B$ into three regions $U,V,W.$
Note that $\Area(U) = \log^2  n$ and more than half of the points $v$ in $U$ have $\sigma(v) > 1$.  For  these points $v$
Lemma \ref{lem:B1B2 measure} implies that 
$f(v)  = \Theta\left(  \frac 1 {\delta^2} \right)$ and thus  $\mu(A) =\Theta\left( \frac {\log^2 n} {\delta^2}\right)
= \Omega\left( \frac {\log^2 n} {n}\right)$ and  also $\mu(A) = 
O\left( \frac {\log^4 n} {n}\right)$.

Futher note that $\Area(V) = O(1)$ so $\mu(V) = O\left( \frac  1 {\delta^2}\right) =O\left( \frac {\log^2 n} {n}\right) .$

From Lemma \ref {lem: basic mu}(c)
\begin{equation*}
\label{eq:B1B2UPcase3PU}
 \Pr(  (S_n \cap U)= \emptyset) =  \left(1 - \mu(U)\right)^n
 \le  e^{- c \log^2 n},
\end{equation*}
for some constant $c >0.$  Since {\em any} point in $U$ dominates {\em all} points in $W,$
$(S_n \cap U))\not = \emptyset$ implies $\MAXSN \cap U= \emptyset.$  Thus
\begin{equation*}
\label{eq:B1B2UP5b}
\EXP{|\MAXSN\cap  W| }  \le n \Pr(  (S_n \cap U)= \emptyset) = o(1).
\end{equation*}
This implies
\begin{eqnarray*}
\EXP{|\MAXSN \cap (A \cup B)|} &=&\EXP{|\MAXSN \cap U|} + \EXP{|\MAXSN \cap V|} + \EXP{|\MAXSN \cap W|}\\
&\le& \EXP{|S_n \cap U|} + \EXP{|S_n \cap V|} +o(1)\\
&=& n \mu(U) + n \mu(V) + o(1) = O(\log^4 n).
\end{eqnarray*}

\par\noindent\underline{(b) $ v \in C:$} 

Before starting the main bound we first note how to restrict the analysis to $v \in C$ with  $\sigma(v) \le d  \delta$ for some constant $d  >0.$

\begin{figure}[t]
\begin{center}
\includegraphics[width=9cm]{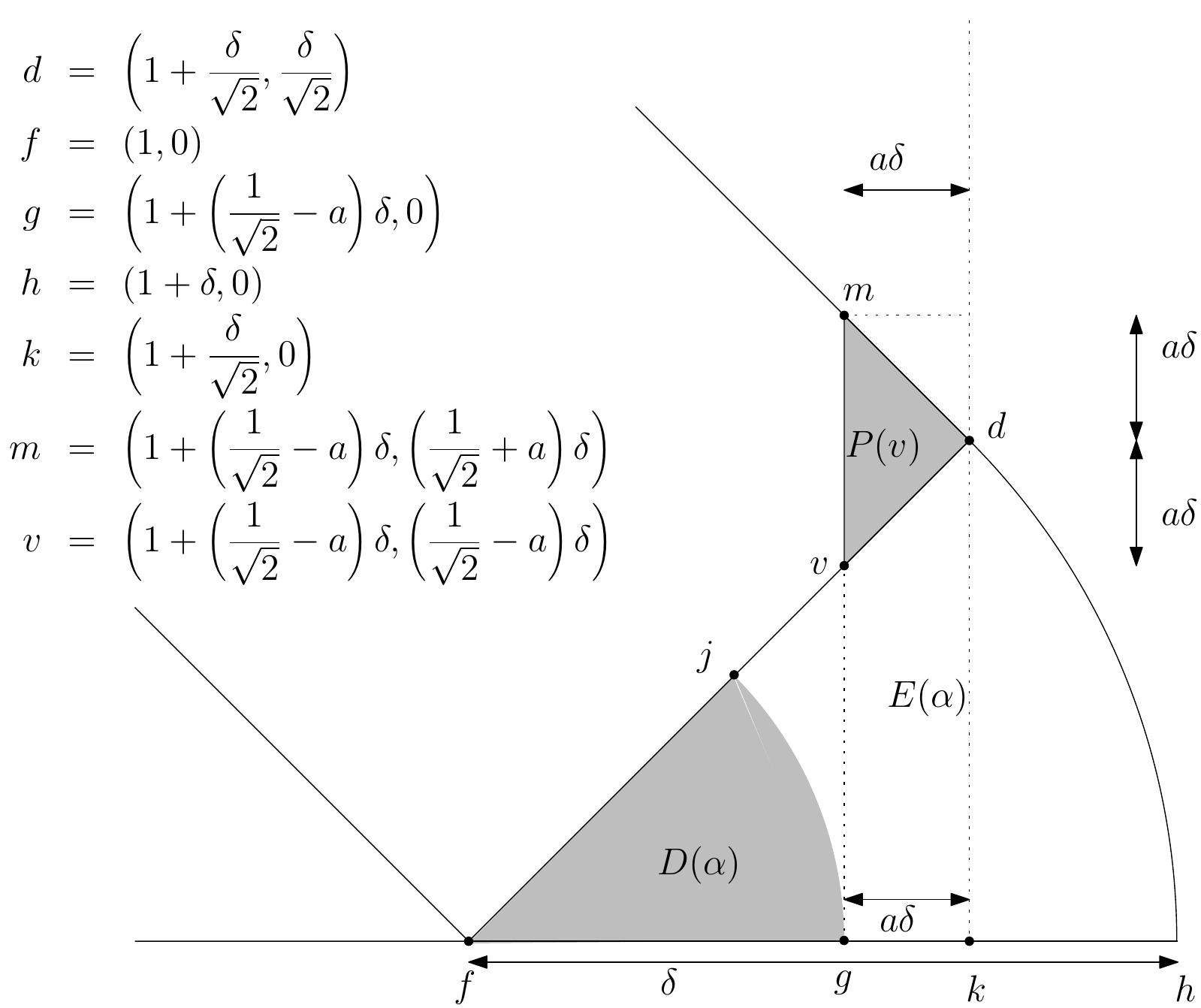}
\end{center}
\caption{ Region $D$. Note that all points in $P(v)$ dominate all points in $D(\alpha)$.
 } %
\label{fig:B2 B2 f4}
\rule{5.5in}{0.5pt}
\end{figure}

\begin{Definition}
Let $0 \le a \le \frac 1 {\sqrt 2}$ be any  fixed constant. Define
\begin{eqnarray*}
D(\alpha)&=&  \left\{ u \in C \,:\,  \sigma(u)  \ge \left(\frac 1 {\sqrt 2} - \alpha \right)  \delta\right\},\\
E(\alpha)&=&  \left\{ u \in C \,:\,  \sigma(u)  \le  \left(\frac 1 {\sqrt 2} - \alpha \right) \delta \right\}.
\end{eqnarray*}
\end{Definition}
We claim that 
\begin{equation}
\label{eq:B1B2UP10}
\EXP{|\MAXSN \cap D(\alpha)|} = O(\log^2 n).
\end{equation}
To see this, first consider the case $\delta \le \frac {\log n} {\sqrt n}$.  From Lemma \ref{lem:easy mu}
$$\mu(D(\alpha)) = O(\Area(D(\alpha))) = O(\delta^2) = O\left(\frac {\log^2 n} n \right).$$
Then
$$\EXP{|\MAXSN \cap D(\alpha)|}
\le \EXP{| S_n \cap D(\alpha)|} =n \mu(D(\alpha)) = O(\log^2 n).$$

For the case $ \delta > \frac {\log n} {\sqrt n}$,
see Figure \ref {fig:B2 B2 f4} for the definitions of points and regions. 
From Eqs.~\ref {eq:B1B2UP1} and \ref{eq:B1B2UP2} together with the fact that $\sigma(v) = \Theta(\delta),$ 
$$\mu(P(v)) = \Theta\left((\Area\Bigl(P(v)) f(v)\right)\Bigr)= \Theta(\delta^2 f(v)).
$$
Recall that $\mu(P(v)) = \Theta(\Area(P(v)) f(v))= \Theta(\delta^2 f(v)).$ 
From Eq.~\ref{eq: B1B2 measureB} in Lemma \ref {lem:B1B2 measure}:  
If  $ \delta \le 1$ then  
$f(v) = \Theta(1)$.  If $ \delta > 1$ then $f(v) = \Theta\left(\frac 1 {\delta^2}\right)$.  In both cases this implies $\mu(P(v)) = \Omega\left( \frac  {\log^2 n}  n \right).$

From Lemma \ref {lem: basic mu}(c)
\begin{equation}
\label{eq:B1B2UPcase4PU}
 \Pr(  (S_n \cap P(v))= \emptyset) =  \left(1 - \mu(P(v))\right)^n
 \le  e^{- c \log^2 n},
\end{equation}
for some constant $c >0.$  Since {\em any} point in $P(v)$ dominates {\em all} points in $D(\alpha),$
$(S_n \cap P(v)))\not = \emptyset$ implies $\MAXSN \cap D(\alpha)= \emptyset.$  Thus
\begin{equation}
\label{eq:B1B2UP5b}
\EXP{|\MAXSN\cap  D(\alpha)| }  \le n \Pr(  (S_n \cap P(v))= \emptyset) = o(1),
\end{equation}
and we have completed the proof of Eq.~\ref{eq:B1B2UP10}  for all cases.
This implies that, for any fixed $ \alpha \le \frac 1 {\sqrt 2},$
\begin{equation}
\label{eq:B1B2UP11}
\EXP{|\MAXSN\cap  C| } =   \EXP{|\MAXSN\cap  E(\alpha)| } + O(\log^2 n). 
\end{equation}

Fix $ \alpha \le \frac 1 {\sqrt 2}$ and set $d = 1 - \left(\frac 1 {\sqrt 2} - \alpha\right)  >0.$  This restricts the analysis to $v \in C$ with  $\sigma(v) \le d  \delta$.

Our analysis will require the following definition and Lemma which are generalizations of the the structures and proofs techniques introduced in Lemma \ref{lem: B2B2UB}.
\begin{Definition}
Let $0 \le \theta_1 \le \theta_2 \le \frac \pi 4.$  Define the $\bar  \theta_1,\, \bar \theta_2$ {\em wedge} as
$$\Wedge\left(\bar \theta_1, \bar \theta_2\right) =  \left\{ v \in E(\alpha) \,:\, \theta_1 <\theta \le \theta_2\right\}.
$$
See Fig.~\ref{fig:B2 B2 b f3} for an example.
\end{Definition}

\begin{Lemma}\label{lem:B1B2bWedge}
Let $W=\Wedge({\bar \theta}_1, {\bar \theta}_2)$ where $0 \le \bar \theta_1 \le \bar \theta_2 < \frac \pi 4.$

Given $m > 0$, set   $\hat\theta= (\bar \theta_2 - \bar \theta_1)/m$.
For $i\in [0,m],$
set $\theta_i =\bar  \theta_1 + i \hat \theta$ and, 
$$\forall i \in [0,m-1],\quad R_i = \Wedge(\theta_i,\theta_{i+1})
\quad\mbox{and}\quad
R_i(\sigma) = \{v \in R_i \,:\,  \sigma(v) \le \delta\}.
$$
Now set $\bar \sigma =  \hat \theta \delta$.  Given that  $\bar \sigma \le d \delta$ and 
$\forall i,\quad \mu(R_i(\bar \sigma)) = O(1/n)$,
then
$$\EXP{\left|MAX(S_n) \cap W\right|}  =  O(m) + O(\log^2 n). $$
\end{Lemma}
The proof of this lemma is deferred until the end of this section.

We now analyze $\EXP{|\MAXSN\cap  E(\alpha)| }$ by splitting into two cases:  (i)  $\delta \le 1$ and (ii)  $\delta \ge 1.$

\medskip

\par\noindent\underline{b(i) $  \delta  \le 1:$}

\begin{figure}[t]
\begin{center}
\includegraphics[width=9cm]{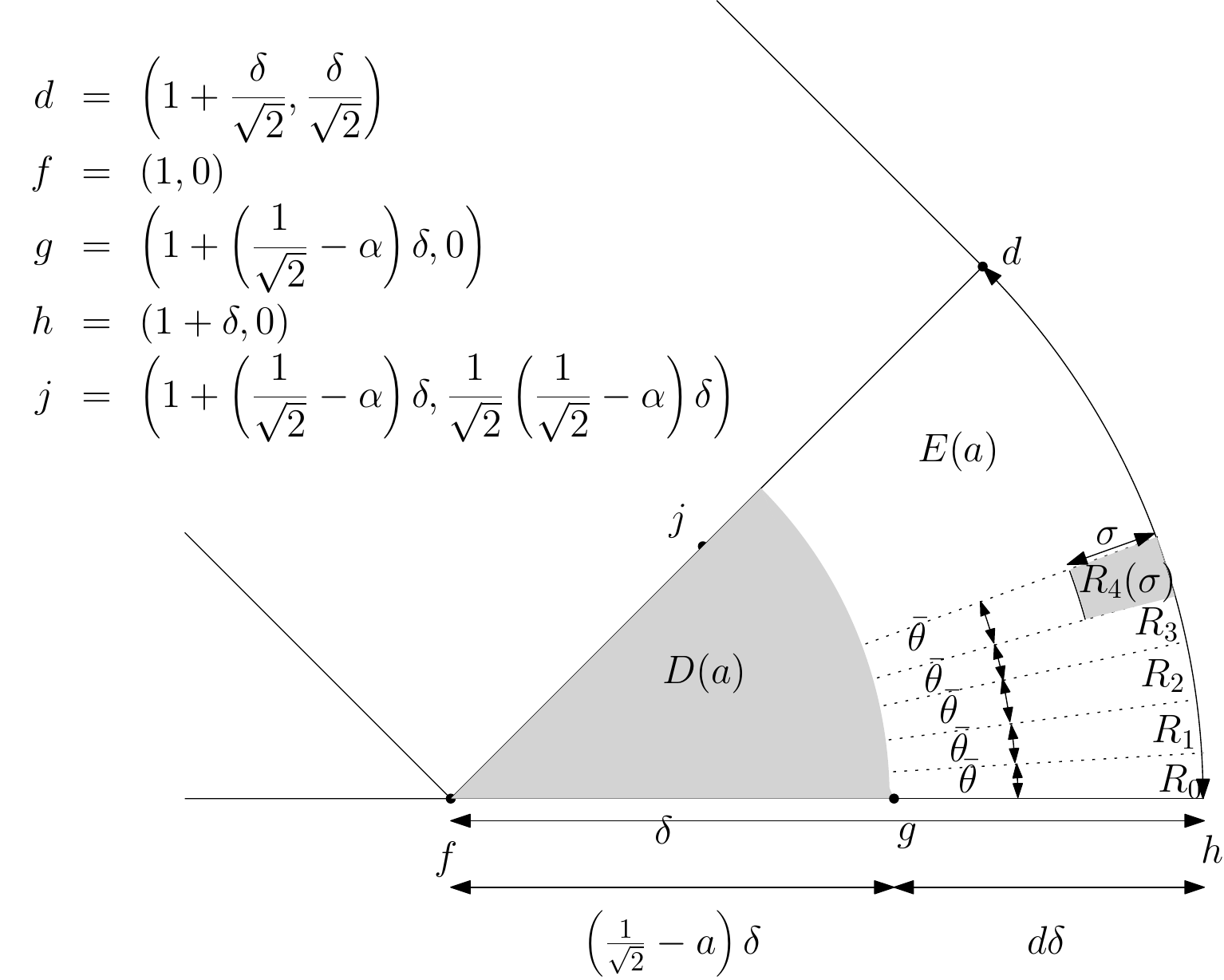}
\end{center}
\caption{ Illustration of the decomposition of $E(a)$ into the $R_i$ in the proof of b(i). 
 } %
\label{fig:B2 B2 f4b}
\rule{5.5in}{0.5pt}
\end{figure}

Set $\bar\theta_1 = 0,$  $\bar\theta_2=\frac \pi 4$ and 
 $m = \left\lfloor    \frac   {n^{2/7}}  {\delta^{3/7}}   \right\rfloor$.

  Let $R_i$ and $R_i(\theta)$ be as defined by Lemma \ref {lem:B1B2bWedge}.
Since for large enough $n,$   $\bar \sigma \le d \delta$, to use that lemma, we only need to show that  $\mu(R_i(\bar \sigma)) = \Theta (1/n).$

To see this, let $v\in B \cap  C$, i.e., $\theta(v) = \pi/4$.  Then, from Lemma \ref {lem:B1B2 measure} 
\begin{equation}
f(v) = O\left( 
                     \left( 
                          \frac  {\sigma(v) } {\delta}    
                      \right)^{3/2}
                \right).
\end{equation}
Corollary \ref {cor: B1B2theta} then implies that 
\begin{equation}
\forall v \in E(\alpha),\quad f(v) = O\left( 
                 \left(  
                          \frac  {\sigma(v) } {\delta}    
                      \right)^{3/2}
                \right).
\end{equation}

The same integration as in the proof of Lemma \ref{lem: B2B2UB} now yields
\begin{equation}
\label{eq:25fmuR}
\mu(R_i(\sigma))  = O
 \left(      \frac  1 m     
\int_{r=0} ^{\sigma} \frac {r^{3/2}} {\delta^{3/2}} d r
\right)
=
O
 \left( \frac {\sigma^{5/2}} {m \delta^{3/2}}
\right).
\end{equation}

In particular,  for $ \bar\sigma = \Theta(1/m)$ this evaluates out to 
\begin{equation}
\mu(R_i(\bar \sigma)) = \Theta (1/n).
\end{equation}

Applying Lemma  \ref {lem:B1B2bWedge} then proves that

$$\EXP{\left|MAX(S_n) \cap \Wedge\left(0, \frac \pi 4 -  \frac \pi  {4m} \right)\right|}  =  O(m) + O(\log^2 n) =
O\left(\frac   {n^{2/7}}  {\delta^{3/7}} \right). $$
\medskip

\par\noindent\underline{b(ii) $  \delta  \ge 1:$}

\begin{figure}[t]
\begin{center}
\includegraphics[width=12cm]{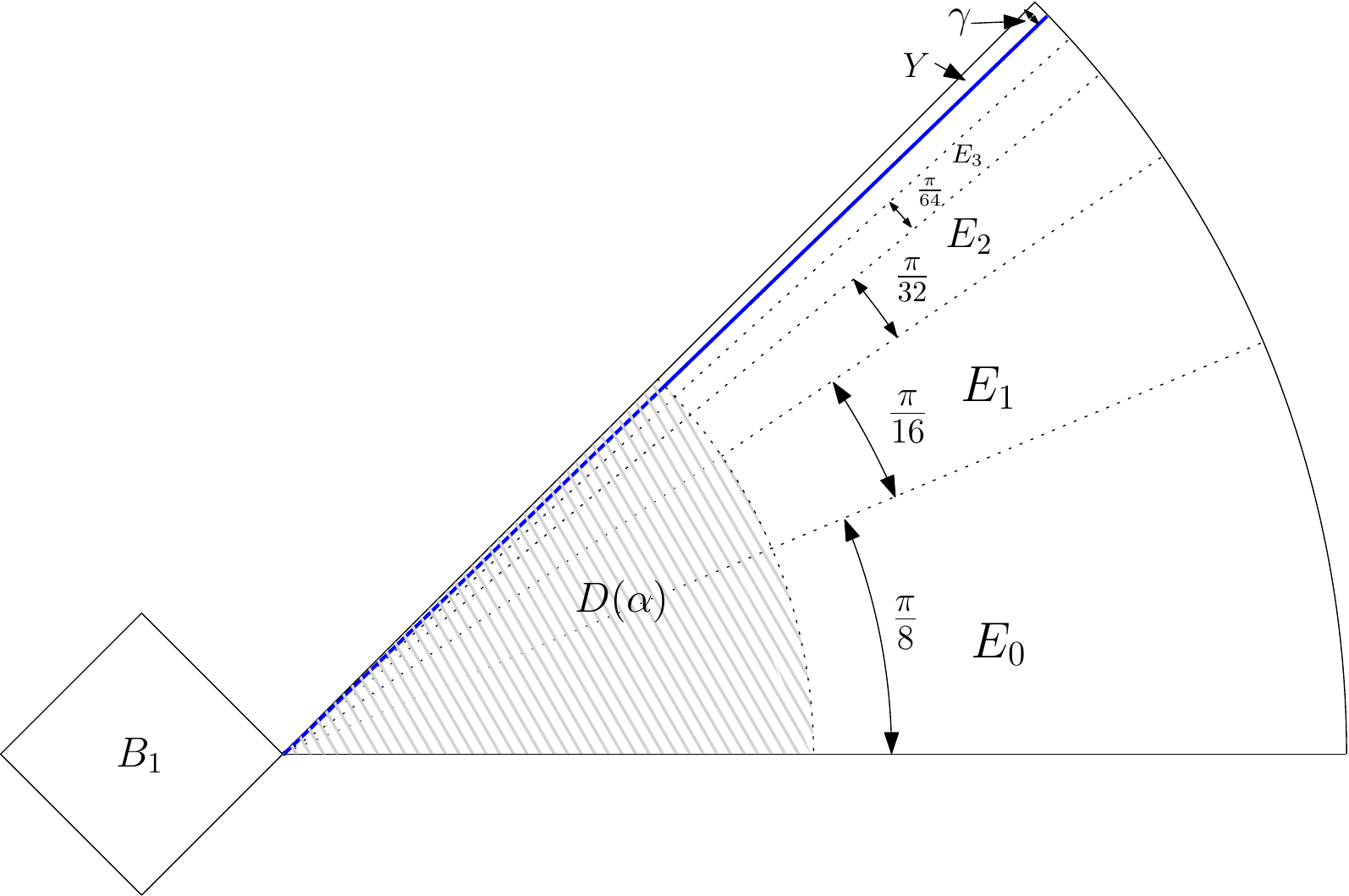}
\end{center}
\caption{Illustrates the further partition of $E(\alpha)$ into smaller wedges $Y$  and $E_0,E_1,E_2,E_3,\ldots$
where $E_i = \Wedge\left(\frac \pi 4- \frac \pi {2^i+2},\, \frac \pi 4- \frac \pi {2^i+3}\right).$
 } %
\label{fig:B2 B2 b f3}
\rule{5.5in}{0.5pt}
\end{figure}

Set $\gamma_0 =\delta^{-2/7} n^{-1/7}.$ 
Further define angles $\theta_i$ and wedges $E_i$ by
$$\theta_0=0,\, \theta_1=\frac \pi 8,\quad
\forall i \ge  1,\   \theta_{i+1} = \theta_i + \frac 1 {2^i} \cdot \frac \pi 8= \frac \pi 4 - \frac {\pi}{ 2^{i+2}},\quad  \theta'_i = \theta_i + \frac 1 {2^{i}} \frac \pi 8,$$
$$E_i = \Wedge(\theta_i,\theta_{i+1}). 
$$
Finally, set 
$$
 k = \max\left\{i \,:\,  \theta_i < \frac \pi 4 - \gamma_0\right\}
\quad \mbox{and} \quad
Y= \Wedge\left(\theta_{k+1}, \frac \pi 4\right).
$$
Then
\begin{equation}
\label{eq:B1B2UP13}
\EXP{\left|\MAXSN  \cap E(\alpha)\right|} \le \sum_{i=0}^{k} \EXP{|\MAXSN \cap E_i|} +  \EXP{|\MAXSN \cap Y|}.
\end{equation}
{\small \em Note:  The reason for the inequality instead of  an equality is that  $E_k$ and $Y$ might overlap.}

We now analyze each item on the right hand side separately and then add them all together to complete the proof.

\medskip

\par\noindent\underline{$E_0$}:
Apply Lemma \ref {lem:B1B2bWedge} to $E_0$ with 
$\bar\theta_1 = \theta_0,$ $\bar\theta_2 = \theta_1$ and $m = \left \lceil \sqrt \delta n^{1/4} \right \rceil $, for some constant $c >0.$
Note that $\hat \theta = \Theta\left( \frac  1 {\sqrt \delta n^{1/4}} \right)$ so 
$\bar \sigma = \hat \theta  \delta = \Theta\left(\frac {d \sqrt \delta} {n^{1/4}}\right).$ 

Note that $ \delta \le \sqrt n$ implies  $\bar \sigma \le 1$ and also, for later use, $ \bar \sigma \le  d \delta.$

For $ v \in E_1$,  $\gamma = \Omega(1)$.  Let  $\sigma = \sigma(v)$.  If $\sigma \le 1$ then
Eq.~\ref{eq: B1B2 measureC} of Lemma \ref{lem:B1B2 measure} implies 
$$f(v) = \Theta \left(  \frac { \sigma^2}  {\delta^2}\right).$$

The same type of integration as in the proof of Lemma \ref{lem: B2B2UB} now yields (assuming  $\sigma \le d \delta$)
\begin{equation}
\label{eq:B1B2UB17}
\mu(R_i(\sigma))  = O
 \left(      \frac  1 m     
\int_{r=0} ^{\sigma} (\delta -r) \frac {r^2} {\delta^{2}} d r
\right)
=
O
 \left( \frac {\sigma^{3}} {m \delta}
\right).
\end{equation}
This yields  $\mu(R_i(\bar \sigma)) = O(1/n)$.  Applying  Lemma \ref {lem:B1B2bWedge} to $E_1$ yields
\begin{equation}
\label{eq:B1B2UP16}
\EXP{|\MAXSN \cap E_i|} =  O(m)  + O(\log^2 n) = O\left( \sqrt \delta n^{1/4}\right).
\end{equation}

\medskip

\par\noindent\underline{$E_i,\, i >0$}:

Let 
$v \in E(\alpha)$.  For simplicity in the remainder of the proof, set  $\gamma = \gamma(v) = \pi /4 - \theta(v)$ and $\sigma= \sigma(v)$.

Let  $x_1 >0$ (with value to be specified later). Set
$$\sigma_\gamma= \frac {\sqrt \delta} {n^{1/4}}  (x_1\gamma)^{1/4}.$$

To ensure that the second case  in Eq.~\ref{eq: B1B2 measureC} of Lemma \ref{lem:B1B2 measure} holds for $\sigma_\gamma$ we need 
 $\sigma_\gamma  \delta = O((\delta \gamma)^2)$, i.e, $\sigma_\gamma \le x_2 \delta  \gamma^2$ for some fixed $x_2> 0.$ Note
\begin{eqnarray*}
\sigma_\gamma \le x_2 \delta \gamma^2  &\Leftrightarrow&  x_1^{1/4} \frac {\sqrt \delta} {n^{1/4}} \delta \gamma^{1/4} \le x_2 \delta  \gamma^2 \\
&\Leftrightarrow&  x_1^{1/4}  \frac 1 {\sqrt \delta n^{1/4}} \le x_2 \gamma^{7/4}\\
&\Leftrightarrow&  \left( \frac {x_1} {x_2^{4}} \right)^{1/7}
  \frac 1 { \delta^{2/7} n^{1/7}} \le \gamma.
\end{eqnarray*}
Assume further that $x_1,x_2$ satisfy $\frac {x_1} {x_2^{4}} <1$.
Now, if 
$v \in \Wedge\left(\frac \pi 8,\frac \pi 4 - \gamma_0\right)$ then
$$ \gamma_0 \le \gamma \le \frac \pi 4.$$ 
In particular, suppose $ \sigma<  \sigma_\gamma.$ Then the above implies
\begin{eqnarray*}
\gamma >  \left( \frac {x_1} {x_2^{4}} \right)^{1/7} \gamma_0  &   \Rightarrow&  \gamma \ge \left( \frac {x_1} {x_2^{4}} \right)^{1/7}
  \frac 1 { \delta^{2/7} n^{1/7}} \\
 & \Rightarrow & \sigma \le \sigma_\gamma \le x_2 \delta \gamma^2\\
  & \Rightarrow& \mbox{ the second case  in Eq.~\ref{eq: B1B2 measureC}  applies}\\
   & \Rightarrow& f(v) = \Theta\left(\frac  {\sigma^2}  {\delta^2} \frac 1 \gamma  \right).
\end{eqnarray*}

Recall that $E_i = \Wedge(\theta_i, \theta_{i+1})$ where 
$\theta'_{i+1} - \theta_i = \frac {\pi }{2^{i+3}}.$

Let $v \in E_i$ and set $\gamma = \frac \pi 4 - \theta(v)$.
By definition 
$$ \frac \pi {2^{i+3}} \le \gamma \le \frac \pi {2^{i+2}}.$$
Set
$$\sigma_i =\frac {\sqrt \delta} {n^{1/4}}  \left( \frac \pi {2^{i+2}} \right)^{1/4}.$$
Then 
$$\sigma_i=  \frac {\sqrt \delta} {n^{1/4}}  (x_1\gamma)^{1/4},$$
where $1 \le x_1 \le 2.$  

Now fix  $x_2 =2.$  By the argument above,  all  $v \in E_i$ with $\sigma \le \sigma_i$, satisfy the $v$ in the second case  of Eq.~\ref{eq: B1B2 measureC}.  So 
\begin{equation}
\label{eq:B1B2bCPf}
f(v) = \Theta \left(  \frac { \sigma^2}  {\delta^2}  \frac  1 \gamma \right)
= \Theta \left(  \frac { \sigma^2}  {\delta^2}   2^i  \right).
\end{equation}

Apply Lemma \ref {lem:B1B2bWedge} to $E_i$ with 
$\bar\theta_1 = \theta_i,$ $\bar\theta_2 = \theta_{i+1}$ and $m =m_i= \left \lceil c  \frac 1 {2^{3i/4}} \sqrt \delta n^{1/4} \right \rceil$,
for some constant $c >0.$
Note that 
$$\hat \theta = \Theta\left( \frac  {\theta_{i+1} - \theta_i} m \right)
= \Theta\left( \frac 1 {{2^{i/4}  \sqrt \delta n^{1/4}}}  \right),$$
therefore
$\bar \sigma = \hat \theta  \delta = \Theta\left(\frac { \sqrt \delta} {2^{1/4} n^{1/4}}\right).$  Choose $c$ large enough so that 
$\bar \sigma  \le  \frac {d \sqrt \delta} {n^{1/4}}$.

Following the same logic as in Eq.~\ref{eq:B1B2UB17}
\begin{equation}
\mu(R_i(\sigma))  = O
 \left(      \hat \theta   
\int_{r=0} ^{\sigma} (\delta -r) \frac {r^2} {\delta^{2}} 2^i d r
\right)
=
O
 \left( \hat \theta \frac {\sigma^{3}} { \delta} 2^i\right).
\end{equation}
Using $\bar \sigma = \hat \theta  \delta$ permits evaluating
$$\mu(R_i(\bar \sigma)) = O\left( \hat \theta \frac {\bar \sigma^{3}} { \delta} 2^i
\right) = O\left( \frac {\bar \sigma^{4}} { \delta^2} 2^i
\right) = O\left( \frac 1 n
\right). 
$$
Applying  Lemma \ref {lem:B1B2bWedge} to $E_i$ yields
\begin{equation}
\label{eq:B1B2UP18}
\EXP{|\MAXSN \cap E_i|} = O(m_i ) + \log^2 n) = O\left( 2^{-3i/4} \sqrt \delta n^{1/4}\right) + O(\log^2 n).
\end{equation}

\medskip

\par\noindent\underline{$Y$}:

Recall that $Y = \Wedge\left(\frac \pi 4 - \gamma_0,\frac \pi 4\right)$, where $\gamma_0 =\delta^{-2/7} n^{-1/7}.$   We apply Lemma \ref {lem:B1B2bWedge} to $Y$ with 
$\bar\theta_1 = \frac \pi 4 - \gamma_0 $ and  $\bar\theta_2= \frac \pi 4$. There are two cases to be analyzed separately:
(A)  $1 \le \delta \le  n^{1/5}$  and  (B) $ n^{1/5} \le \delta \le \sqrt n.$

\medskip

\par\noindent\underline{(A)  $1 \le \delta \le  n^{1/5}:$ }
Set 
$$m = \left\lceil c n^{1/7} \delta^{2/7} \right\rceil,\quad
\hat \theta = \frac {\theta_2 - \theta_1} m =  \frac {\gamma_0} m = \Theta\left(\delta^{-4/7} n^{-2/7}\right),\quad
\bar \sigma = \hat \theta \delta = \Theta\left(\frac {\delta^{3/7}} {n^{2/7}}\right),
$$
for some constant $c >0.$

First observe that because $\delta \le n^{-1/5}$, if $c$ is large enough, then 
$\bar \sigma <1.$
Next observe that
$$\delta \gamma_0^2 =\delta \delta^{-4/7} n^{-2/7} = \bar \sigma.$$
Therefore, if $v \in Y$ with $\sigma(v) = \Theta( \bar \sigma)$, then this is the second  case  in Eq.~\ref{eq: B1B2 measureC} of Lemma \ref{lem:B1B2 measure}. 
Since $\delta \le n^{1/5},$
$$\sqrt {\delta \bar \sigma}  = \sqrt \frac {\delta^{10/7}} {n^{2/7}} = \left(\frac {\delta^5} n \right)^{1/7} \le 1.$$
Plugging into Eq.~\ref{eq: B1B2 measureC} gives 
$$f(v) =\Theta\left( \frac { \bar \sigma} { \delta^2}  \sqrt {\delta \bar \sigma} \right)=\Theta\left( \frac {\bar \sigma^{3/2}} {\delta^{3/2}}\right).$$
Thus, for all $v \in Y$ with $\sigma(v) \le \bar \sigma,$  
$$f(v) =O\left( \frac {\bar \sigma^{3/2}} {\delta^{3/2}}\right).$$

Repeating the
standard integration gives
$$\mu(R_i(\bar \sigma))  = O
 \left(  \hat \theta
\int_{r=0} ^{\bar \sigma}  (\delta - r) \frac {\bar \sigma^{3/2}} {\delta^{3/2}} dr
\right)
=
O
 \left( \hat \theta \delta \bar \sigma  \frac {{\bar \sigma}^{3/2}} {\delta^{3/2}} 
\right)
=
O
 \left( \frac {{\bar \sigma}^{7/2}} {\delta^{3/2}} 
\right)
=  O (1/n).$$
Thus

$$\EXP{\left|MAX(S_n) \cap Y \right|}  =  O(m) + O(\log^2 n) =O\left(n^{1/7} \delta^{2/7}\right)=O\left(\sqrt \delta n^{1/4}\right). $$

\medskip

\par\noindent\underline{(B)  $ n^{1/5} \le \delta \le  \sqrt n:$ }
Set 
$$
m=   \left\lceil   \delta^{1/21} n^{4/21} \right \rceil,\quad 
\hat \theta  =\frac {\gamma_0} m =  \Theta\left(\delta^{-1/3} n^{-1/3}\right),\quad
\bar \sigma =  \hat \theta \delta =\Theta\left(  \frac {\delta^{2/3}} {n^{1/3}}\right),
$$
for some constant $c >0.$
Note that of $c$ is chosen large enough, then   $\bar \sigma \le 1$ and $\bar \sigma \le d \delta.$

Next observe that
\begin{eqnarray*}
\delta \gamma^2 = \frac {\delta^{3/7}} {n^{2/7}} \le  \frac {\delta^{2/3}} {n^{1/3}} = \bar \sigma
&\Leftrightarrow&  
n^{1/21} = n^{\frac 1 3 - \frac 2 7} \le  \delta^{\frac 2 3 -\frac 3 7} = 2\delta^{ 5/21}\\
&\Leftrightarrow&   n^{1/5} \le \delta.
\end{eqnarray*}

So, if $v \in Y$ with $\sigma(v) = \Theta( \bar \sigma)$, this is again
 the second  case  in Eq.~\ref{eq: B1B2 measureC} of Lemma \ref{lem:B1B2 measure}. 
Unlike in case (A)  above, though, $\delta \ge n^{1/5}$ implies that
$$\sqrt {\delta \sigma} = \sqrt \frac {\delta^{5/3}} {n^{1/3}}  = \left( \frac {\delta^5}  n \right)^{1/6} = \Omega(1).$$
Plugging into Eq.~\ref{eq: B1B2 measureC}, then gives
$$f(v) = \Theta \left( \frac { \bar \sigma} { \delta^2} \right).$$
Thus, for all $v \in Y$ with $\sigma(v) \le \bar \sigma,$  
$$f(v) =O\left( \frac {\bar \sigma} {\delta^{2}}\right).$$

Repeating the
standard integration gives
$$\mu(R_i(\bar \sigma))  = O
 \left(  \hat \theta
\int_{r=0} ^{\bar \sigma}  (\delta - r) \frac {\bar \sigma} {\delta^{2}} dr
\right)
=
O
 \left( \hat \theta \delta \bar \sigma  \frac {{\bar \sigma}} {\delta^{2}} 
\right)
=
O
 \left( \frac {{\bar \sigma}^{3}} {\delta^{2}} 
\right)
=  O (1/n).$$
Thus

$$\EXP{\left|MAX(S_n) \cap  Y \right|}  =  O(m) + O(\log^2 n) =O\left(\delta^{1/21} n^{4/21}\right)=O\left(\sqrt \delta n^{1/4}\right). $$

\par\noindent\underline{Combining the pieces:}

We now complete the analysis in part (b) (ii) by combining the pieces together, plugging into Eq.~\ref{eq:B1B2UP13} and using the fact that $k = O(\log n)$ to yield
\begin{eqnarray*}
EXP{\left|\MAXSN  \cap E(\alpha)\right|} &\le& \sum_{i=0}^{k} \EXP{|\MAXSN \cap E_i|} +  \EXP{|\MAXSN \cap Y|}\\
&=&  \sum_{i=0}^{k} \left(O\left( 2^{-3i/4} \sqrt \delta n^{1/4}\right) + O(\log^2 n)\right) +O\left(\sqrt \delta n^{1/4}\right)  \\
&=& O\left(\sqrt \delta n^{1/4}\right) + O(k \log^2 n) + O\left(\sqrt \delta n^{1/4}\right)  \\
&=& O\left(\sqrt \delta n^{1/4}\right).
\end{eqnarray*}
\end{proof}


\begin{proof} of Lemma \ref{lem:B1B2 measure}.

From Lemma \ref{lem: measure integral}, 
\begin{equation}
\label{eq:fB2B1}
f(v) = \Theta\left( \frac {\Area(B_2(v,\delta) \cap B_1)} {\delta^2}\right).
\end{equation}

For each region the proof performs a case-by-case analysis to derive the value of 
$\Area(B_2(v,\delta) \cap B_1).$

\begin{figure}[t]
\begin{center}
\includegraphics[width=7cm]{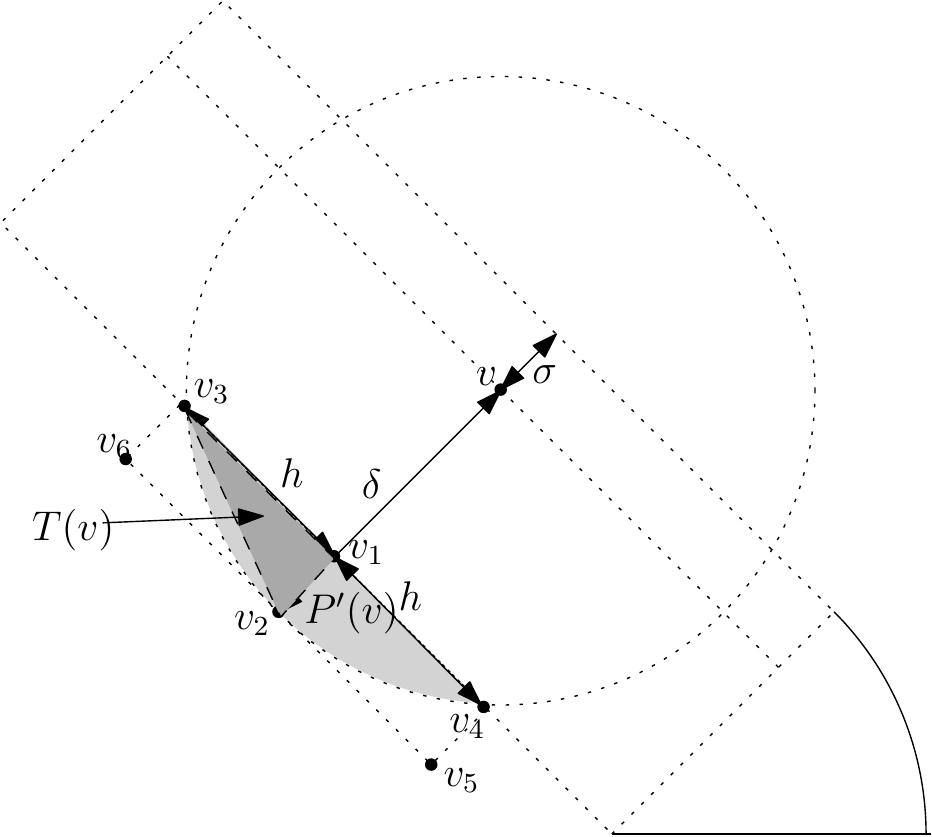}
\end{center}
\caption{Illustration of  the derivation of Lemma \ref{lem:B1B2 measure}  when $v \in B.$
 } %
\label{fig:B2 B2 f2}
\rule{5.5in}{0.5pt}
\end{figure}

\par\noindent\underline{$v \in A:$}

If $\delta \le 1$ then the result follows directly from Lemma \ref{lem:easy mu}(b).

If $\delta >1$  and $v \in A$, then $\Area(B_2(v,\delta) \cap B_1)= \Theta(1)$ and the result follows from
plugging into Eq.~\ref{eq:fB2B1}.

\medskip

\par\noindent\underline{$v \in B$}

Consider Figure \ref {fig:B2 B2 f2}.   
For simplicity, let $\sigma$ denote $\sigma(v).$
 Let $v_3$ and $v_4$ be the upper and lower intersection of $B_2(v,\delta)$ with the line $x+y=1$ (the upper-right boundary of $B_1$).  Now set $v_1= \frac {v_2 +v_3} 2$ to be their midpoint, and $h=||v_1-v_3||.$ 
  Draw the radius of $B_2(v,\delta)$ that passes through $v_1$. Label its other endpoint as $v_2.$
 Finally, let $T$ be the triangle with vertices $v_1,v_2,v_3.$ Straightforward geometric arguments show that
 $$\Area(B_2(v,\delta) \cap B_1)  =
 \Theta(\Area(T \cap B_1))
 = \Theta\Bigl(
 \min\left(h,\frac 1 {\sqrt 2} \right) \cdot  \min(\sigma,\sqrt 2)
 \Bigr).
 $$
 Next, note that
$$h^2 = \delta^2 - (\delta-\sigma)^2  = \delta^2 \left(1 - \left(1 - \frac \sigma \delta\right)^2 \right).$$
Since  $\sigma \le \delta,$
\begin{eqnarray*}
h &=&  \delta \sqrt {  1 - \left( 1 -  \frac \sigma \delta\right)^2 }\\
   &=&   \delta \sqrt {\frac  {2 \sigma   } \delta - \frac  {\sigma^2} {\delta^2}}\\
   &=&  \Theta\left(  \delta \sqrt {\frac \sigma \delta \left( 2 - \frac \sigma \delta \right)} \right) = \Theta( \sqrt {\delta \sigma}).
\end{eqnarray*}
Again since $\sigma \le \delta$,   if  $\sigma \ge \sqrt 2$, then $\delta = \Omega(1)$ and $h = \Omega(1).$  Thus, $\Area(B_2(v,\delta) \cap B_1) = \Theta(1).$
Working through the other cases similarly and plugging the derived values into Eq.~\ref{eq:fB2B1}  yields
$$
f(v) = \Theta\left( \frac {\Area(T \cap B_1)} {\delta^2}\right)
=\left\{
\begin{array}{ll}
\Theta\left(\frac 1 {\delta^2} \right) & \mbox{if  $\sigma = \Omega(1)$},\\
\Theta \left(\frac  \sigma {\delta^2} \right) & \mbox{if $\delta \sigma = \Omega(1)$ and  $\sigma = O(1)$},\\
\Theta  \left( \frac  { \sqrt {\delta \sigma}  \sigma} {\delta^2}  \right)=
\Theta  \left(  \left( \frac  {\sigma} {\delta}  \right)^{3/2} \right)& \mbox{if $\delta \sigma = O(1)$ and  $\sigma = O(1).$}
\end{array}
\right.
$$

\medskip

\par\noindent\underline{$v \in C:$}
Fix $v \in C.$ Set $\sigma = \sigma(v),$  $r = r(v)= \delta - \sigma(v),$ and $\theta=\theta(v).$ By definition,  $0 \le \theta(v) \le \pi/4.$ 
Also,  from the assumption that $ \sigma < d  \delta$ for some $ \delta < 1,$ we know that  $r = \Theta(\delta).$

\medskip
\begin{figure}[t]
\begin{center}
\includegraphics[width=9cm]{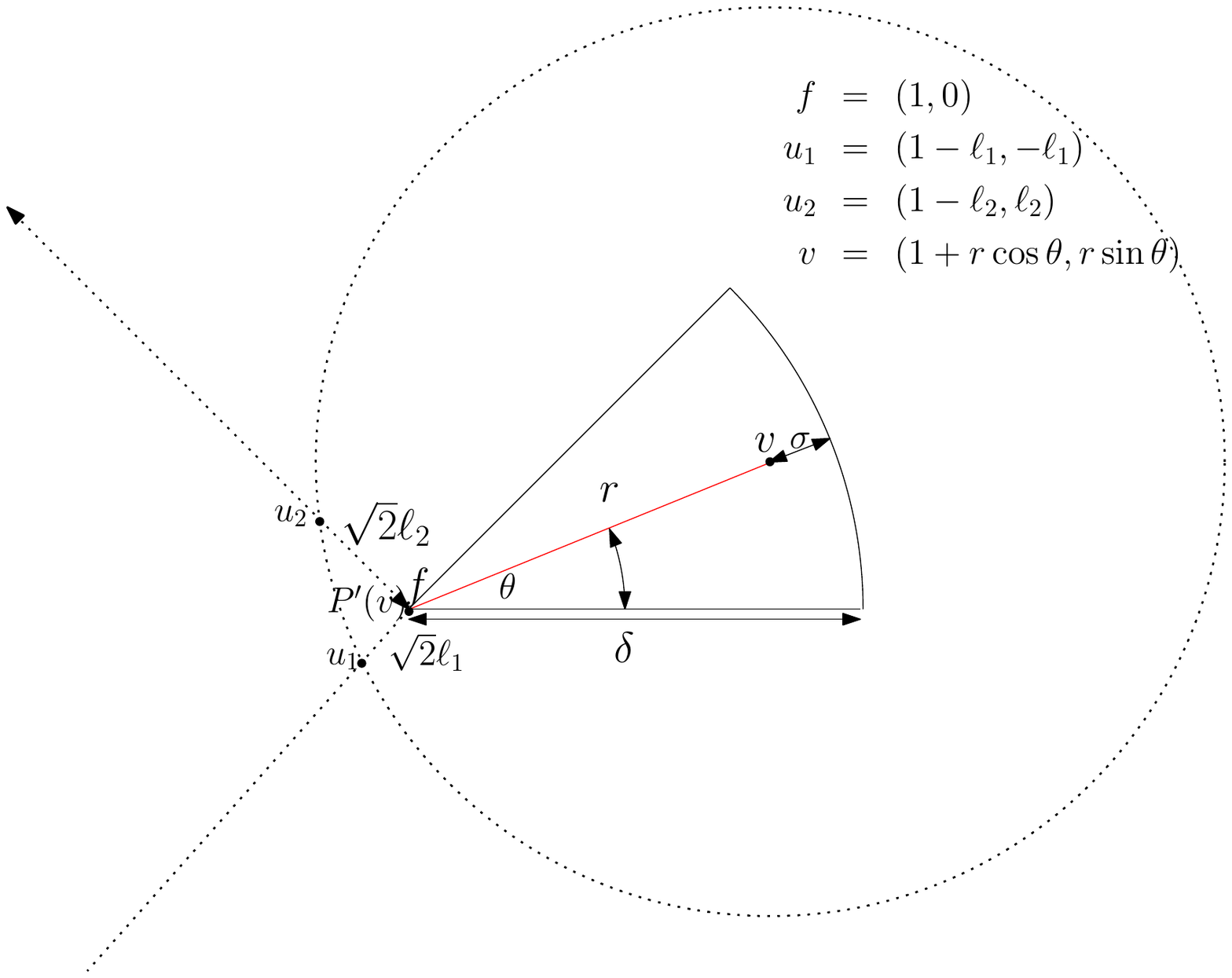}
\end{center}
\caption{Illustration of  the derivation of Lemma \ref{lem:B1B2 measure}  when $v \in C.$ This is the simple case in which $\ell_1, \ell_2 \le 1.$
 } %
\label{fig:B2 B2 f3}
\rule{5.5in}{0.5pt}
\end{figure}

See Fig.~\ref{fig:B2 B2 f3}. Let $u_1, u_2$ be the points of intersection of $B_2(v,\delta) $ with the boundary lines $x+y=1$ and $x-y=1$  of $B_1$ and $f=(1,0).$  Let $\ell_1 = \frac 1 {\sqrt 2}  ||u_1 -f||$ and $\ell_2 = \frac 1 {\sqrt 2} ||u_2-f||$.
 
 Set $P'(v) = B_2(v,\delta) \cap B_1$ and $T$ to be  the triangle formed by $u_1,u_2,f$.  First note that if $\ell_1, \ell_2 \le 1$,  then $T \subset B_1$ and, since 
 $$\Area(T) < \Area(P'(v)) < 2\Area (T),$$
 we find $\Area(P'(v)) = \Theta(\Area(T)) = \Theta( \ell_1 \ell_2).$
Note that if $\ell_1 >1$ or $\ell_2 >1$ then $T \not \subseteq B_1$.  This can be fixed by capping the lengths of the triangle to stay within $B_1,$ ,i.e, if $\ell_1 \ge 1,$ replacing $v_1$ by $(0,-\sqrt 2)$ and if $\ell_2 \ge 1$, replacing $v_2$ by $(0,\sqrt 2).$  Straightforward geometric arguments  show that 
 \begin{equation}
 \label{eq:B1B2fArea}
 f(v) =\Theta \left(\frac { \Area(P'(v))} {\delta^2}   \right)=  
 \Theta\left( \frac {\min(\ell_1,1) \cdot \min( \ell_2,1)} {\delta^2}\right).
 \end{equation}

   Now note that
\begin{eqnarray*}
\delta^2 &=&  ||v - u_1||^2 \\
               &=&  (r \cos \theta + \ell_1)^2  +  (r \sin \theta + \ell_1)^2 \\
               &=& r^2 \cos^2 \theta + 2 r \ell_1 \cos \theta + \ell_1^2 + r^2 \sin^2 \theta + 2 r \ell_1 \sin \theta + \ell_1^2 \\
		&=& r^2 + 2 r  \ell_1 (\cos \theta + \sin \theta) + 2\ell_1^2.
\end{eqnarray*}
Similarly 
\begin{eqnarray*}
\delta^2 &=&  ||v - u_2||^2 \\
               &=&  (r \cos \theta + \ell_2)^2  +  (r \sin \theta - \ell_2)^2 \\
               &=& r^2 \cos^2 \theta + 2 r \ell_2 \cos \theta + \ell_2^2 + r^2 \sin^2 \theta + - r \ell_2 \sin \theta + \ell_2^2 \\
		&=& r^2 + 2 r  \ell_2 (\cos \theta - \sin \theta) + 2\ell_2^2.
\end{eqnarray*}

These  can be rewritten as
\begin{eqnarray}
\Theta(\delta \sigma) &=&
\frac 1 2 (2 \delta - \sigma) \sigma =\frac { \delta^2 - r^2} 2 = r  \ell_1 \beta_\theta + \ell_1^2, \label{eq:B1B2l1}\\
\Theta(\delta \sigma) &=&
\frac 1 2 (2 \delta - \sigma) \sigma = \frac { \delta^2 - r^2} 2 = r  \ell_2 \alpha_\theta + \ell_2^2,\label{eq:B1B2l2}
\end{eqnarray}
where
$$\beta_\theta = \cos \theta + \sin \theta
\quad\mbox{and}\quad
\alpha_\theta=\cos \theta - \sin \theta.
$$

Since $\alpha_\theta, \beta_\theta \ge 0$ and $\sigma \le \delta,$   Eqs.~\ref{eq:B1B2l1} and \ref{eq:B1B2l2} immediately imply     $\ell_1,\ell_2 = O(\delta)$.
Since $0\le \theta(v) \le \pi/4,$  $\beta_\theta=\Theta(1)$ so Eq.~\ref{eq:B1B2l1} can be rewritten as
\begin{equation}
\label{eq:B1B2mul1eq}
\Theta(\delta \sigma) 
 =\Theta( r  \ell_1 + \ell_1^2),
 \end{equation}
which combined with  $r = \Theta(\delta)$, implies
$$\ell_1 = \Theta(\sigma).$$

The analysis of $\ell_2$ is more complicated.
Unlike $\beta_\theta$, $\alpha_\theta$ {\em decreases}  monotonically to $0$ as $\theta \uparrow \pi/4.$ In particular,  $\alpha_\theta\not=\Theta(1)$ so, for values of $\theta$ close enough to $\pi/4$,
 we have $\ell_2\not= \Theta(\sigma).$  Its  analysis requires a more careful derivation.

Set 
$$B(\sigma) = r(\sigma) \alpha_\theta,\quad  C(\sigma) = \frac 1 2  (2 \delta - \sigma) \sigma, \quad \gamma(\theta) = \pi/4 - \theta.$$
We will write $B,C,\gamma$ when  $\sigma,\theta$ values are fixed.

Set $\gamma= \pi/4 - \theta.$  Taking Taylor series around $\pi/4$ yields
$$\alpha_\theta = \cos \theta - \sin \theta = \sqrt 2 \gamma + \Theta(\gamma^2),$$
where the constant implicit in the $\Theta()$ is again only dependent upon $\hat \Theta$.
This implies
$$
B = \Theta(\delta \alpha_\theta) = \theta(\delta \gamma)
\quad\mbox{and}\quad 
C = \Theta(\delta \sigma).$$

From Eq.~\ref{eq:B1B2l2},  $\ell_2$ is the solution to the quadratic equation 
$$\ell_2^2 + B \ell_2 - C = 0,$$
 therefore
\begin{equation}
\label{eq:B1B2muquad}
 \ell_2 = \frac { -B + \sqrt {B^2 + 4C}} 2.
 \end{equation}

Fix any arbitrary $ c >0.$  There are now two possibilities.  
\begin{enumerate}
\item If    $\delta \sigma  \ge c (\delta \gamma)^2$:
 Eq.~\ref{eq:B1B2muquad}    then yields
$\ell_2 = \Theta (\sqrt {\delta \sigma}).$
\item If    $\delta \sigma  < c (\delta \gamma)^2$:
Then Eq.~\ref{eq:B1B2muquad}  yields   
\begin{eqnarray*}
\ell_2 &=& \frac 1 2 \left(  -B  + B \sqrt {1 + \frac  {4 C }{ B^2}} \right)\\
&=&\Theta\left(  -\delta \gamma  + \delta \gamma \left(1 + \Theta\left(\frac  {4 \delta \sigma} { (\delta \gamma)^2}\right)\right) \right)\\ 
&=&  \Theta \left( \frac  {\sigma}  {\gamma}   \right).
\end{eqnarray*}
\end{enumerate}
The above can be rewritten as 

$\sigma \le \delta$ so $\delta \sigma  <  \delta^2 \le  c \delta^2  \gamma^2$ where $c = 1/(\hat \theta)^2.$ Thus, for all cases of $\delta \sigma  \le c (\delta \gamma)^2$,  $\Theta \left( \frac  {\sigma}  {\gamma}   \right)$. This implies
$$
\ell_2 =
\left\{
\begin{array}{ll}
\Theta (\sqrt {\delta \sigma}) &  \mbox{if $\sigma  \ge c (\delta \gamma^2)$},\\
 \Theta \left( \frac  {\sigma}  {\gamma}   \right) &  \mbox{if $\sigma  < c (\delta \gamma^2)$}.
\end{array}
\right.
$$
Note that if this equation is correct for {\em any } $c >0$, it is correct for {\em all}  $c >0.$  Thus, using the fact that
$\ell_1 = \Theta(\sigma)$ and plugging into Eq.~\ref{eq:B1B2fArea}  we have proven

\begin{equation}
 f(v) = 
 \left\{
 \begin{array}{ll}
 \Theta\left( \frac {\min(\sigma,1) \cdot \min\left( \sqrt {\delta \sigma},1)\right)} {\delta^2}\right) & \mbox{if  $\sigma  = \Omega(\delta \gamma^2)$},\\
 \Theta\left( \frac {\min(\sigma,1) \cdot \min\left( \frac  {\sigma}  {\gamma},1\right)} {\delta^2}\right) & \mbox{if  $\sigma  = O(\delta \gamma^2)$}.\\
 \end{array}
 \right.
 \end{equation}

\end{proof}


\begin{proof} of Lemma \ref{lem:B1B2bWedge}.

Until otherwise stated,  $i < m-1$ is assumed

This part of the analysis will  completely mimic the upper bound derivation for $\Ballpq 2 2 $ in Lemma \ref{lem: B2B2UB}, but with the $\Theta(\,)$ term replaced by an $O(\,)$ term in many places. In particular, we   use the sweep lemma almost exactly the same way as was done in the proof of Lemma \ref{lem: B2B2UB}.
We therefore only state the main items and skip the details which are exactly the same as in Lemma \ref{lem: B2B2UB}.

Let $A_i(\sigma) = R_{i+1}(\sigma).$  Set 
$$v_i =( (1+(\delta-\sigma) \cos \theta_{i+2},\,  (\delta-\sigma)  \sin \theta_{i+2}).$$
This is the leftmost point of $A_i(\sigma)$ on the boundary line between $R_{i+1}$ and $R_{i+2}.$
Now drop a vertical line from $v_i$ to the $x$-axis and let $v'_i$ be the point at which it intersects  the boundary between $R_i$ and $R_{i-1}.$ Let $\sigma'= \sigma(v_i').$    By construction

$$(\delta-\sigma)  \cos \theta_{i+2} = v.x = v'.x = (\delta -\sigma') \cos \theta_{i}.$$

But
\begin{eqnarray*}
\cos \theta_{i+2} &=&  \cos \left(\theta_i + 2 \hat \theta \right)\\
			    &=&  \cos \theta_i \cos (2 \hat \theta) - \sin \theta_i \sin (2 \hat \theta)\\
                               &=&  (1  +O(\hat \theta)) \cos \theta_i  - O(\hat \theta)\\
                               & =&  \cos \theta_i    - O(\hat \theta).
\end{eqnarray*}
Then
$$
(\delta-\sigma)\left(  \cos \theta_i    - O(\hat \theta)\right) = (\delta -\sigma')\cos \theta_{i}.
$$
Because the $R_i$ are in the first octant, $\theta_i$ is bounded away from $\pi/2$  so $\cos \theta_i$ is bounded away from $0.$  Thus, 
$\sigma' - \sigma = O(\delta \hat \theta).$

Now define 
$$B'_i (\sigma)=  \{v \in R_i \,:\,  v.x \le v_i.x\},\quad    B_i(\sigma) = R_i \setminus B'_i(\sigma).$$

By construction
\begin{itemize}
\item  $\forall \sigma,$ every point in $A_i(\sigma)$ dominates every point in $B'_i(\sigma)$.
\item $\forall \sigma,$  $B_i(\sigma) \subset R_i(\sigma')$ $\Rightarrow$  $\mu(B_i(\sigma)) \le  \mu(R_i(\sigma'))$.
\item  From the arguments above,  $\bar \sigma = \delta \hat \theta$  $\Rightarrow$    $\bar\sigma' = \Theta(\delta \hat \theta) =\Theta(\bar \sigma)$.
\item From Corollary \ref{cor: B1B2theta} and the definition of the $R_i$,\\ $\bar\sigma'  =\Theta(\bar \sigma)$  $\Rightarrow$ 
$\mu(R_i(\bar\sigma')) = O(\mu(R_i(\bar\sigma))).$

\end{itemize}
  Combining these points yields
$$\Rightarrow \mu(B_i(\bar\sigma)) \le  \mu(R_i(\bar\sigma')) = O(\mu(R_i(\bar\sigma))) = O(1/n).$$

Note that  $R_i = B'_i(\sigma) \cup B_i(\sigma)$  and thus
\begin{eqnarray*}
\EXP{|MAX(S_n) \cap R_i|}  &\le&  \EXP{|MAX(S_n) \cap B'_i(\bar \sigma)|}   + \EXP{|MAX(S_n) \cap B_i(\bar \sigma)|}  \\
					 &\le&  \EXP{|MAX(S_n) \cap B'_i(\bar \sigma)|}   + \EXP{|S_n \cap B_i(\bar \sigma)|}  \\
				          &=&     \EXP{|MAX(S_n) \cap B'_i(\bar \sigma)|}   + n \mu(B_i(\bar \sigma)) \\
					&=&      \EXP{|MAX(S_n) \cap B'_i(\bar \sigma)|}   + O(1).		
\end{eqnarray*}

To show that $\EXP{|MAX(S_n) \cap R_i|} =O(1)$ it now suffices to prove that
$\EXP{|MAX(S_n) \cap B'_i(\bar \sigma)|} = O(1),$  which we will now do via the sweep lemma using $\sigma$ as the sweep parameter.

Set 
$$
\begin{array}{ccccccl}
A &=& R_{i+1}, &                   \quad & A(\sigma) &=& A_i(\sigma),\\
B &=& B'_{i}({\bar \sigma}), & \quad & B(\sigma) &=& B_i(\sigma) \cap B'_i(\bar \sigma).
\end{array}
$$
By the previous discussion,
\begin{itemize}
\item  Every point in $B \setminus B(\sigma) $ is dominated by every point in $A_i(\sigma).$
\item If   $\sigma \le \bar \sigma$ $\Rightarrow $  $B(\sigma) = \emptyset.$
\item If $\sigma > \bar \sigma $$ \Rightarrow$  $B(\sigma)  \subset B_i(\sigma) \subset R_i(\sigma').$
\end{itemize}
Thus,  if  $\sigma \le \bar \sigma$, then
$$ \mu(B(\sigma)) = 0 = O(\mu(A(\sigma)),$$
while if $ \sigma > \bar \sigma$ then, using the fact that $\sigma' - \sigma = O(\delta \hat \theta) = O(\bar \sigma).$
\begin{eqnarray*}
\mu(B(\sigma))  &= & O (\mu(R_i(\sigma'))) \\
                             &= & O (\mu(R_{i+1}(\sigma'))) \\
                           &=  & O (\mu(R_{i+1}(\sigma + O(\bar \sigma) )))  \\
                           &=  & O (\mu(R_{i+1}(\sigma) )),
  \end{eqnarray*}
where the last inequality comes from       the fact that $ \sigma > \bar \sigma$  so $\sigma + O(\bar \sigma) = O(\sigma).$

This explicitly satisfies the conditions of the sweep Lemma  and thus,
$\EXP{|MAX(S_n) \cap B'_i(\bar \sigma)|} = O(1)$  and thus
$$\EXP{|MAX(S_n) \cap R_i|} =O(1)$$
  
Since $R_i$, $i-0,1,\ldots, m-2$, partition $\Wedge\left(\theta_1,\theta_2-\hat \theta\right)$, this implies
 \begin{equation}\label{eq:B1B2UP20}
 \EXP{\left|MAX(S_n) \cap \Wedge\left(\theta_1,\theta_2-\hat \theta\right)\right|}  =  \sum_{i=0}^{m-2}\EXP{|MAX(S_n) \cap R_i|}  = O(m).
 \end{equation}
 
 We now examine $R_{m-1}.$  Define  $\hat R = R_{m-1}$ and $\hat R(\sigma) = R_{m-1}(\sigma).$
 
 To  complete the proof of the lemma,  it  only remains to prove that 
\begin{equation}\label{eq:B1B2UP21}
 \EXP{\left|MAX(S_n) \cap \hat R \right|}  =  O(\log^2 n). 
 \end{equation}
 
  Set 
$$ v=( (1+(\delta-\sigma) \cos \bar \theta_2,\,  (\delta-\sigma)  \sin \bar \theta_{2}).$$
See Figure \ref{fig:B2 B2 b f4}.
This is the leftmost point of $A_i(\sigma)$ on the top boundary line of $\Wedge(\bar \theta_1,\bar\theta_2).$
Now drop a vertical line from $v$ to the $x$-axis and let $v'$ be the point at which it intersects  the boundary between $R_{m-1}$ and $R_{m-2}.$ Let $\sigma'= \sigma(v').$    By construction  every point in $\hat R(\sigma)$ dominates every point in 
$\hat R \setminus \hat R(\sigma').$

$$(\delta-\sigma)  \cos \bar \theta_2 = v.x = v'.x = (\delta -\sigma') \cos \left(\bar \theta_2 - \hat \theta\right).$$

But
\begin{eqnarray*}
\cos \left(\bar \theta_2 - \hat \theta\right) 
			    &=&  \cos \bar \theta_2 \cos (\hat \theta) + \sin \bar \theta_2 \sin ( \hat \theta)\\
                               &=&  (1  +O(\hat \theta)) \cos \bar \theta_2  - O(\hat \theta)\\
                               & =&  \cos \bar \theta_2    - O(\hat \theta).
\end{eqnarray*}
 Then
$$
(\delta-\sigma)\cos \bar \theta_2   = (\delta -\sigma') \left(\cos \bar \theta_2    - O(\hat \theta) \right).
$$
Because the $R_i$ are in the first octant, $\theta_i$ is bounded away from $\pi/2$  so $\cos \theta_i$ is bounded away from $0.$  Thus, 
$\sigma' - \sigma = O(\delta \hat \theta) = O (\bar \sigma).$

If $\mu(\hat R) \le \frac{ \log^2 n} n $ then Eq.~\ref{eq:B1B2UP21} is trivially correct because
$$
 \EXP{\left|MAX(S_n) \cap \hat R \right|}  \le \EXP{\left|S_n \cap \hat R \right|} = n \mu \left(\hat R\right) =
 O(\log^2 n). 
 $$
We therefore assume that $\mu\left(\hat R\right)  >  \log^2 n$ and let $\hat \sigma$ be the unique value  such that 
$\mu\left(\hat R(\hat \sigma)\right) =  \frac{ \log^2 n} n.$ Since $\mu\left(\hat R(\bar \sigma) \right)= \Theta(1/n)$,  $\hat \sigma = \Omega (\bar \sigma).$ Thus
$\hat \sigma' = \Theta(\hat \sigma)$ and, from Corollary \ref {cor: B1B2theta},
$$ \mu\left(\hat R(\hat \sigma')\right)  = \Theta\left(\mu\left(\hat R(\hat \sigma\right)\right) = \Theta \left(\frac{ \log^2 n} n  \right).$$

Now
\begin{eqnarray*}
\EXP{\left|\MAXSN \cap \hat R \right|}
&=& \EXP{\left|\MAXSN \cap \hat R\left(\hat \sigma'\right) \right|} +
\EXP{\left|\MAXSN \cap \left( \hat R \setminus \hat R\left(\hat \sigma'\right)\right) \right|} \\
&\le& \EXP{\left|S_n\cap \hat R\left(\hat \sigma'\right)  \right|} 
+
\EXP{\left|\MAXSN \cap \left( \hat R \setminus \hat R\left(\hat \sigma'\right)\right) \right|}.
\end{eqnarray*}

 Now  
 $$\EXP{\left|S_n\cap \hat R\left(\hat \sigma'\right)  \right|}  = n \mu\left(\hat R(\hat \sigma')\right) = O(\log^2 n).$$

From Lemma \ref {lem: basic mu}(c)
$$
 \Pr(  (S_n \cap \hat R(\hat \sigma))= \emptyset) =  \left(1 - \mu\left(\hat R(\hat \sigma)\right)\right)^n
 \le  e^{- c \log^2 n},
$$
for some constant $c >0.$  Since, as noted, {\em any} point in $\hat R(\hat \sigma))$
 dominates {\em all} points in $\hat R \setminus \hat R\left(\hat \sigma'\right), $
 $(S_n \cap \hat R(\hat \sigma)))\not = \emptyset$ implies $\MAXSN \cap \left(\hat R \setminus \hat R\left(\hat \sigma'\right)\right)= \emptyset.$  Thus
$$
\EXP{\left|\MAXSN \cap \left( \hat R \setminus \hat R\left(\hat \sigma'\right)\right) \right|} \le n \Pr(  (S_n \cap \hat R(\hat \sigma))= \emptyset) = o(1).
$$
 This yields
 $$\EXP{\left|\MAXSN \cap \hat R \right|} \le O(\log^2 n) + o(1) = O(\log^2 n),$$
 proving Eq.\ref{eq:B1B2UP21} and thus completing the proof of the lemma.

 \end{proof}

\begin{figure}[t]
\begin{center}
\includegraphics[width=10cm]{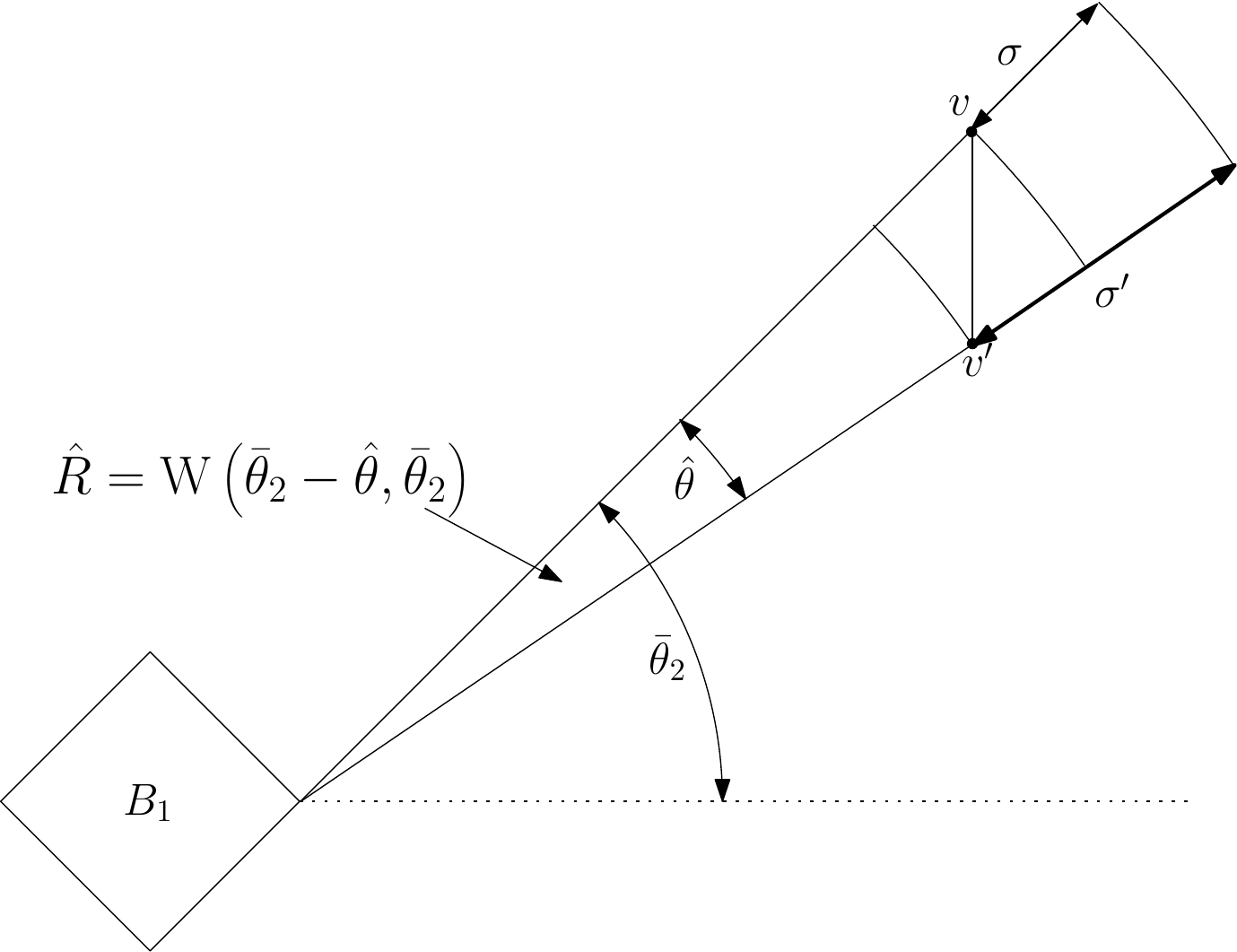}
\end{center}
\caption{Illustrates the last piece of the proof of  Lemma \ref{lem:B1B2bWedge}.  The diagram is not to scale.  $\hat \theta$ is actually  very small, and decreasing with $n.$  Note that any point in $\hat R(\sigma')$ will dominate {\em all} points in $\hat R \setminus \hat R (\sigma'+\sigma'').$
 } %
\label{fig:B2 B2 b f4}
\rule{5.5in}{0.5pt}
\end{figure}
\section{Analysis of $\Ballpq \infty q$} \label{sec: BiBq}
This section derives cells (iv)(c-d) in Theorem \ref{thm: main}, that is,  if $n$  points are chosen from $\bfd = \Ballpq \infty q$,  for any fixed $q \in [1,\infty)$  and  $\frac 1  {\sqrt n} \le \delta \le 1$, then 
$\EMN = \Theta \left(  \ln n + \sqrt \delta n^{1/4}  \right)$.
Note that this implies that $\EMN= \Theta(\ln n)$ for 
$\frac 1 {\sqrt n} \le \delta \le \frac {\ln^2 n} {\sqrt n}$ and it only starts increasing as $\delta > \frac {\ln^2 n} {\sqrt n}.$

Applying Lemma  \ref {lem: scaling} also provides  a full analysis for $\Ballpq q \infty.$

Corollary \ref {cor: Quadrants} states that 
$$\EMN =   \EXP {|Q_1 \cap \MAX(S_n)|}  + O(1),$$
so our analysis will be restricted to the upper-right  quadrant $Q_1$.
Our approach will be to
\begin{enumerate}
\item State a convenient expression for $\mu(u)$ (proof delayed until later).
\item Derive a lower bound using Lemma \ref {lem: lb} by defining an appropriate  pairwise disjoint collection of dominant regions. 
\item Derive an upper bound by partitioning $D$ into appropriate regions and applying the sweep Lemma.
\end{enumerate}

Note that this section differs slightly from the previous ones in that it allows $q$ to be {\em any} value in $[1,\infty]$. The Lemmas and Theorems are correct for all such $q,$ but the constants implicit in the $\Theta(),$ $O()$ and $\Omega()$ terms will depend explicitly upon $q.$  We caution the reader that the diagrams are all drawn for the case $q=2$ and other cases might look quite different.

\begin{Definition}
\lab{def:BiBq Regions}
Let $D_1= (B_\infty+ \delta B_q) \cap Q_1$ be the support in the first quadrant.  
Partition $D_1$ into the following four regions
\begin{align*}
A &= \left\{ u \in D_1\,:\, u.x \le 1,\, u.y \le 1\right\}, &   C &= \left\{ u \in D_1\,:\, u.x \ge 1,\, u.y \le 1\right\}, \\
B&= \left\{ u \in D_1\,:\, u.x \ge 1,\, u.y \ge 1\right\}, &   C &= \left\{ u \in D_1\,:\, u.x \le 1,\, u.y \ge 1\right\}. \\
\end{align*}
If  $v \in B$, define
$$ \alpha(v) = \left(\delta^q - (p.y)^q\right)^{1/q} - p.x
\quad\mbox{and}\quad
 \beta(v) = \left(\delta^q - (p.x)^q\right)^{1/q} - p.y.
$$
If $v \in C$, define
$$ \alpha(v) = 1+ \delta - p.x
\quad\mbox{and}\quad
 \beta(v) =  \left(\delta^q - (\delta -\alpha(v))^q\right)^{1/q}.
$$
$\alpha(v)$ is the
horizontal distance from $v$ to the right border of $D$.  If $v \in B,$ $\beta(v)$ is the distance from $v$ to the vertical border but  if $v \in C,$ $ \beta(v)$ is  just a function of $\alpha(v).$
\end{Definition}

\begin{figure}[t]
\begin{center}
\includegraphics[width=4in]{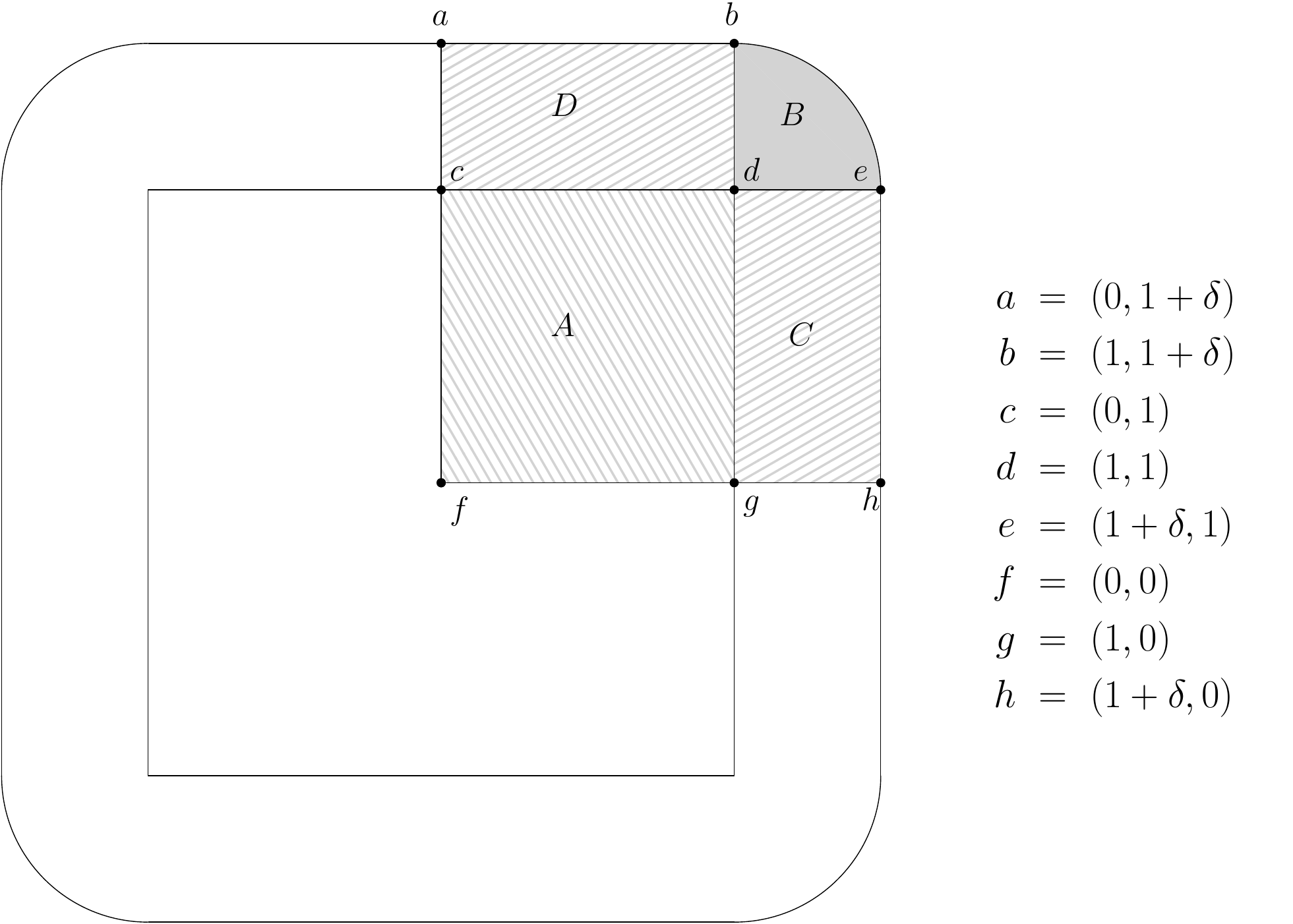}
\end{center}

\caption{ The regions introduced  in Definition \ref{def:BiBq Regions}. }
\lab{fig:LB-infty_Scema}
\rule{5.5in}{0.5pt}
\end{figure}

\begin{Lemma}
\lab{lem:BiBq measure}
Let $\bfd = \Ballpq \infty q$ 
and $v \in D_1.$
\begin{itemize}
\item     
 If $v \in A$, then 
\begin{equation}
\label{eq:BiBqm4}
f(v) =\left\{
\begin{array}{ll}
\Theta (1) & \mbox{if $ \delta \le 1$},\\
\Theta\left( \frac  1 {\delta^2} \right)& \mbox{if $ \delta > 1$}.
\end{array}
\right.
\end{equation}
\item  If $v \in B$, then 
\begin{eqnarray}
f(v) &=&\Theta\left(
\frac    {  \min\Bigl(\alpha(v),\, 1\Bigr) \cdot  \min\Bigl(\beta(v),\, 1\Bigr)} {\delta^2}   \label{eq:BiBqm5}\right).\\[0.1in]
\mu\left(P(v)\right) &=& \Theta\left(\alpha(v) \beta(v) f(v)  \label{eq:BiBqm6} \right).
\end{eqnarray}
\item If $v \in C$, then 
\begin{equation}
\label{eq:BiBqm7}
f(v) =\Theta\left(
\frac    {\min\Bigl(\alpha(v),\, 1\Bigr) \cdot  \min\Bigl(\beta(v),\, 1\Bigr)} {\delta^2}
\right).
\end{equation}
\end{itemize}
\end{Lemma}

\begin{figure}[t]
\begin{center}
\includegraphics[width=4in]{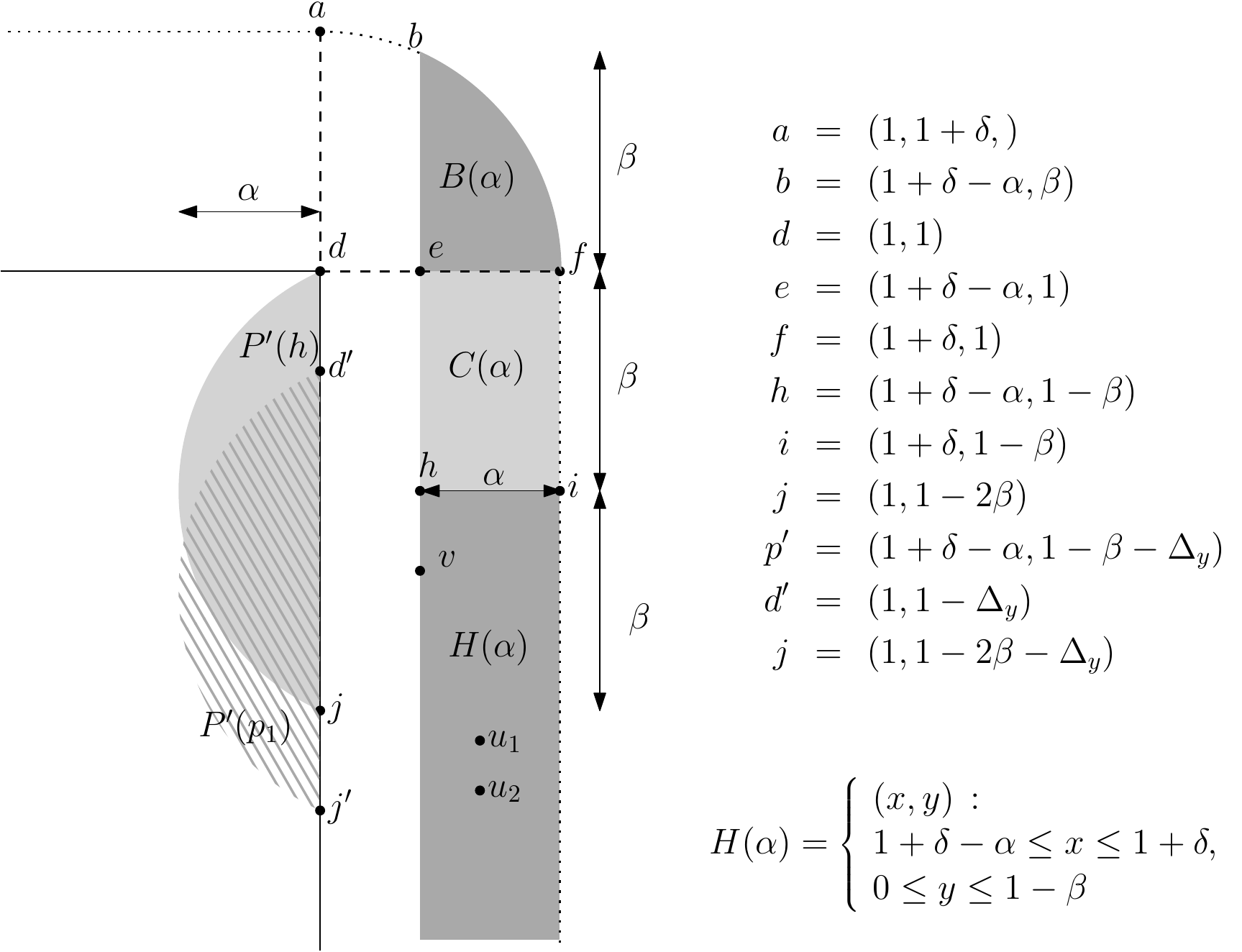}
\end{center}
\caption{ Illustration of Lemma~\ref{lem:LB-infty-BQ-side} and its proof.   Note that only the top of $H(\alpha)$ is shown.  This illustrates the case $\delta \le 1;$ the case $\delta > 1$ looks different.
}
\lab{fig:LB-infty_Vb}
\rule{5.5in}{0.5pt}
\end{figure}

\begin{Lemma}
\lab{lem:LB-infty-BQ-side}
Fix  $\alpha \le \delta$. Define $\beta = \left(\delta^q - (\delta -\alpha)^q\right)^{1/q}$. See Fig.~\ref{fig:LB-infty_Vb}. 

Set $x=1 + \delta - \alpha$  and points  $e=(x,1)$, $h = (x, \max(0,1-\beta))$. Further define regions
\begin{eqnarray*}
B(\alpha) & =& P(e),\\
C(\alpha) & =& P(h) \setminus P(e),\\
H(\alpha) & =& \left\{ u \in C \,:\, \alpha(u)\le \alpha(x),\  1+\delta - \alpha \le u.x \le 1 + \delta, \quad  u.y \le \max(1 - \beta,0)\right\}.
\end{eqnarray*}
Note that if $\beta \ge 1$, then $H(\alpha) = \emptyset.$

\par\noindent (A)  if $\beta \le 1,$  then 
\begin{enumerate}
\item   If $v=(x,v.y)$ where  $ -1 + \beta \le p.y \le 1 + \beta$
then 
$f(h) = f(v)$.
\item  If $u_1,u_2 \in H(\alpha)$ with  $u_1.x = u_2.x,$ then  $f(u_1) = f(u_2).$
\item  $\mu(C(\alpha)) = \Theta(\mu(B(\alpha))) = \Theta\left( \frac {\alpha^2 \beta^2} {\delta^2} \right).$
\end{enumerate}

\par\noindent  (B)  if $\beta > 1,$  then 
$$\mu(C(\alpha)) = O(\mu(B(\alpha))). $$
\end{Lemma}
\begin{proof}
We first assume $\beta \le 1$ and prove (A).

Recall from Lemma \ref{lem: measure integral} that 
$$f(u)  =  \frac {\Area(P'(u))} {a_\infty\,  a_q \, \delta^2},
\quad\mbox{where}\quad P'(u) = B_q(u,\delta) \cap B_\infty.$$
By the definition of $h$ and the fact that $\beta \le 1,$
$$P'(h)  =
\left\{ u\,:\,  ||u-h||_q \le \delta,\,  |u.x| \le 1
\right\},
$$
which is the light gray sector in Fig. \ref{fig:LB-infty_Vb}. By basic geometric arguments, $P'(v) = P'(h) - (0,h.y -v.y),$ i.e., $P'(h)$ is shifted down appropriately. 
Thus, $ \Area(P'(h)) = \Area(P'(v))$ and  $f(h) = f(v)$, proving 1.\

Next set $\alpha' = 1 + \delta - u_1.x$ and  $\beta' = \left(\delta^q - (\alpha')^q\right)^{1/q}.$  Since $\alpha' \le \alpha,$  $\beta' \le \beta$, we can define $h'$ associated  with $\alpha',\beta'$ and then apply   part 1  twice to get $f(u_1) = f(h') = f(u_2)$, proving part 2.

For part 3 first  note that from Eq.~\ref{eq:BiBqm6} in Lemma \ref{lem:BiBq measure},  
$\mu(B(\alpha)) = \mu(P(e)) = \Theta\left( \frac {\alpha^2 \beta^2} {\delta^2} \right).$

As $\Area (C(\alpha)) = \alpha \beta$ and 
$\forall u  \in C(\alpha),   \alpha(u) \le \alpha$ then, from Eq.~\ref {eq:BiBqm7}  
$f(u) = O \left( \frac {\alpha \beta} {\delta^2} \right).$
Thus
$$\mu(C(\alpha))) = \int_{u \in C(\alpha)} f(u) du \le \Area (C(\alpha) ) \cdot \max_{u \in C(\alpha)} f(u) = O \left( \frac {\alpha^2 \beta^2} {\delta^2} \right).$$
For the other direction,  set
$$C'(\alpha) =\left\{u \in C(\alpha)\,:\, \alpha(u) \ge \alpha/2\right\}.$$
For $\alpha' \in [0,\delta]$,  $\beta(\alpha')$ is a monotonically increasing concave function.  Since $\beta(0) =0$, 
$\beta(\alpha/2)  \ge \frac 1 2 \beta(\alpha).$ Thus,  
$\forall \alpha' \in C'(\alpha),$  $\beta(\alpha' )=\ge  \beta/2.$ 
 Then
\begin{eqnarray*}
\mu(C(\alpha))) &=& \int_{u \in C(\alpha)} f(u) du
                           = \int_{u \in C'(\alpha)} f(u) du\\
                           &\ge&  \Area(C'(\alpha)   \cdot  \min\{f(u) \,:\, u \in C'(\alpha\}
                           \ge  \frac {\alpha \beta} 2 \frac {\alpha  \beta} 4 = \Omega\left( \frac {\alpha^2 \beta^2} {\delta^2} \right),
\end{eqnarray*}
and the proof of (A) is complete.

Part (B) follows  directly from Eqs.~\ref{eq:BiBqm6} and \ref{eq:BiBqm7}.
\end{proof}

\begin{Lemma}
\label{lem: BiBqLB}
Let $S_n$ be $n$ points chosen from the distribution $\bfd = \Ballpq \infty q$ with $\frac 1 {\sqrt n} \le \delta \le \sqrt n$. Then
$$\EMN = \Omega \left(  \ln n + \sqrt \delta n^{1/4}  \right).$$ 
\end{Lemma}

\begin{proof}
The proof splits into two parts;
\begin{itemize}
\item[(a)]  $\forall \delta$ satisfying   $ \frac 1 {\sqrt n} \le \delta \le \sqrt n$, \quad  $\EXP{|\MAXSN \cap B|} = \Omega\left(\sqrt \delta n^{1/4}\right).$ 
\item[(b)]  $\forall \delta$ satisfying   $ \frac 1 {\sqrt n} \le \delta \le \frac {\log^2} {\sqrt n}$, \quad  $\EXP{|\MAXSN \cap C|} = \Omega\left(\ln n\right).$ 
\end{itemize}
By symmetry,   $\EXP{|\MAXSN \cap D|} =\EXP{|\MAXSN \cap B|} $.  Combining this with 
(a) and  (b)  proves the lemma. 

\bigskip

\par\noindent\underline{(a)  $\EXP{|\MAXSN \cap B|}$ for  $ \frac 1 {\sqrt n} \le \delta \le \sqrt n$:}

\begin{figure}[t]
\begin{center}
\includegraphics[width=3in]{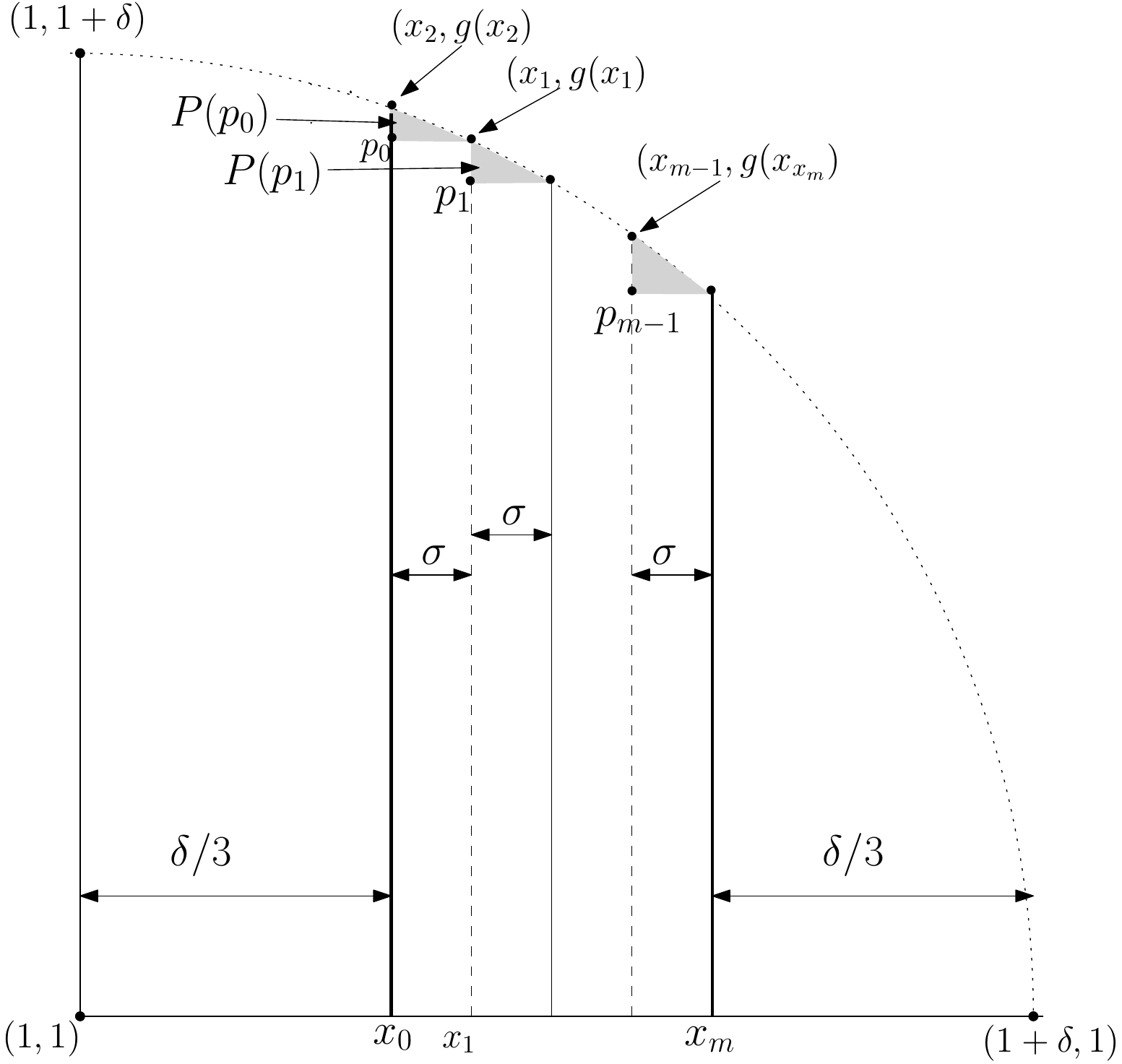}
\end{center}
\caption{ Illustration of  case (a) of the upper bound.     
}
\label{fig:BiBqUPa}
\rule{5.5in}{0.5pt}
\end{figure}
Fix $\sigma$  as  a  value to be determined later. Set $m = \lfloor \delta/(3 \sigma) \rfloor$ and, 
for $ i \ge 0,$ 
$$ x_i =1 + \frac \delta 3 +  i \sigma,
\quad \quad
y_i = g(x_{i+1}),
\quad\quad  p_i = (x_i,y_i).
$$
Then
$$g(x) =  1 + \left(\delta^q - (x-1)^q\right)^{1/q}$$
is the equation of the top boundary of $B.$
We will show that for $0 \le i < j \le m-1,$  with  an appropriate choice of $\sigma$,   then $P(p_i) \cap P(p_j)= \emptyset $ and $\mu(P(p_i)) = \Omega(1/n)$.
Lemma \ref {lem: lb}
 would then imply  that $\EXP{M_n} = \Omega(m)$.

By construction, 
$$\alpha(p_i)  = x_{i+1} - x_i = \sigma,\quad \beta(p_i)  = y_i - y_{i+1} = g(x_i) - g(x_{i+1}),$$
and $P(p_i) \cap P(p_j)= \emptyset $ .

The Theorem of the Mean states that,  for $ 1 < x_i < x_{i+1} <1 + \delta,$ 
\begin{equation}
\label{eq:asigma relation}
\beta(p_i)= |g(x_{i+1}) - g(x_i)| = |g'(z) | |x_{i+1} - x_i| = |g'(z) | \sigma,
\end{equation}
for some $x_i \le z \le x_{i+1}$. Direct differentiation yields
\begin{equation}
\label{eq:gderiv}
g'(x)=\frac{q(x-1)^{q-1}}{q}(\delta^q-(x-1)^q)^{\frac{1}{q}-1}.
\end{equation}

Observe that $\forall x \in  [ 1+ \delta/3, 1 + 2 \delta/3]$,  then
$$|g'x| = \Theta\left(   \delta^{q-1} \left(\delta^q\right)^{\frac 1 q -1}\right)  = \Theta(1),$$ 
and thus,  by
Eq.~\ref{eq:asigma relation},   $\beta(p_i)= \Theta (\sigma)$,
where the constant in the $\Theta(\, )$ depends only upon $q.$

Set $\sigma = c'\sqrt \delta/ n^{1/4}$ for some small $c' > 0.$
For all $i \in [0,m-1]$,  Eq.~\ref{eq:BiBqm6} in Lemma \ref{lem:BiBq measure} yields\footnote{This  implicitly requires that $\beta < \min(1,\delta/3)$, but this is guaranteed by $\delta > 1 /\sqrt n$ and appropriate choice of $c'.$}
\begin{eqnarray*}
\mu(P(p_i)) &=&
\Theta\left(\alpha(v) \beta(v) f(v)\right)\\
&=& \Theta\left( \frac {\alpha^2(v) \beta^2(v) } {\delta^2}\right)\\
&=&  \Theta\left( \frac {\sigma^4} {\delta^2}\right) = \Theta\left( \frac 1 n \right).
\end{eqnarray*}
Lemma \ref {lem: lb}
then yields   
 $$\EXP{M_n} = \Omega(m) = \Omega\left( \frac \delta  \sigma  \right) =\Omega\left( \sqrt \delta n^{1/4}\right).$$

\medskip

\par\noindent\underline{(b)  $\EXP{|\MAXSN \cap C|}$ for  $ \frac 1 {\sqrt n} \le \delta \le \frac {\log^2} {\sqrt n}$:}

For fixed   $\alpha' \le \delta$ define $\beta_{\alpha'}  = \left(\delta^q - (\delta - \alpha')^q\right)^{1/q}$ and point $e(\alpha') = (1 + \delta - \alpha,1)$.  This is illustrated in  Fig.~\ref{fig:LB-infty_Vb}.

 Note that if $\alpha'= \delta$ then $\beta_{\alpha'} =\delta$ and
$\frac {(\alpha')^2 (\beta_{\alpha'})^2} {\delta^2} = \Omega \left( \frac 1 n \right)$.  Since $\beta_{\alpha'}$ increases monotonically with $\alpha'$ we can, for suitably small but fixed  $c$ always  find $\alpha \le \delta$ such that $\frac {\alpha^2 \beta^2} {\delta^2} = \frac c n$.  Fix this $\alpha$ and set $\beta = \beta_\alpha,$ $e = e(\alpha).$

\medskip

Set   $B(\alpha) = P(e),$  and let $C(\alpha)$ and $H(\alpha)$ be  as defined in Lemma \ref{lem:LB-infty-BQ-side}.

  To simplify the  remainder of the proof, set  $T(\alpha) =  B(\alpha) \cup C(\alpha).$ 
  Points in $H(\alpha)$ can only be dominated by points in  $T(\alpha)$. So, if  $S_n \cap T(\alpha)$  is empty then
$$\MAXSN \cap  H(\alpha)  =  \MAX\left(S_n  \cap  H(\alpha)\right).$$
 This implies 
\begin{eqnarray*}
\EXP{|\MAX(S_n)|}  &\ge & \EXP{|\MAX(S_n \cap H(\alpha))|}  \\
			&\ge & \EXP{|\MAX(S_n \cap H(\alpha)) \,\Bigm|\,   S_n \cap T(\alpha) = \emptyset}  \cdot \PR{S_n \cap T(\alpha) = \emptyset}.
\end{eqnarray*}
  
  From  Lemma \ref{lem:LB-infty-BQ-side}.
$$\mu(T(\alpha)) = \Theta\left(\mu(B(\alpha) + \mu(C(\alpha)\right)  =  \Theta\left(\frac {\alpha^2 \beta^2} {\delta^2} \right)= \Theta\left(\frac 1 n\right).$$
Thus,  from Lemma \ref{lem: basic mu}, 
$$\Pr(T(\alpha)  \cap S_n = \emptyset) =   \left(1 - \mu(T(\alpha))\right)^n = \Theta(1).$$
Let 
$$
\hat f(v) =
\left\{
\begin{array}{ll}
\frac {f(v)} {1 - \mu(T(\alpha))} & \mbox{ if $v \in (B_\infty \cap \delta B_q)\setminus T(\alpha)$},\\
0 & \mbox{ otherwise}.
\end{array}
\right.
$$
Note that $\hat f(v)$ is the density of the original distribution $\bfd$ 
{\em  conditioned on  the point chosen not being in $T(\alpha)$}.

Thus 
{\em  Conditioned on $S_n \cap T(\alpha)=\emptyset$}, the distribution of 
$S_n \cap H(\alpha)$,   is equivalent to   the one generated by the following procedure:
\begin{enumerate}
\item Choosing a random variable $X$ from a binomial distribution 
$B\left(n, \hat \mu(H(\alpha)\right).$
\item  For $v \in  H(\alpha)$, setting
$ \bar f(v) = \frac { f(v)} {\mu(H(\alpha))}$, where   $\bar f(v)$ is the conditional  probability density function for choosing a point $v$ from $\Ballpq \infty q$ conditioned on knowing that $v \in H(\alpha)$. 
\item  Choosing $X$ points (in $H(\alpha)$)   from the distribution defined by $\bar f(p).$
\end{enumerate}

In  particular,  point (2) combined with  Lemma \ref {lem:LB-infty-BQ-side}  (3),  implies  that the distribution on $H(\alpha)$ defined by $\bar f(v)$ is only dependent upon the $x$-coordinate of $v,$ i.e., $\bar f(v)$ denotes a distribution in which the $x$ and $y$ coordinates  are independent of each other.

 As stated in the introduction, the number of maxima for $X$ points chosen from such a distribution behaves exactly as if the points are chosen from $\Ball \infty$.  Thus, if $X$ points are chosen using $\bar f(v)$,  the expected number of maxima among them will be $\Theta(\ln X).$   This implies 
 $$\EXP{|\MAX(S_n \cap H(\alpha))| \,\Bigm|\,   S_n \cap T(\alpha) = \emptyset} =  \Theta(\EXP{\log X}).$$ 
If $u \in C$,  Eq.~\ref{eq:BiBqm7}  states that $f(u)= \Theta(g(\alpha(u),\beta(u))$, for some function $g().$ Since $\beta(u)$ is a function of $\alpha(u)$, this implies that if $u_1,u_2 \in C$ with $u_1.x = u_2.x$, then  $f(u_1) = \Theta(f(u_2)).$   Since $C(\alpha)$ is an $\alpha \times \beta$ rectangle and $H(\alpha)$ is an $\alpha \times (1 - \beta)$ rectangle this yields
$$\mu(H(\alpha))  =  \Theta\left( \frac  {1-\beta}  \beta \mu(C(\alpha) \right)
=   \Theta\left( \frac  1  \beta \cdot \frac 1 n \right)
=\Omega \left( \frac  1  {\sqrt n \log^2 n} \right),
 $$
where the fact that $\beta \le \delta \le \log^2 n / {\sqrt n}$ is explicitly used.
Thus
$$\hat \mu(H(\alpha)) = \frac {\mu(H(\alpha)) }  {1 - \mu(T(\alpha))} = \Omega \left( \frac 1   {\sqrt n \log^2 n} \right).$$
Recall that $X$ was chosen  from a binomial distribution 
$B\left(n, \hat \mu(H(\alpha)\right).$  Using the  Chernoff bounds applied to Binomial random variables,  $X = \Omega\left(n^{1/3}\right)$ with high probability. Since $X \le n,$  this implies $\EXP{\log X} = \Theta(\ln n).$
Thus
$$\EXP{|\MAX(S_n \cap H(\alpha)) \,\Bigm|\,   S_n \cap T(\alpha) = \emptyset} = \Theta(\log n).$$
Then
\begin{eqnarray*}
\EXP{|\MAX(S_n)|}  
			&\ge & \EXP{|\MAX(S_n \cap H(\alpha)) \,\Bigm|\,   S_n \cap T(\alpha) = \emptyset}  \cdot \PR{S_n \cap T(\alpha) = \emptyset}\\
			&=&  \Theta(\log n) \Theta(1) = \Theta(\log n).
\end{eqnarray*}
completing the proof.
\end{proof}

\begin{Lemma}
\label{lem: BiBqUB}
Let $S_n$ be $n$ points chosen from the distribution $\bfd = \Ballpq \infty q$ with $\frac 1 {\sqrt n} \le \delta \le \sqrt n$. Then
$$\EMN = O \left(  \ln n + \sqrt \delta n^{1/4}  \right).$$ 
\end{Lemma}

\begin{proof}
The proof splits into three parts that show
\begin{itemize}
\item[(a)]  $\EXP{|\MAXSN \cap B|} =O\left(  \ln n + \sqrt \delta n^{1/4}  \right).$
\item[(b)]  $\EXP{|\MAXSN \cap C|} =O\left(  \ln n \right).$
\item[(b)]  $\EXP{|\MAXSN \cap A|} =O\left(  \ln n \right).$
\end{itemize}
By symmetry,   $\EXP{|\MAXSN \cap D|} =\EXP{|\MAXSN \cap B|} $.  Combining this with 
(a), (b) and (c)  proves the lemma.

\bigskip

\par\noindent\underline{(a)  $\EXP{|\MAXSN \cap B|}$:}

\begin{figure}[t]
\begin{center}
\includegraphics[height=2in]{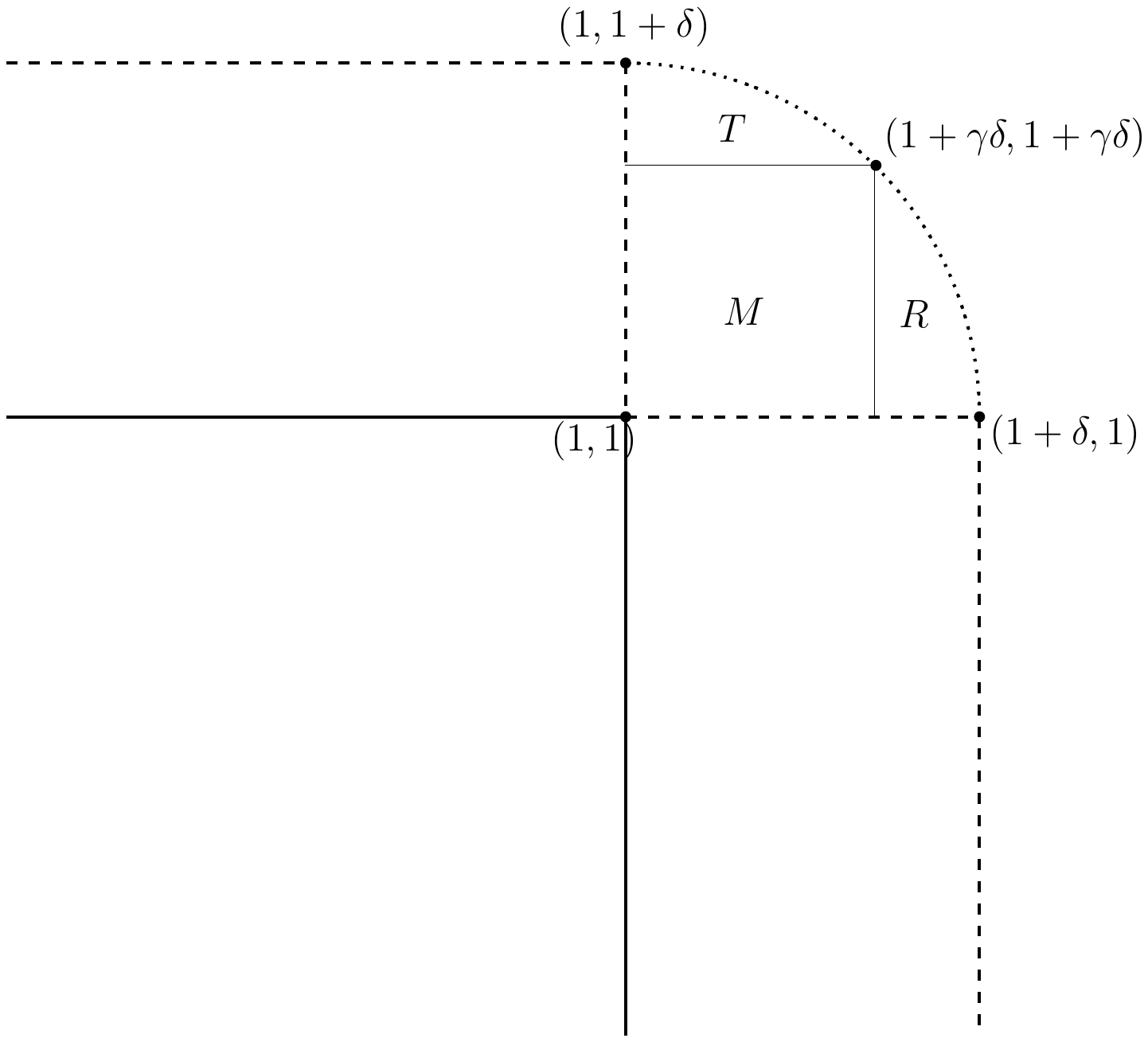}
\end{center}
\caption{ Illustration of  Case (a) of  Lemma \ref{lem: BiBqUB}
that partitions $B$ into
 three regions $M$, $T$ and $R.$ 
The figure illustrates the $q=2$ case.}
\lab{fig:UB-infty_I}
\rule{5.5in}{0.5pt}
\end{figure}

The proof decomposes $B$ into  $M$,  $T,$ and $R$ as illustrated in Figure \ref{fig:UB-infty_I}.  We prove 
$$\EXP{|\MAX(S_n) \cap  (T \cup M)|} = O\left(\sqrt \delta n^{1/4}\right).$$
A symmetrical argument will prove that  
$$\EXP{|\MAX(S_n) \cap  (T \cup R)|} = O\left(\sqrt \delta n^{1/4}\right),$$ and thus
$\EXP{|\MAXSN \cap B|}=O\left(\sqrt \delta n^{1/4}\right)$.

Set $m = \left\lfloor \sqrt  \delta n^{1/4}\right\rfloor$ for an arbitrary small constant $c >0$ and  $\sigma = \frac  {\gamma \delta} m$, where
$\gamma$ is the solution to $\gamma^q + (1 - \gamma)^q = \delta^q.$ We may assume that $ c_1 \frac 1 {\sqrt n} \le \delta  c_2 \le \sqrt n$ for any constance $c_1, c_2 >0$
($\delta$ outside that interval are treated using Lemma \ref {lem: limiting}) so we may assume, for any fixed $c>0$ that 
$ \frac 1 {\sqrt n} \le \delta \le \sqrt n,$  $\sigma \le \min(c', c' \delta).$
  In particular, this implies that 
$\sigma < (1+\gamma) \delta$.  Thus,  $\forall i \in [1,m],$  $x_i \in [1, 1+\gamma']$ where $\gamma' = (1 + \gamma)/2$.
Finally,   note that the definitions implies  $\frac  {\sigma^4} {\delta^2} = \Theta \left(\frac  1 n \right).$

Set $g(x) = 1+\left( \delta^q - (x-1)^q\right)^{1/q}$,  the equation of the upper cap of $B$.

For $ i  \ge 0$ define
 $$x_i  = 1 + i \sigma, $$
\begin{align*}
  y_i 		 &= g(x_i + 2 \sigma),  & p_i         &= (x_i,y_i),          & p_i(t)&= (x_i,y_i -t),\\
                   &                                   & p'_{i} &= (x_{i+1},y_i),    & p'_{i}(t)&= (x_{i+1},y_i -t),\\
\end{align*}
and
\begin{align*}
\Strip_i & = \left\{  u \in B\,:\,   x_i  \le u.x  \le x_{i+1}\right\},  & \Stripp_i& =\Strip_i \setminus P(p_i),\\
 &    &\Stripp_i(t) &= \left\{u \in \Stripp_i \,:\,  u.y \ge y_i -t\right\}.
 \end{align*}

Note that
$$T \cup M = \bigcup_{i=0}^{m-1} \Strip_i
\quad\mbox{and}\quad
\bigcup_{i=0}^{m} \Strip_i \subset C \cap \left\{u \in \Re^2\,:\, 1 \le u.x < (1+\gamma')\delta\right\}.
$$
Thus

\begin{eqnarray*}
\EXP{ \cap (T \cup M|)}&\le & \sum_{i=0}^{m-1}\EXP{|\MAXSN\cap \Stripp_i|} + \sum_{i=0}^{m-1}\EXP{|\MAXSN \cap P(p_i)|}\\
						&\le & \sum_{i=0}^{m-1}\EXP{|\MAXSN\cap \Stripp_i|} + \sum_{i=0}^{m-1}\EXP{|S_n \cap P(p_i)|}.\\
\end{eqnarray*}

 As in the derivations of  Eqs.~\ref {eq:asigma relation} and \ref{eq:gderiv}, we have from the Theorem of the Mean that
$$\forall i,\quad g(x_{i+1}) - g(x_i) = g'(z) (x_{i+1} - x_i)| = g'(z) \sigma,$$
where $z \in [x_i,x_{i+1}].$  
In particular, since $g'(z)$ is bounded for $z \in [ 1, 1 +\delta  \gamma']$, then 
$$\forall i \in [0,m],\quad g(x_{i}) - g(x_i+1) =O(\sigma).$$
In addition, because $g(x)$ is concave in the interval $[1, 1+\delta]$, 
$$g(x_{i+1}) - g(x_i) \le g(x_{i+2}) - g(x_{i+1}).$$
From this and the definitions,  $\forall i$ and $\forall t \ge 0,$ 
\begin{align*}
\alpha(p_i) &= x_{i+2} - x_i = 2 \sigma,  &  \alpha(p'_i) &= x_{i+1} - x_i = \sigma\\
\alpha(p_i)(t) &= \alpha(p'_i)(t) + \sigma,   &  & \Rightarrow \alpha(p_i)(t) = \Theta\left(\alpha(p'_i)(t)\right), \\[0.1in]
\beta(p_i) &= g(x_i) - g(x_{i+2}) = O(\sigma),  &  \beta(p'_i) &= g(x_{i+1}) - g(x_{i+2}) = O(\sigma) \\
\beta(p_i)(t) &= \beta(p'_i)(t) +g(x_i) - g(x_{i+1}),  & & \Rightarrow \beta(p_i)(t) = \Theta\left(\beta(p'_i)(t)\right).  \\
\end{align*}

\begin{figure}[t]
\begin{center}
\includegraphics[height=2.3in]{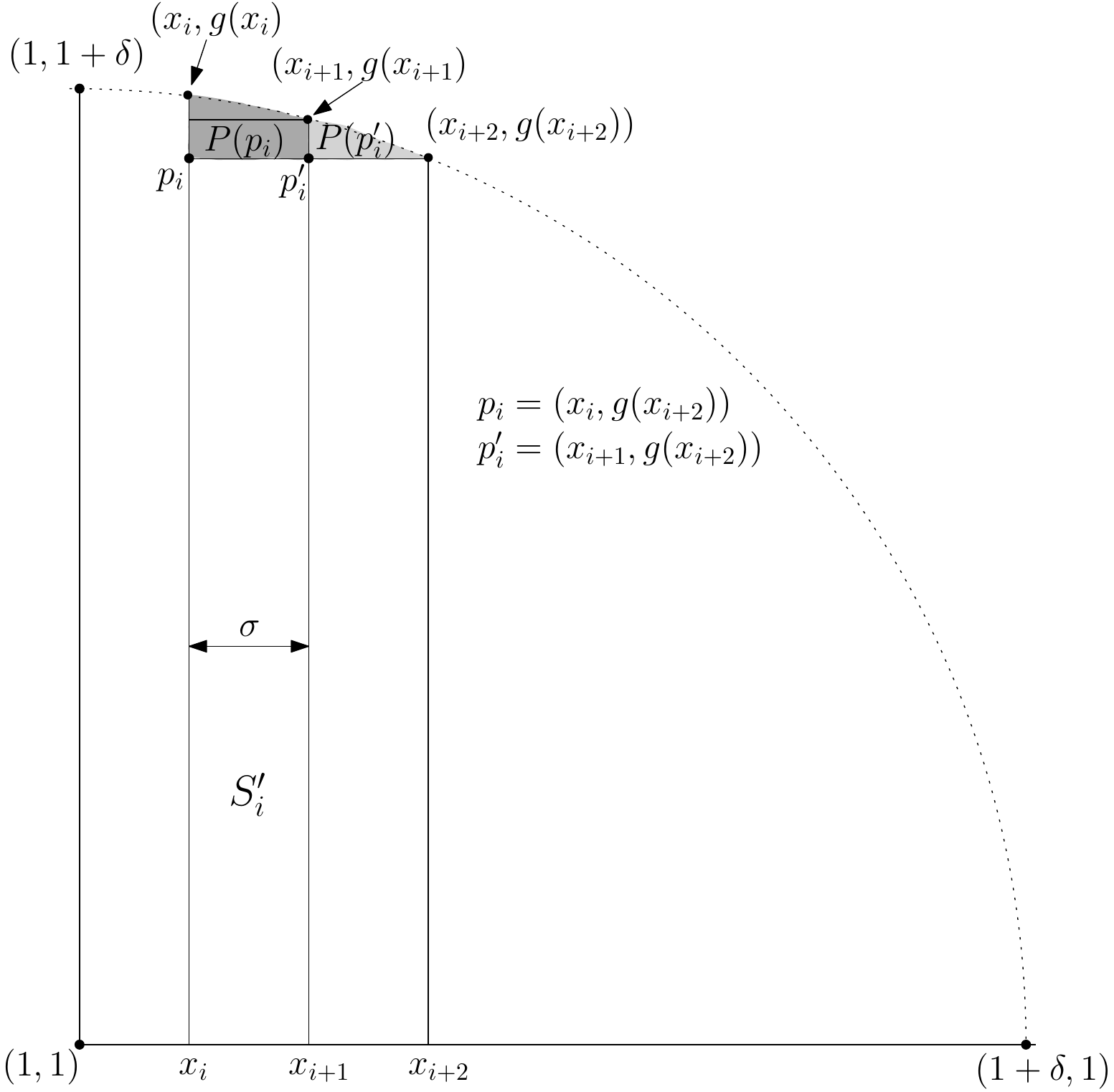}  \hspace*{.1in} \includegraphics[height=2.3in]{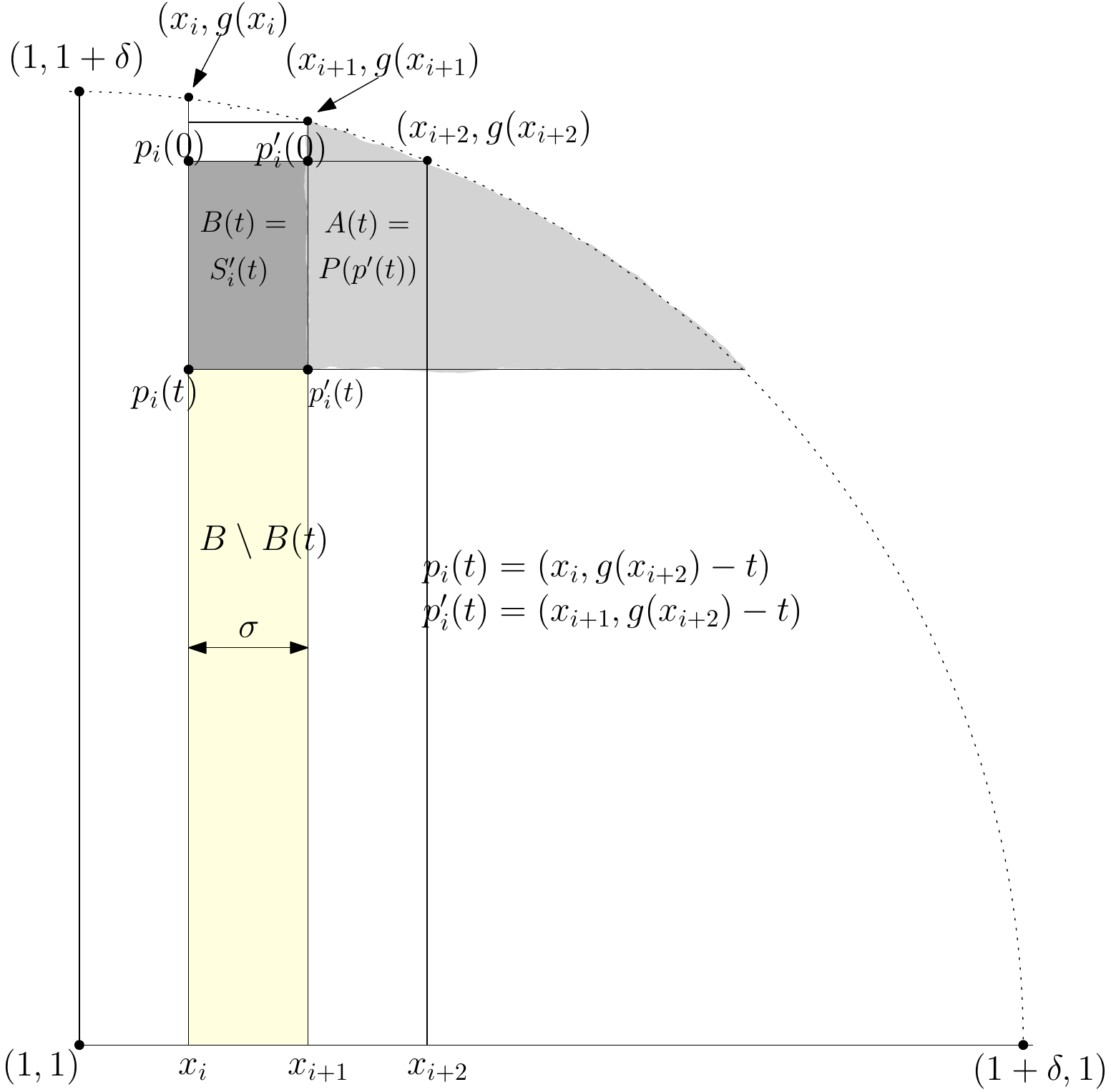}
\end{center}
\caption{ Illustration of  case  (a) of the Upper  Bound.   $g(x) = 1+\left( \delta^q - (x-1)^q\right)^{1/q}$ is the equation of the upper cap. $S'_i$ denotes $\Stripp_i.$
  The figure illustrates the $q=2$ case.}
\lab{fig:UB-infty_II}
\rule{5.5in}{0.5pt}
\end{figure}

Thus from Eqs.~\ref{eq:BiBqm6} and  \ref{eq:BiBqm7} $\forall t > 0,$
\begin{eqnarray*}
f(P(p_i(t)) &=&\Theta\left( \frac    {  \min\Bigl(\alpha(P(p_i(t)),\, 1\Bigr) \cdot  \min\Bigl(\beta(P(p_i(t)),\, 1\Bigr)} {\delta^2}\right)\\
                &=&\Theta\left( \frac    {  \min\Bigl(\alpha(P(p'_i(t)),\, 1\Bigr) \cdot  \min\Bigl(\beta(P(p'_i(t)),\, 1\Bigr)} {\delta^2}\right)\\
                &=& f(P(p'_i(t)).
\end{eqnarray*}
and
\begin{equation}
\label{eq:BiBqUB6}
\mu(P(p_i(t))) = \Theta\mu(P(p'_i(t))).
\end{equation}
Furthermore,
since $\alpha(p_i) = 2 \sigma$ and $\beta(p_i) = O(\sigma)$,
from Eq.~\ref{eq:BiBqm5}
$$ 
 \mu(P(p_i)) = \Theta\left(\frac  {\alpha^2 (v) \beta^2(v)} {\delta^2} \right) = O\left( \frac {\sigma^4} {\delta^2} \right)  
= O\left( \frac  1 n \right).
$$
and thus
\begin{equation}
\label{eq:BiBqUP3}
\sum_{i=0}^{m-1}\EXP{|S_n \cap P(p_i)|} = \sum_{i=0}^{m-1} n \mu(P(p_i)) = \sum_{i=0}^{m-1} 1 = O(m).
\end{equation}
Now, in the parameters of   the Sweep Lemma  set 
\begin{align*}
A &= P(p'_i)(y_i) = P\bigl((x_i,1)\bigr),  &  B &= \Strip'_i,\\
A(t) &= P(p'_i)(t),  &  B &= \Strip'_i(t).\\
\end{align*}
\medskip
Every point in $B\setminus B(t)$ is dominated by every point in $A(t).$  Since
$$
\Stripp_i(t) \subset P(p_i(t)),
$$
Eq.~\ref{eq:BiBqUB6}  implies
$$
\mu(B(t)) =  \mu(\Stripp_i(t)) \le \mu(P(p_i(t))) = \Theta\left(\mu(P(p_i(t)))\right) = \Theta(\mu(A(t))).$$
The sweep lemma then shows that, $ \forall i \in [0,m-1]$,
$$EXP{|\MAXSN\cap \Stripp_i|} = O(1).$$
Thus
$$\EXP{ \MAXSN \cap (T \cup M|)}\le  \sum_{i=0}^{m-1}\EXP{|\MAXSN\cap \Stripp_i|} + \sum_{i=0}^{m-1}\EXP{|S_n \cap P(p_i)|} = O(m),$$
and the proof of this section is complete.

\medskip

\par\noindent\underline{(b)  $\EXP{|\MAXSN \cap C|}$:}
 Consider the regions   $B(\alpha),$  $C(\alpha),$ $H(\alpha),$ defined in 
Lemma \ref{lem:LB-infty-BQ-side}. Set
$$V(\alpha) = C \setminus\Bigl( (C\alpha) \cup H(\alpha)\Bigr).$$
Note that if $\delta >1$ then, for large enough values of $\alpha$, $H(\alpha) = V(\alpha) = \emptyset.$

By the decomposition, for every $ \alpha$,
$$| \MAXSN \cap C| = | \MAXSN \cap C(\alpha)|  +| \MAXSN \cap H(\alpha)|  +| \MAXSN \cap V(\alpha)|.$$
Now define the random variable
$$\bar \alpha =
\left\{
\begin{array}{ll}
\max\Bigl(  1+\delta - u.x \,:\, u \in \Bigr)   & \mbox{ if $S_n \cap B \not = \emptyset$},\\
\delta & \mbox{ if $S_n \cap B  = \emptyset$.}
\end{array}
\right.
$$
From the definitions, $\forall \alpha\le \delta,$  any point in $B(\alpha)$ dominates any point in $V(\alpha)$.
If   $S_n \cap B(\alpha)  \not = \emptyset$,  then such a point exists and thus 
$|\MAX(S_n \cap V(\bar \alpha))| = 0.$   If not, then $V(\alpha) = \emptyset$, so trivially $|\MAX(S_n \cap V(\bar \alpha))| = 0.$  Thus, in all cases
$$|\MAX(S_n \cap V(\bar \alpha))| = 0.$$

Using a similar argument to that in part (a) of Lemma \ref{lem: BiBqLB} (the upper bound) if we let $X(\alpha)  = |S_n \cap H(\alpha)|$, then, conditioning on  {$\bar \alpha = \alpha$} 
$$\EXP{|\MAX(S_n \cap H(\bar \alpha) \,\Bigm| \bar \alpha = \alpha}=
\EXP {H_{X_{\bar \alpha}}  \,\Bigm| \bar \alpha = \alpha}=
\Theta\left( \EXP { \ln (X(\alpha))  \,\Bigm| \bar \alpha = \alpha}  \right).
$$
But, since $X(\alpha) \le n$ this immediately yields
$$\EXP{|\MAX(S_n \cap H(\bar \alpha) }= O(\ln n).$$

The last piece is to apply the Sweep Lemma to bound   $\EXP{|   \MAXSN \cap C(\alpha)|}.$
Set
\begin{align*}
\bar A &=   B = B(\delta),       &  \bar B &= C = C(\delta),\\
\bar A(\alpha) & = B(\alpha),  &  \bar B(\alpha) &= C(\alpha).
\end{align*}
By definition every point in $\bar A(\alpha)$ dominates every point in $\bar B \setminus \bar B(\alpha)$ and
from  Lemma \ref {lem:LB-infty-BQ-side}, $\forall \alpha \le \delta,$   $\mu(\bar B(\alpha)) = O\mu(\bar A(\alpha)).$ 
Applying the Sweep  Lemma with $\bar A(t)$,  $\bar B (t)$ and $\alpha$ as the sweep parameter yields
$$\EXP{| \MAXSN \cap C|} = O(1),$$ so
$$\EXP{| \MAXSN \cap C(\bar \alpha)|} \le \EXP{| \MAXSN \cap C|} = O(1).$$

Combining the items proved above yields
\begin{eqnarray*}
\EXP{| \MAXSN \cap C| } &=&   \EXP{ | \MAXSN \cap C(\bar \alpha)|} \\ 
					&=& O(1) + O(\ln n) = O(\ln n).
\end{eqnarray*}
\medskip

\par\noindent\underline{(c)  $\EXP{|\MAXSN \cap A|}$:}

\begin{figure}[t]
\begin{center}
\includegraphics[height=3in]{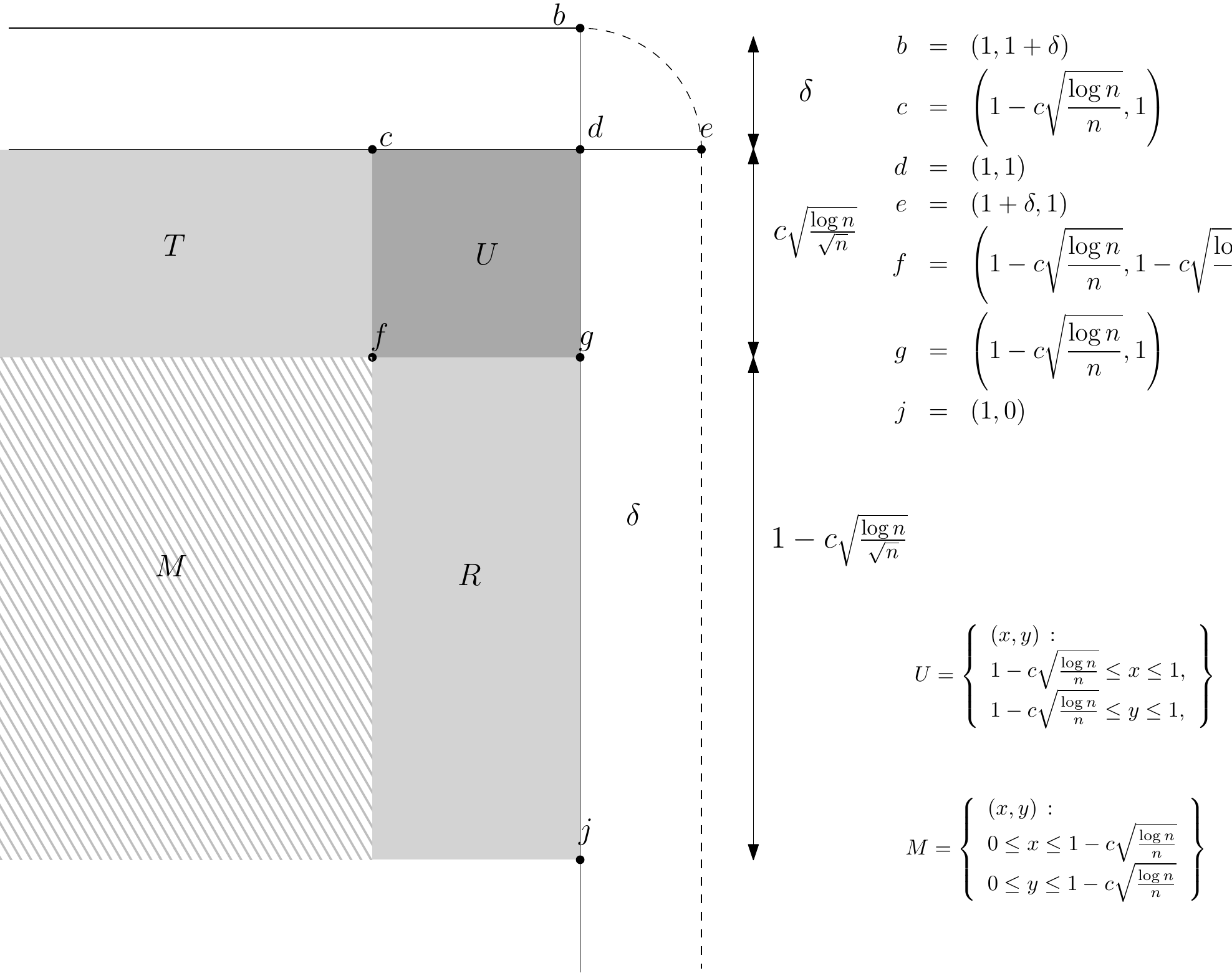}
\end{center}
\caption{ Illustration of Upper Bound case (c).  
}
\lab{fig:UB-infty-IX}
\rule{5.5in}{0.5pt}
\end{figure}
\medskip

Let $d= (1,1)$.  Then  $B = P(d).$ See Fig.~\ref {fig:UB-infty-IX}.  Note that $\alpha(d) = \beta(d) = \delta$.

If  $\delta > 1$ then Eqs.~\ref {eq:BiBqm5}  and \ref {eq:BiBqm6} give
$\mu(B) = \Theta(1).$
If $ c' \sqrt \frac {\log n} n \le \delta  \le 1$  for large enough $c'$, then Eqs.~\ref {eq:BiBqm5}  and \ref {eq:BiBqm6} give
 $\mu (B) > 3 \frac {\log n} {n}.$
So,  in both of these cases, for large enough $n,$   $\mu (B) > 3 \frac {\ln n} {n}.$
Thus,  from Lemma \ref{lem: basic mu}, 
$$\Pr(B \cap S_n = \emptyset) =   \left(1 - \mu(B)\right)^n \le    \frac 1 n.$$
Since every point in $B$ dominates every point in $A,$ if $\Pr(B \cap S_n = \emptyset),$ then  
$\MAX(S_n) \cap A = \emptyset.$  Thus, if  $ c \sqrt \frac {\log n} n \le \delta,$
$$\EXP{|\MAXSN \cap B|} \le  n Pr(B \cap S_n = \emptyset) = O(1).$$

We therefore now assume that $ \frac 1 {\sqrt n} \le \delta \le c \sqrt \frac {\log n} n.$  Eq.~\ref{eq:BiBqm4} implies that, if $V \subseteq A,$ then $\mu(V) = \Theta(\Area(V)).$

Partition $A$ into four subregions,  $U,T,R,M$  as illustrated in Figure \ref{fig:UB-infty-IX} where $U$ is a square with side-length  
$c \sqrt {\frac { \ln n} n}$ for some $c>0$.

Then 
$$ \Area(U) = \frac  {c^2 \log n} n  \quad \Rightarrow  \quad \mu(U) = \Theta\left( \frac  {c^2 \log n} n \right).
$$
Choose $c> c'$ large enough so that   $\mu(U) \ge \frac {3 \ln n} n.$  
Since every point in $U$ dominates every point in $M$ exactly the same type of analysis as performed above shows that
$$\EXP{|\MAXSN \cap M|} \le  n \Pr(U \cap S_n = \emptyset) = O(1).$$
Also
$$\EXP{|\MAXSN \cap U|} \le \EXP{|S_n \cap U|}  = n \mu(U) = O(\log n).$$
From the symmetry between $T$ and $R$ this yields that
\begin{equation}
\label{eq:BqBiUP8}
\EXP{|\MAXSN \cap A|} \le O(1)  + O(\log n)  + 2 \EXP{|\MAXSN \cap R|}.
\end{equation}

Because the side-length of $U$ is greater than $\delta$ we find that for all $u_1,u_2 \in R$ with $u_1.x = u_2.x,$  
$$\Area\left( B_q(u_1,\delta) \cap D\right) = \Area\left( B_q(u_2,\delta) \cap D\right),$$
and thus
$$f(u_1) = f(u_2).$$
Using a very similar analysis to that performed in part (a) this shows that
$$\EXP{|MAX(S_n \cap R)|} = \EXP{ H_X} = \Theta\left(\EXP{ \ln X}\right) = \Theta(\ln n),$$
where $X= |S_n \cap R|$ and the last equality come from the fact that $X \le n.$
Since 
$$\EXP{|MAX(S_n) \cap R|}   \le \EXP{|MAX(S_n \cap R)|},$$
combining this with
Eq.~\ref{eq:BqBiUP8}
proves
$$\EXP{|MAX(S_n \cap R)|} = O(\ln n).$$
\end{proof}

\bigskip
\begin{proof} {\bf of Lemma \ref{lem:BiBq measure}.}
\begin{figure}[t]
\begin{center}
\includegraphics[width=11cm]{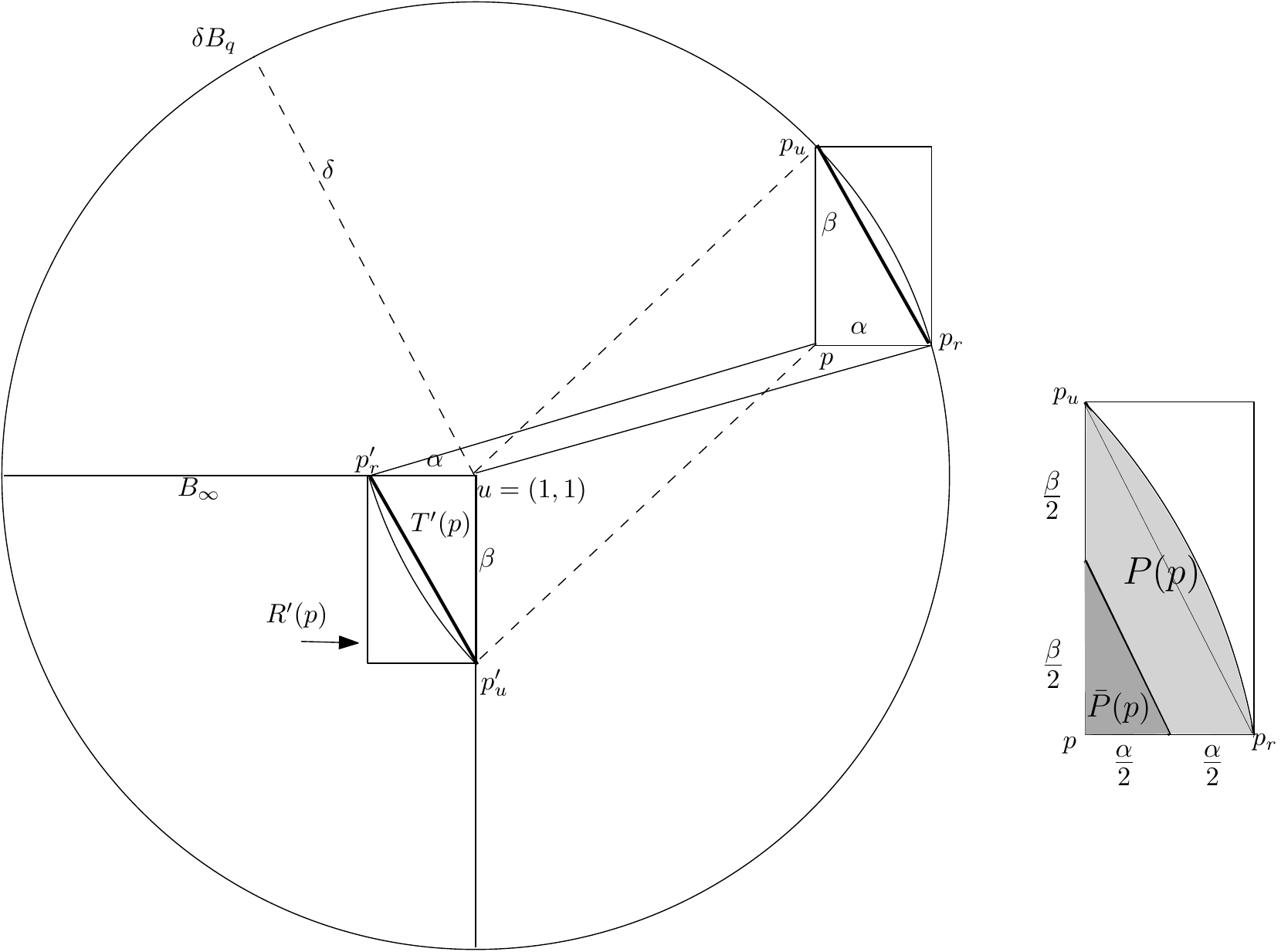} 
\end{center}
\caption{Illustration of the proof of Lemma  \ref{lem:BiBq measure},  Eqs.~\ref{eq:BiBqm5} and  \ref{eq:BiBqm6} when $p \in B.$  The right hand side shows a blown-up version of $P(p)$ with $\bar P(p)$ sitting inside of it.} 
\lab{fig:BiBqf}
\rule{5.5in}{0.5pt}
\end{figure}

Recall from Lemma \ref {lem: measure integral}  that
$
f(v) =
\frac {\Area(P'(v))} {\delta^2}.$
In what follows we use $\alpha, \beta$ to denote $\alpha(v)$ and $\beta(v).$
\medskip

\par\noindent\underline{$p \in A:$}\\
 If $\delta \le 1$ then the result follows directly from Lemma \ref{lem:easy mu}(b).

If $\delta >1$  and $v \in A$ then $\Area(P(v))= \Theta(1)$ and the result follows from
plugging into Lemma \ref {lem: measure integral}.

\medskip
\par\noindent\underline{$p \in B:$}\\

Define 
$$U = \{u \in \Re^2\,:\, u.x \le 1,\, u.y \le 1\}.$$
The main observation is that, from the definition of the $L_q$ metric, 
\begin{eqnarray*}
\hspace*{-.4in}
u \in P(p)  &\Leftrightarrow&  u.x \ge p.x,\, u.y \ge p.y   \quad \mbox {and } \quad   \left|u.x - 1\right|^q +  \left|u.y - 1\right|^q \le \delta^q\\
               &\Leftrightarrow&   u.x \ge p.x,\, u.y \ge p.y  \\
               && \quad  \mbox {and } \quad    \left| \left( 1 - (u.x-p.x)\right) - p.x\right|^q   +  \left| \left( 1 - (u.y-upy)\right) - p.y\right|^q\le \delta^q\\
                              &\quad   \Leftrightarrow   \quad&  u.x \ge p.x,\, u.y \ge p.y  \quad \mbox {and } \quad  \Bigl(1 - (u.x-p.x),\,  1 - (u.y-p.y)  \Bigr)  \in B_q(p,\delta)\\
                              &\quad   \Leftrightarrow   \quad& v \in B_q(p,\delta) \cup  U \ \mbox{where} \quad v.x = 1 - (u.x - p.x),\, v.y = 1 - (u.y - p.y).
\end{eqnarray*}
Physically, this mean that $B_q(p,\delta) \cap U $ is the mirror image of $P(p)$ flipped across the line $x= -y$ and then placed with the image of $p$ at $(1,1).$  
See Fig.~\ref{fig:BiBqf}.

Let $T'(p)$ be the triangle with with corners  $(1-\alpha,0)$, $(0,1 - \beta)$ and $(1,1)$ and $R(p)$ the rectangle with 
 lower left corner 
$(1-\alpha, 1 - \beta)$ and upper-right corner $(1,1).$   The convexity of $ B_q(p,\delta) \cap U$ implies
$$ T'(p) \subseteq  B_q(p,\delta) \cap U \subseteq R'(p), $$
therefore
$$\frac 1 2 \alpha \beta  = \Area(T'(p))  \le  \Area\left( B_q(p,\delta) \cap U  \right) \le \Area(R'(p)) \le \alpha \beta.$$
Note that $B_\infty$ is the upper right hand corner of   $U.$

If $\alpha, \beta \le 1$ then $R'(p) \subset B_\infty$, therefore  
$$P'(p) = B_q(p,\delta) \cap  B_\infty = B_q(p,\delta) \cap U,$$ yielding
$\Area(P'(p)) = \Theta(\alpha\beta).$

If $ \alpha, \beta \ge 1$, then $\Area\Bigl(T(p) \cap B _1\Bigr) \ge 1/2$ so $\Area(P'(p)) = \Theta(1).$

If $ \alpha \ge 1$ and  $\beta \le 1$, then $\Area\Bigl(T(p) \cap B _1\Bigr) = \Theta(\beta/2).$

f $ \alpha \le 1$ and  $\beta \ge 1$, then $\Area\Bigl(T(p) \cap B _1\Bigr) = \Theta(\alpha/2).$

Combining the four cases  proves Eq.~\ref{eq:BiBqm5}.

\medskip

To prove Eq.~\ref{eq:BiBqm6} recall that $\Area(P(v)) = \Theta(\alpha \beta)$ and 
$$\mu(P(v)) = \int_{u \in P(v)} f(u) du.$$
Also note that $\forall u \in P(v),$  $\alpha(u) \le \alpha(v)= \alpha$ and $\beta(u) \le \beta(v)=\beta.$   Eq.~\ref{eq:BiBqm5} then implies that  $\forall u \in P(u),$  $f(u) = O(f(v)).$  Thus
$$\mu(P(v)) = \int_{u \in P(v)} f(u) du \le \Area(P(v)) \cdot \max_{u \in P(v)} f(u) = O (\alpha \beta f(v)).$$
For the other direction 
define $\bar P(p)$ to be the upper right triangle with lower left corner $p$ and horizontal, vertical side lengths $\alpha/2,\beta/2$. 
See  the right hand side of Fig.~\ref{fig:BiBqf}.     Note that 
$\forall u \in \bar P(v),$  $\alpha(u) \ge \alpha(v)/2= \alpha/2$ and $\beta(u) \ge \beta(v)/2=\beta/2.$
 Thus, $\forall u \in \bar P(u),$  $f(u) = \Omega(f(v)).$  
Finish by noting that because $\bar P(p) \subset P(p)$, 
\begin{eqnarray*}
\mu(P(v)) = &\int_{u \in P(v)} f(u) du \ge \int_{u \in \bar P(v)} f(u) du \\
\ge&
\Area(\bar P(v)) \cdot \min_{u \in \bar P(v)} f(u) = \Omega (\alpha \beta f(v)).
\end{eqnarray*}
We have thus shown  $\mu(P(v)) = \Theta(\alpha \beta f(v)).$

\medskip
\par\noindent\underline{$p \in C:$}\\
By straightforward geometric arguments  (see Fig.~\ref{fig:BiBqm3})
$$\Area(P'(v))  = \Theta\left(
{  \min\Bigl(\alpha(v),\, 1\Bigr) \cdot  \min\Bigl(\beta(v),\, 1\Bigr)} 
\right).$$

 \begin{figure}[t]  
\begin{center}
\includegraphics[width=8cm]{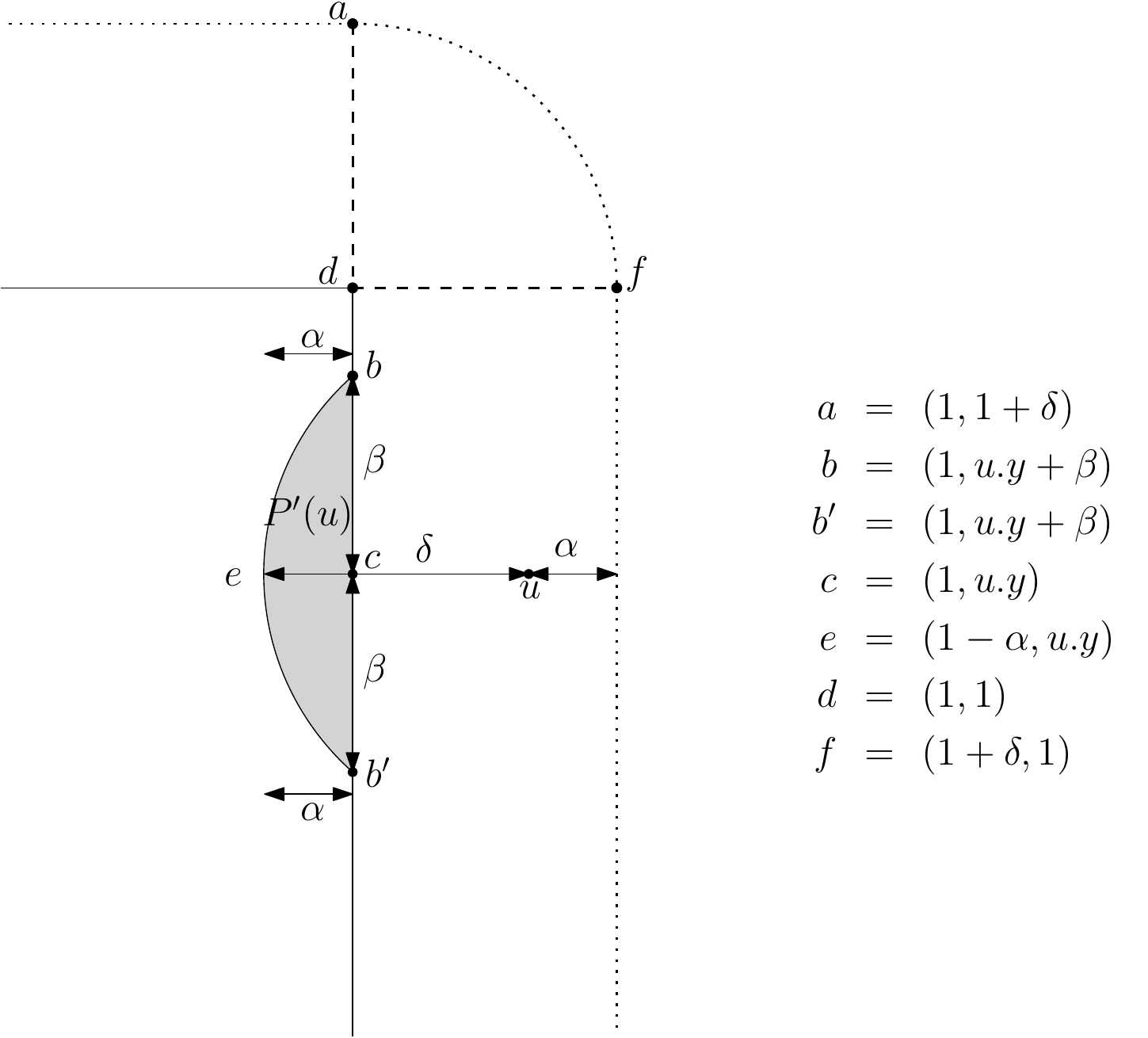} 
\end{center}
\caption{Illustration of the proof of Lemma  \ref{lem:BiBq measure} Eq.~\ref{eq:BiBqm7} when $p \in C.$.} 
\label{fig:BiBqm3}
\rule{5.5in}{0.5pt}
\end{figure}

\end{proof}
\section{Conclusion}
\label{sec: Conclusion}
This paper developed a suite of tools for deriving $\EMN,$  the expected number of maximal points in a set of $n$ points chosen IID from  $\Ballpq p q$,  which is the convolution of two distributions;  the first is  $\Ball p$, the  uniform distribution over the $L_p$ ball  and the second $\delta \Ball q$, the uniform distribution over a $\delta$-scaled $L_q$ ball.  For small $\delta$, $\Ballpq p q$ could be considered as a {\em smoothed} version of $\Ball p$ with $\delta \Ball q$ error.  This result seems to be the first analysis of $\EMN$ for non-uniform and non-Gaussian distributions.

This paper is only a  first step.  Obvious next steps are

\begin{itemize}
\item The results in the paper were only proven for $p,q \in \{1,2,\infty\}$  and  $p=\infty, q \in [1,\infty.]$ 
The next step would be to attempt to extend the results to {\em all} pairs $p,q, \in [1,\infty]$.

\item The results in this paper only derive first-order asymptotics.   Another obvious direction would be to try to prove limit-theorems, e.g., paralleling the results in \cite{bai2001limit} for 2-dimensional uniform samples.

\item There is a rich literature stretching back more than fifty years on the average number of points on the {\em convex hull} of points chosen IID from a uniform distribution in a planar region or a Gaussian distribution, e.g.,  \cite{Dwyer1990,renyi1963konvexe}. It would be interesting to see how the convex hull evolves in these convoluted distributions.

\item  Finally, we note that the results on $\EMN$ for $n$ points chosen IID from a uniform distribution over an $L_p$ ball have analogues in higher dimensions, i.e.,  $\Theta\left( \log^{d-1} n \right)$ if $p = \infty$ and
$\Theta\left( n^{1-\frac 1 d} n \right)$ if $p\in [1,\infty)$ \cite{Dwyer1990,baryshnikov2007expected}.  The next step would be to attempt to extend the results in this paper to higher distributions. 

\end{itemize}

\bibliographystyle{plain}
\bibliography{smoothed.bib}
\end{document}